\documentclass{amsart}
\usepackage{amsmath}
\usepackage{amssymb}
\usepackage{color}
\usepackage{url} 
\usepackage{amsthm}
\usepackage{bm}
\usepackage{xy}
\usepackage{enumitem}
\usepackage[ruled,vlined]{algorithm2e}
\usepackage{mathrsfs}
\usepackage[gen]{eurosym}
\usepackage{graphicx}

\usepackage[colorinlistoftodos]{todonotes}

\theoremstyle{plain}
\newtheorem{theorem}{Theorem}[section]

\newtheorem{lemma}[theorem]{Lemma}
\newtheorem{notation}[theorem]{Notation}
\newtheorem{problem}[theorem]{Problem}
\newtheorem{proposition}[theorem]{Proposition}

\newtheorem{definition}[theorem]{Definition}

\newtheorem{assumption}[theorem]{Assumption}
\theoremstyle{remark}
\newtheorem{remark}[theorem]{Remark}
\newtheorem{example}[theorem]{Example}

\numberwithin{equation}{section}

\newcommand\myprime{\mkern-3.5mu\raise0.6ex\hbox{$\scriptstyle\prime$}}

\begin{document}
\title[Functional portfolio optimization in SPT]{Functional portfolio optimization in stochastic portfolio theory}

\begin{abstract}
In this paper we develop a concrete and fully implementable approach to the optimization of functionally generated portfolios in stochastic portfolio theory. The main idea is to optimize over a family of rank-based portfolios parameterized by an exponentially concave function on the unit interval. This choice can be motivated by the long term stability of the capital distribution observed in large equity markets, and allows us to circumvent  the curse of dimensionality. The resulting optimization problem, which is convex, allows for various regularizations and constraints to be imposed on the generating function. We prove an existence and uniqueness result for our optimization problem and provide a stability estimate in terms of a Wasserstein metric of the input measure. Then, we formulate a discretization which can be implemented numerically using available software packages and analyze its approximation error. Finally, we present empirical examples using CRSP data from the US stock market, including the performance of the portfolios allowing for dividends, defaults, and transaction costs.

\end{abstract}

\keywords{stochastic portfolio theory, portfolio optimization, functionally generated portfolio, capital distribution, convex optimization, exponentially concave function, Wasserstein metric}


\author{Steven Campbell}
\address{Department of Statistical Sciences, University of Toronto, Toronto, Ontario, Canada.}
\email{steven.campbell@mail.utoronto.ca}

\author{Ting-Kam Leonard Wong}
\address{Department of Statistical Sciences, University of Toronto, Toronto, Ontario, Canada.}
\email{tkl.wong@utoronto.ca}

\date{\today}

\maketitle %

\section{Introduction} \label{sec:intro}
Stochastic portfolio theory (SPT) is a mathematical framework introduced by Robert Fernholz \cite{F02} for analyzing the behaviors of equity markets and portfolio selection; also see \cite{FK09} for a  more recent introduction. The theory identifies macroscopic properties of large equity markets, including the mean-reverting behavior of market diversity and the persistence of market volatility, that can be exploited by carefully constructed portfolios to outperform the market. These portfolios are called relative arbitrages. To exploit these properties, Fernholz \cite{F99} constructed a family of portfolios known as {\it functionally generated portfolios}. Following the treatment of \cite{PW16}, a functionally generated portfolio is specified by a generating function $\varphi$ on the unit simplex which is typically exponentially concave (i.e., $e^{\varphi}$ is concave), and its weights are deterministic functions of the current market weights given in terms of the derivatives of the generating function. The relative (log) value process of a functionally generated portfolio satisfies, in both discrete and continuous time, a remarkable pathwise decomposition (see \eqref{eqn:pathwise.decomposition}), where the first term reflects the change of market diversity and the second term measures accumulated market volatility. Assuming that the market is diverse and sufficiently volatile, the first term is bounded and the second term is strictly increasing, resulting in a (long term) relative arbitrage with respect to the market. Further properties and generalizations of functionally generated portfolios, as well as constructions of short and long term relative arbitrages under appropriate conditions, have since been studied by many authors; see for example \cite{FK05, KK20, KR17, M21, P19, RX19, W19} and the references therein. 

A natural question arising from the discovery of functional portfolio generation is the ``optimal'' choice of the generating function. Several approaches have been considered in the literature. Fernholz studied in his book \cite[Section 6.3]{F02} the one-parameter family of diversity-weighted portfolios from the viewpoint of turnover and trading costs. In general, optimization of functionally generated portfolios is nonparametric as the space of generating functions is an infinite dimensional function space. In \cite{W15} the second author introduced an optimization problem motivated by (shape-constrained) nonparameteric density estimation and studied its theoretical properties, but did not obtain practical algorithms in multi-dimensions. A Bayesian version of this problem is mathematically equivalent to Cover's universal portfolio \cite{C91} and was studied in \cite{CSW19,W15b} using discrete and continuous time set-ups. See \cite{ACLP21} for another generalization using rough paths. In this approach a major challenge is to construct suitable prior distributions on the space of generating functions. This was partially addressed by \cite{BW19} but a practical implementation is still open. In \cite{SV16} optimization of portfolio maps was studied using Gaussian process priors; we note that the resulting portfolios are typically not functionally generated. Rather than optimizing over the generating function, one can also study construction of optimal relative arbitrage over a given horizon, as well as robust optimization of asymptotic growth under  suitable conditions on the market model. The former problem was pioneered by \cite{FK10, FK11}; also see \cite{BHS12, CT15, IY20, R13} and the references therein. Robust optimization of the asymptotic growth rate was introduced in \cite{KR12} and further results can be found in \cite{IL20,KR18}. In particular, the authors of \cite{IL20} considered an optimization over portfolios generated by exponentially concave functions and obtained, for a given covariance structure and an invariant density, existence and uniqueness results for the optimizer of the asymptotic growth rate. Here we do not assume such exact information is available.

In this paper we introduce a portfolio optimization problem which is motivated by regularized empirical risk minimization in machine learning theory; here, risk minimization is replaced by return maximization.  Our main aim is to develop a concrete optimization problem for functionally generated portfolios that does not rely on specific probabilistic assumptions on the market model and can be implemented numerically. Additionally, a secondary goal is to improve the finite sample performance of these portfolios. The key idea is to restrict to generating functions of the form $\varphi(\mathbf{p}) = \frac{1}{n} \sum_{i = 1}^n \ell(p_i)$, where $\ell$ is exponentially concave (i.e., $e^{\ell}$ is concave) on the unit interval. The symmetry of $\varphi$, with respect to relabeling of coordinates, implies that the induced portfolio map is {\it rank-based}. This means that the optimization only depends on the rank-based properties of the data which are more stable than the name-based market weights; see Section \ref{sec:rank.based.stability} for some empirical evidence in terms of a Wasserstein metric. Also see Remark \ref{rmk:log.optimal} for another theoretical justification concerned with log-optimality of these portfolios. Additional supporting evidence for this phenomenon can be found in the recent paper \cite{banner2018diversification} which documents the stable relationship between rank and volatility. The specific functional form of $\varphi$ serves as a dimension reduction which allows us to circumvent the curse of dimensionality as the number of stocks becomes larger and larger, and makes numerical implementation tractable. 
For technical purposes, we consider $\ell$ in a space $\mathcal{E}_{\beta}$ of $\beta$-smooth functions where $\beta > 0$ is a fixed but arbitrary constant. If we optimize using historical data,  given by a sequence of market weights over a training period $[0, t]$, the basic version of our optimization problem (Definition \ref{prob:main.optimization}), which is convex, has the form
\begin{equation} \label{eqn:optimization.problem}
\sup_{\ell \in \mathcal{E}_{\beta}} \left( \frac{1}{t} \log V_{\ell}(t) - \lambda R(\ell) \right),
\end{equation}
where $\frac{1}{t} \log V_{\ell}(t)$ is the relative logarithmic growth rate of the portfolio induced by $\ell$, and $\lambda R(\ell)$ is a regularization term which is convex in $\ell$. More generally, we may replace the historical data by a probability measure over a suitable state space $\Delta_{n, \geq} \times \Delta_n$, where $\Delta_{n, \geq}$ is the ordered unit simplex (see \eqref{eqn:ordered.simplex}), and give different weights to the diversity and volatility components in the pathwise decomposition. One use of the regularization term is to penalize deviation from the market portfolio (where $\ell$ is constant), hence controlling indirectly the turnover. We show that this optimization problem has desirable theoretical properties and can be implemented, after a discretization, using available tools of convex optimization. The discretization error can be estimated explicitly, and our analysis, which involves approximation of univariate exponentially concave functions, may be of independent interest.

The rest of the paper is organized as follows. Section \ref{sec:prelim} reviews the discrete time set-up of SPT  and studies ranked-based portfolios induced by exponentially concave functions $\ell$ in the space $\mathcal{E}_{\beta}$. In Section \ref{sec:optimization.problem} we formulate the portfolio optimization problem and establish its theoretical properties including a stability estimate in terms of a Wasserstein distance on the data space. Discretization and algorithmic considerations are given in Section \ref{sec:algorithm}. Our main theoretical results are Theorems \ref{thm:existence}, \ref{thm:continuity.Wasserstein} and \ref{thm:consistency}. In brief, these theorems say that our problem admits an optimal solution which is unique in an appropriate sense, is stable in the input measure, and is well approximated by a discretization. Problem \ref{prob:discretized} is a discretized formulation which is a numerically tractable version of Problem \ref{prob:main.optimization}. The two problems are related by the aforementioned Theorem \ref{thm:consistency}. In Section \ref{sec:empirical} we present a careful empirical illustration using data from the US stock market. Section \ref{sec:conclusion} concludes the paper and points out several directions for further study. Some technical proofs and lemmas are gathered in the Appendix.

\section{Rank-based functionally generated portfolios} \label{sec:prelim}

\subsection{Market weight and relative value} \label{sec:market}
We work with the discrete time set-up of stochastic portfolio theory adopted in \cite{W15, PW16, PW13}. Let $n \geq 2$ be the number of stocks in the market and consider the open unit simplex
\[
\Delta_n = \{ \mathbf{p} = (p_1, \ldots, p_n) \in (0, 1)^n : p_1 + \cdots + p_n = 1 \}.
\]
By a \textit{market sequence} we mean a sequence $\{\boldsymbol{\mu}(t)\}_{t = 0}^{\infty}$ with values in $\Delta_n$. We interpret $\boldsymbol{\mu}(t)$ as the vector of market weights at time $t$, i.e., $\mu_i(t) = \frac{X_i(t)}{X_1(t) + \cdots + X_n(t)}$ where $X_i(t) > 0$ is the market capitalization of stock $i$ at time $t$. Note that we make no probabilistic assumptions about the market weight process. We call this set-up a {\it closed market} as opposed to the more realistic {\it open market} (see \cite{karatzas2020open}) which is used in the empirical demonstration in Section \ref{sec:empirical}.

We consider self-financing trading strategies that do not involve short selling. 
A \textit{portfolio vector} is an element of the closed simplex $\overline{\Delta}_n$. Given a sequence $\{\boldsymbol{\pi}(t)\}_{t = 0}^{\infty}$ of portfolio vectors, we define the relative value process by $V(0) = 1$ and
\begin{equation} \label{eqn:portfolio.relative.value}
V(t) = \prod_{s = 0}^{t - 1} \left( \sum_{i = 1}^n \pi_i(s) \frac{\mu_i(s + 1)}{\mu_i(s)} \right).
\end{equation}
While \eqref{eqn:portfolio.relative.value} assumes implicitly that there is no transaction cost, it will be included in the simulation in Section \ref{sec:empirical}. If $\boldsymbol{\pi}(t)$ has the form $\boldsymbol{\pi}(\mu(t))$ for some deterministic function $\boldsymbol{\pi} : \Delta_n \rightarrow \overline{\Delta}_n$, we say that $\boldsymbol{\pi}$ is induced by the {\it portfolio map} $\boldsymbol{\pi}$.

Of special interest is the class of functionally generated portfolio maps. (To be precise, in this paper we study {\it multiplicatively} generated portfolios according to the terminology of Karatzas and Ruf \cite{KR17}. See the papers cited in the Introduction for this and other notions of functional portfolio generation.)  By definition, we say that a real-valued function $\varphi$ on a convex set is \textit{exponentially concave} if $e^{\varphi}$ is concave.  Given a differentiable exponentially concave function $\varphi$ on $\Delta_n$, we define a portfolio map $\boldsymbol{\pi} : \Delta_n \rightarrow \overline{\Delta}_n$ by
\begin{equation} \label{eqn:fgp}
\boldsymbol{\pi}_i(\mathbf{p}) = p_i \left[1 + \frac{\partial \varphi}{\partial p_i} - \sum_{j = 1}^n p_j \frac{\partial \varphi}{\partial p_j} \right], \quad i = 1, \ldots, n.
\end{equation}
We call $\boldsymbol{\pi}$ the {\it portfolio map generated by $\varphi$}. Note that we may specify $\varphi$ up to an additive constant without affecting the portfolio map. Conversely, the generating function for a given $\boldsymbol{\pi}$ is unique up to an additive constant.

The relative value process of a functionally generated portfolio satisfies a pathwise decomposition. To state this decomposition in discrete time, consider the {\it $L$-divergence} of $\varphi$ defined by
\begin{equation} \label{eqn:L.divergence}
{\bf L}_{\varphi}[\mathbf{q} : \mathbf{p}] = \log \left(1 + \nabla \varphi(\mathbf{p}) \cdot (\mathbf{q} - \mathbf{p}) \right) - (\varphi(\mathbf{q}) - \varphi(\mathbf{p})), \quad \mathbf{p}, \mathbf{q} \in \Delta_n.
\end{equation}
Exponential concavity of $\varphi$ guarantees that ${\bf L}_{\varphi}[\mathbf{q} : \mathbf{p}] \geq 0$; see \cite{PW16, W19}. If $e^{\varphi}$ is strictly concave, then ${\bf L}_{\varphi}[\mathbf{q} : \mathbf{p}] = 0$ only if $\mathbf{p} = \mathbf{q}$. See \cite{PW18, W18, W19} for in-depth studies of $L$-divergence from the perspectives of optimal transport and information geometry. The pathwise decomposition of the relative value process is given by
\begin{equation} \label{eqn:pathwise.decomposition}
\log V(t) = \left( \varphi(\boldsymbol{\mu}(t)) - \varphi(\boldsymbol{\mu}(0))\right) + \sum_{s = 0}^{t - 1} {\bf L}_{\varphi}[\boldsymbol{\mu}(s + 1) : \boldsymbol{\mu}(s) ], \quad t \geq 0.
\end{equation}
In this decomposition the former term can be thought of as representing  the change in market diversity: the market diversity increases (with respect to $\varphi$) if $\varphi(\boldsymbol{\mu})$ increases. The latter term represents the contribution of market volatility. If the first term is bounded and the market is sufficiently volatile we can see from this expression that the portfolio will outperform the market in the long term.

\subsection{A family of rank-based functionally generated portfolios} \label{sec:subfamily.fgp}
The set of all functionally generated portfolios is large even if we restrict to exponentially concave generating functions. In a typical data set the market weight process only occupies a small region of the (name-based) unit simplex $\Delta_n$, especially when $n$ is large. Optimizing directly over the space of all functionally generated portfolios (without further constraints or regularization) is likely to lead to overfitting, meaning poor out-of-sample performance. In \cite{IL20} this difficulty is avoided by assuming that the market weight process (in continuous time) has a known covariance structure and a known invariant density, and by focusing on the asymptotic growth rate as the time horizon tends to infinity. Here, we do not wish to make these assumptions as these objects, even if they exist, are not known exactly. Also, we want to have tools to tune the behaviors of the optimized portfolio over a finite horizon.  To exploit the long term stability (which is different from stationarity) of the ranked market weights $\mu_{(1)}(t) \geq \cdots \geq \mu_{(n)}(t)$ (see \cite[Chapter 5]{F02}), we single out a tractable family of rank-based functionally generated portfolios that will serve as the domain of our optimization problem.

We let $\mathcal{E}$ be the convex set of $C^1$ (continuously differentiable) functions $\ell$ on the unit interval $[0,1]$ that are exponentially concave, and such that $\ell(\frac{1}{2}) = 0$. Note that the derivatives are assumed to exist and be continuous up to the endpoints. 

\begin{lemma} \label{lem:fgp.subclass}
For $\ell \in \mathcal{E}$, the function
\begin{equation} \label{eqn:varphi.decomposable}
\varphi(\mathbf{p}) = \frac{1}{n} \sum_{i = 1}^n \ell (p_i)
\end{equation}
is exponentially concave on $\Delta_n$ and generates the portfolio map
\begin{equation} \label{eqn:fgp.weights}
\begin{split}
\boldsymbol{\pi}_i(\mathbf{p}) &= p_i \left(1 + \frac{1}{n} \ell'(p_i) - \frac{1}{n} \sum_{j = 1}^n p_j \ell'(p_j) \right), \quad 1 \leq i \leq n, \quad \mathbf{p} \in \Delta_n.
\end{split}
\end{equation}
By an abuse of notations we also say that $\boldsymbol{\pi}$ is generated by $\ell$. We denote the relative value of this portfolio by $V_{\ell}(t)$. Furthermore, if $\ell, \tilde{\ell} \in \mathcal{E}$ generate the same portfolio map $\boldsymbol{\pi}$, then $\ell - \tilde{\ell}\equiv 0$ on $[0, 1]$. 
\end{lemma}
\begin{proof}
The exponential concavity of $\varphi$ can be proved by the inequality of arithmetic and geometric means. The expression of the portfolio map follows from a computation using \eqref{eqn:fgp}. The last statement is a consequence of \cite[Proposition 6(i)]{PW16} and the fact that $\ell(\frac{1}{2})=\tilde{\ell}(\frac{1}{2})=0$.
\end{proof}

For later use, let us note that if $\boldsymbol{\pi}$ is given by \eqref{eqn:fgp.weights}, then for $\mathbf{p}, \mathbf{q} \in \Delta_n$ we have
\begin{equation} \label{eqn:rank.based.return}
\sum_{i = 1}^n \boldsymbol{\pi}_i(\mathbf{p}) \frac{q_i}{p_i} = 1 + \frac{1}{n} \sum_{j = 1}^n \ell'(p_i)(q_i - p_i).
\end{equation}

Now we give some examples to show that this construction contains a wide variety of portfolios. Note that for \eqref{eqn:fgp.weights} to be defined we only require that $\ell$ is differentiable on $(0, 1)$, and the constraint $\ell(\frac{1}{2}) = 0$ can be satisfied by adding a suitable constant to the function.

\begin{example}\label{eg:decomposable.generators.1}
If $\ell(x) \equiv 0$ is constant, then $\boldsymbol{\pi}(\mathbf{p}) \equiv \mathbf{p}$ is the market portfolio.
\end{example}

\begin{example}\label{eg:decomposable.generators.2}
If $\ell(x) = \log x$, then $\boldsymbol{\pi}(\mathbf{p}) \equiv \left(\frac{1}{n}, \ldots, \frac{1}{n}\right)$ is the equal-weighted portfolio.
\end{example}

\begin{example}\label{eg:decomposable.generators.3}
More generally, if $\ell(x) = \lambda \log x$, then $\ell$ is exponentially concave for $0 \leq \lambda \leq 1$. The portfolio generated is
\[
\boldsymbol{\pi}(\mathbf{p}) = (1 - \lambda) \mathbf{p} + \lambda \left(\frac{1}{n}, \ldots, \frac{1}{n}\right),
\]
which is a weighted average between the market and equal-weighted portfolios.
\end{example}

\begin{example}\label{eg:decomposable.generators.4}
Let $\ell(x) = -x^2/2$. Then $e^{\ell(x)} = e^{-x^2/2}$ is concave on $[0, 1]$. The portfolio generated is
		\[
		\boldsymbol{\pi}_i(\mathbf{p}) =   p_i \left(1 - p_i + \sum_{j = 1}^n p_j^2 \right).
		\]
\end{example}

\begin{remark}\label{rmk:log.optimal}
Portfolios generated by functions of the form $\varphi(\mathbf{p}) = \frac{1}{n} \sum_{i = 1}^n \ell(p_i)$ arise as the log-optimal portfolios under certain market models in continuous time. For a precise statement see \cite[Proposition 4.7]{CSW19}. Our setting is a special case where the ``drift characteristic'' of the market has the form $\nabla \varphi$.
\end{remark}


Thanks to the symmetry of the generating function $\varphi(\mathbf{p}) = \frac{1}{n} \sum_{i = 1}^n \ell(p_i)$, the induced portfolio map $\boldsymbol{\pi}$ is {\it rank-based}. More precisely, a permutation $\sigma$ of $(1, \ldots, n)$ acts on vectors by relabeling the coordinates, i.e., $\sigma \mathbf{p} = (p_{\sigma(1)}, \ldots, p_{\sigma(n)})$. From \eqref{eqn:fgp.weights}, for any $\sigma$ we have
\begin{equation} \label{eqn:portfolio.rank.based}
\sigma \boldsymbol{\pi}(\mathbf{p}) = \boldsymbol{\pi} (\sigma \mathbf{p}) \Rightarrow \boldsymbol{\pi}(\mathbf{p}) = \sigma^{-1} \boldsymbol{\pi}(\sigma \mathbf{p}).
\end{equation}
Given $\mathbf{p} = (p_1, \ldots, p_n) \in \Delta_n$, let
\[
p_{(1)} \geq p_{(2)} \geq \cdots \geq p_{(n)}
\]
be the ordered values of $\mathbf{p}$. In \eqref{eqn:portfolio.rank.based}, let $\sigma$ be the permutation such that $p_{\sigma_p(k)} = p_{(k)}$ for $1 \leq k \leq n$ (ties can be resolved by a fixed protocol). If $p_i = p_{(k)}$ then $\boldsymbol{\pi}_i(\mathbf{p}) = \boldsymbol{\pi}_k(\mathbf{u})$, where $\mathbf{u}=(u_1, \ldots, u_n):=(p_{(1)}, \ldots, p_{(n)})$. So the portfolio weight vector depends only on the ordered values of $\mathbf{p}$ (i.e., the capital distribution) as well as the permutation which gives the labels of the stocks. Thus $\boldsymbol{\pi}$ can be regarded as a portfolio generated by a function of ranked market weights (see \cite[Section 4.2]{F02}). By \cite[Proposition 3.4.2]{F02}, the symmetry of $\varphi$ implies that the weight ratio is monotone:
\begin{equation} \label{eqn:weight.ratio.monotone}
p_i \geq p_j \Rightarrow \frac{\boldsymbol{\pi}_j(\mathbf{p})}{p_j} \geq \frac{\boldsymbol{\pi}_i(\mathbf{p})}{p_i}.
\end{equation}

\begin{remark} \label{rem:div.weighted}
	The diversity-weighted portfolio
	\begin{equation} \label{eqn:diversity.weighted.portfolio}
		\boldsymbol{\pi}_i(\mathbf{p}) = \frac{p_i^{\theta}}{\sum_{j = 1}^n p_j^{\theta}}, \quad 1 \leq i \leq n,
	\end{equation}
	where $\theta$ is a parameter with values in $(-\infty, 1)$, is generated by the exponentially concave function $\varphi(\mathbf{p}) = \frac{1}{\theta} \log \left( \sum_{i = 1}^n p_i^{\theta} \right)$. Although $\varphi$ is symmetric, it is not given by \eqref{eqn:varphi.decomposable} for some $\ell \in \mathcal{E}$. Nevertheless, it can be approximated by portfolios of the form \eqref{eqn:fgp.weights}. Examples will be given in Section \ref{sec:empirical}.
\end{remark}

\medskip

Consider the portfolio map $\boldsymbol{\pi}: \Delta_n \rightarrow \overline{\Delta}_n$ generated by $\ell \in \mathcal{E}$. The relative log return over the time interval $[t, t + 1]$, given that the market weight moves from $\mathbf{p} = \boldsymbol{\mu}(t)$ to $\mathbf{q} = \boldsymbol{\mu}(t + 1)$, is given by
\begin{equation} \label{eqn:relative.return.general}
\log \frac{V_{\ell}(t + 1)}{V_{\ell}(t)} = \log \left( \sum_{i = 1}^n \boldsymbol{\pi}_i(\mathbf{p}) \frac{q_i}{p_i} \right).
\end{equation}
To take advantage of the rank-based nature of the portfolio we introduce another parameterization. Given $\boldsymbol{\mu}(t) = \mathbf{p}$ and $\boldsymbol{\mu}(t + 1) = \mathbf{q}$, let
\begin{equation} \label{eqn:def.u}
\mathbf{u} = (u_1, \ldots, u_n) = (p_{(1)}, \ldots, p_{(n)})
\end{equation}
be the vector of ordered values of $p$. It takes values in the {\it ordered unit simplex}
\begin{equation} \label{eqn:ordered.simplex}
\Delta_{n,\geq} = \{\mathbf{u} \in \Delta_n : u_1 \geq \cdots \geq u_n\}.
\end{equation}
Given $\mathbf{p}$, let $\sigma$ be the permutation such that $p_{\sigma(k)} = p_{(k)}$ for all $k$. Define 
\begin{equation} \label{eqn:def.v}
\mathbf{v}=(v_1,\dots,v_n)=(q_{\sigma(1)},\dots,q_{\sigma(n)}) \in \Delta_{n},
\end{equation}
where $v_k$ is the new market weight of the stock which was at rank $k$ at time $t$. corresponding to the ordered vector $\mathbf{u}$.
We define a probability vector $\mathbf{r} \in \Delta_n$ by
\begin{equation} \label{eqn:def.r}
r_k = \frac{v_{k}/u_{k}}{\sum_{m = 1}^n v_{m}/u_{m}}, \quad 1 \leq k \leq n.
\end{equation}
In words, $r_k$ is the normalized relative return of the stock which is at rank $k$ at time $t$. 

Introduce the notation
\begin{equation} \label{eqn:aitchison.addition}
\mathbf{a} \oplus \mathbf{b} = \left( \frac{a_1b_1}{ \mathbf{a} \cdot \mathbf{b}}, \ldots, \frac{a_nb_n}{ \mathbf{a} \cdot \mathbf{b}} \right)
\end{equation}
for $\mathbf{a}, \mathbf{b} \in \Delta_n$, where $\cdot$ is the dot product. This is the vector addition operation under the {\it Aitchison geometry} on the simplex \cite{egozcue2006hilbert}. The corresponding vector subtraction is given by
\begin{equation} \label{eqn:aitchison.subtraction}
\mathbf{a} \ominus \mathbf{b} = \left( \frac{a_1/b_1}{ \mathbf{a} \cdot (1/\mathbf{b})}, \ldots, \frac{a_n/b_n}{ \mathbf{a} \cdot (1/\mathbf{b})} \right),
\end{equation}
where $(1/\mathbf{b})_i = 1/b_i$. Comparing \eqref{eqn:def.r} with \eqref{eqn:aitchison.addition} and \eqref{eqn:aitchison.subtraction}, we have $\mathbf{r} = \mathbf{v} \ominus \mathbf{u}$ and $\mathbf{v} = \mathbf{u} \oplus \mathbf{r}$. Note that if $\mathbf{p} = \mathbf{q}$ (i.e., there is no volatility) then $\mathbf{r} = \overline{\mathbf{e}} :
= (\frac{1}{n}, \ldots, \frac{1}{n})$ which is the zero element of the Aitchison vector space.  We may recover $\mathbf{q}$ from $\mathbf{u}$, $\mathbf{r}$ and $\sigma$ via
\begin{equation} \label{eqn:pr.to.q}
q_{\sigma(k)} = v_k= \frac{u_k r_k}{\sum_{m = 1}^n u_m r_{m}}, \quad 1 \leq k \leq n.
\end{equation}
As shown in Section \ref{sec:rank.based.stability}, empirically the rank-based pair $(\mathbf{u}, \mathbf{r})$ is more stable than that of the name-based pair $(\mathbf{p}, \mathbf{q})$. For easy reference we state the notations explicitly:

\begin{notation}
Given name-based market weights $\mathbf{p}, \mathbf{q} \in \Delta_n$, we define $\mathbf{u} \in \Delta_{n, \geq}$ and $\mathbf{v} \in \Delta_n$ (see \eqref{eqn:def.u} and \eqref{eqn:def.v}) corresponding $\mathbf{p}$ and $\mathbf{q}$ respectively, and define $\mathbf{r} = \mathbf{v} \ominus \mathbf{u}$ which we interpret as the rank-based volatility.
\end{notation}

Rewriting \eqref{eqn:relative.return.general} by summing over the rank, we have
\begin{equation}
\log \frac{V_{\ell}(t + 1)}{V_{\ell}(t)} = \log \left( \sum_{k = 1}^n \frac{\boldsymbol{\pi}_k(\mathbf{u})}{u_k} \frac{u_k r_k}{\sum_{m = 1}^m u_m r_m}\right) = \log \left( \frac{\boldsymbol{\pi}(\mathbf{u}) \cdot \mathbf{r}}{\mathbf{u} \cdot \mathbf{r}} \right).
\end{equation}
Define $\mathcal{L}: \Delta_{n, \geq} \times \Delta_n \times \mathcal{E} \rightarrow \mathbb{R}$ by
\begin{equation}
    \label{eqn:L.p.r}
     \mathcal{L}(\mathbf{u}, \mathbf{r}; \ell):=\log \left( \frac{\boldsymbol{\pi}(\mathbf{u}) \cdot \mathbf{r}}{\mathbf{u} \cdot \mathbf{r}} \right).
\end{equation}
Note $\mathcal{L}$ is concave in $\ell$. While the rank-based representation \eqref{eqn:L.p.r} makes sense whenever the exponentially concave generating function is symmetric, the additive form $\varphi(\mathbf{p}) = \frac{1}{n} \sum_{i = 1}^n \ell(p_i)$ allows us to show that the portfolio optimization problem is mathematically tractable and can be implemented numerically.

Let $\tilde{\gamma} = \frac{1}{t} \sum_{s = 0}^{t - 1} \delta_{(\boldsymbol{\mu}(s), \boldsymbol{\mu}(s + 1))}$ be a probability measure on the product space $\Delta_n \times \Delta_n$. From \eqref{eqn:L.p.r}, we may write
\begin{equation} \label{eqn:logV.as.likelihood2}
\frac{1}{t} \log V_{\ell}(t) = \int_{\Delta_{n,\geq} \times \Delta_n} \mathcal{L}(\mathbf{u}, \mathbf{r} ; \ell) \mathrm{d} \gamma(\mathbf{u}, \mathbf{r}),
\end{equation}
where $\gamma$ is the probability measure on $\Delta_{n,\geq} \times \Delta_n$ given by the pushforward of $\tilde{\gamma}$ under the mapping $(\mathbf{p}, \mathbf{q}) \mapsto (\mathbf{u}, \mathbf{r})$. The name-based analogue of \eqref{eqn:logV.as.likelihood2} has been used in \cite{W15} to study a nonparameteric optimization problem over all functionally generated portfolios. Generalizing this set-up, the data of our optimization problem, to be stated formally in Problem \ref{prob:main.optimization}, will be given by a Borel probability measure $\gamma$ on the space $\Delta_{n,\geq} \times \Delta_n$.

In our portfolio optimization problem (Problem \ref{prob:main.optimization}) it is possible to assign separate weights to the diversity and volatility components in the pathwise decomposition \eqref{eqn:pathwise.decomposition}. To do so using rank-based parameterization, define the {\it diversity contribution}
\begin{equation}\label{eqn:div.contr}
\mathbf{D}_{\ell}[\mathbf{v} : \mathbf{u}] := \varphi(\mathbf{v}) - \varphi(\mathbf{u}) = \frac{1}{n} \sum_{k = 1}^n (\ell(v_k) - \ell(u_k))
\end{equation}
which is linear in $\ell$, and the $L$-divergence ${\bf L}_{\ell}[\cdot : \cdot] := {\bf L}_{\varphi}[\cdot : \cdot]$ (see \eqref{eqn:L.divergence}) which is concave in $\ell$. Then we have the decomposition
\begin{equation}\label{eqn:L.decomposition.w.D}
\begin{split}
\mathcal{L}(\mathbf{u},\mathbf{r};\ell) &=  \mathbf{D}_{\ell}[\mathbf{u} \oplus \mathbf{r} : \mathbf{u} ] +\mathbf{L}_{\ell}[\mathbf{u} \oplus \mathbf{r}:\mathbf{u}].
\end{split}
\end{equation}

\medskip


\subsection{$\beta$-smooth generating functions} \label{sec:beta.smooth}
For technical purposes, in particular to guarantee compactness of the feasible set, we will impose some regularity conditions on the generating function $\ell$.

\begin{definition} [$\beta$-smooth functions]
Let $\beta > 0$ be a constant. We say that a $C^1$ function $\ell$ on $[0, 1]$ is $\beta$-smooth if $\ell'$ is Lipschitz on $[0, 1]$ with constant $\beta$, i.e., $|\ell'(x) - \ell'(y)| \leq \beta |x - y|$ for $x, y \in [0, 1]$. We define $\mathcal{E}_{\beta}$ be those functions in $\mathcal{E}$ that are $\beta$-smooth. By an abuse of notation we also use $\mathcal{E}_{\beta}$ to denote the collection of portfolio maps generated by functions in $\mathcal{E}_{\beta}$.
\end{definition} 

Clearly, if $\ell \in \mathcal{E}\cap C^2([0, 1])$ then $\ell \in \mathcal{E}_{\beta}$ if and only if $|\ell''| \leq \beta$. The following result is standard and the proof is omitted.

\begin{lemma} \label{lem:compactness}
Define a metric on $\mathcal{E}_{\beta}$ by
\begin{equation} \label{eqn:beta.smooth.metric}
d(\ell_1, \ell_2) = \max_{x \in [0, 1]} |\ell_1'(x) - \ell_2'(x)|.
\end{equation}
Then $(\mathcal{E}_{\beta}, d)$ is a compact metric space.
\end{lemma}

\begin{example} {\ }
\begin{enumerate}
	\item[(i)] Let $\ell(x) = -\frac{x^2}{2} + \frac{1}{8}$ (see Example \ref{eg:decomposable.generators.4}). Since $\ell''(x) = -1$, $\ell$ is $1$-smooth and so $\ell \in \mathcal{E}_1$.
	\item[(ii)] Let $\ell(x) = \log (a + x) + b$ where $a > 0$ and $b \in \mathbb{R}$ is determined by the condition $\ell(\frac{1}{2}) = 0$. Since $|\ell''(x)| = \frac{1}{(a + x)^2} \leq \frac{1}{a^2}$ on $[0, 1]$, we have $\ell \in \mathcal{E}_{\beta}$ where $\beta = \frac{1}{a^2}$. The induced portfolio map is
	\[
	\boldsymbol{\pi}_i(\mathbf{p}) = \frac{1}{n} \frac{p_i}{a + p_i} + p_i \sum_{j = 1}^n \frac{1}{n} \frac{a}{a + p_j},
	\]
	which can be regarded as an interpolation between the equal-weighted portfolio ($\beta \rightarrow \infty$) and the market portfolio ($\beta \rightarrow 0$). Note that the generating function $\ell(x) = \lambda \log x$ (Example \ref{eg:decomposable.generators.3}), which also interpolates between the two portfolios, does not belong to $\mathcal{E}_{\beta}$ for any $\beta > 0$. Although the equal-weighted portfolio does not belong to $\mathcal{E}_{\beta}$ for any $\beta > 0$, the portfolio maps in our function class can be quite aggressive; some examples will be given in Section \ref{sec:empirical}.
\end{enumerate}
\end{example}

As suggested by the above examples, the parameter $\beta$ controls the maximum deviation of the portfolio map from the market portfolio. Here is a precise statement.

\begin{lemma} \label{lem:beta.interpretation}
Let $\ell \in \mathcal{E}_{\beta}$. Then for $\mathbf{p} \in \Delta_n$ and $1 \leq i \leq n$ we have
\begin{equation} \label{eqn:weight.ratio.bound}
e^{-\frac{2\sqrt{\beta}}{n}} - 1 \leq \frac{\boldsymbol{\pi}_i(\mathbf{p})}{p_i} - 1 \leq \frac{\beta}{n^2}, 
\end{equation}
Consequently, for any $\mathbf{u} \in \Delta_{n, \geq}$ and $\mathbf{r} \in \Delta_n$, we have
\begin{equation} \label{eqn:expL.bound}
	e^{-\frac{2\sqrt{\beta}}{n}}\leq \exp\left(\mathcal{L}(\mathbf{u}, \mathbf{r}; \ell) \right) \leq 1 + \frac{\beta}{n^2}.
\end{equation}
In particular, there exists a constant $M= M(n, \beta) > 0$ such that for $\ell \in \mathcal{E}_{\beta}$ we have $M^{-1} \leq \frac{V_{\ell}(t + 1)}{V_{\ell}(t)} \leq M$ for all $t$ and all market sequences.
\end{lemma}
\begin{proof}
See the Appendix. 
\end{proof}

\section{The portfolio optimization problem} \label{sec:optimization.problem}
Using the family $\mathcal{E}_{\beta}$ of rank-based functionally generated portfolios defined in Section \ref{sec:beta.smooth}, we will now formulate the portfolio optimization problem and investigate its theoretical properties.

\subsection{The optimization problem}
Let $\beta > 0$ be a given constant representing the maximum allowable deviation from the market portfolio. In practice $\beta$ can be quite large depending on the number of stocks (see Lemma \ref{lem:beta.interpretation}). Our choice of optimizing over $\mathcal{E}_{\beta}$ is motivated by its compactness (see Lemma  \ref{lem:compactness}) as well as the following result whose proof is provided in the Appendix.

\begin{lemma}\label{lem:density}
The set $\bigcup_{\beta>0} \mathcal{E}_\beta$ is dense in $\left(\mathcal{E},d\right)$, where $d$ is given by \eqref{eqn:beta.smooth.metric}.
\end{lemma}


To further control the behaviors of the portfolio map, we also specify a {\it regularization} $R: \mathcal{E}_{\beta} \rightarrow \mathbb{R}$ which is assumed to be convex and continuous with respect to the metric $d$. The regularization will appear in the objective function as $\lambda R(\ell)$ where $\lambda > 0$ is a tuning parameter. If needed, we may also specify a convex set $C \subset \mathcal{E}$ which is closed with respect to $d$ in order to further constrain the allowable portfolios. Examples of $R$ and $C$ will be given after we state the optimization problem.

Next we specify the data. Typically the raw data is given as a sequence of (named) market weight vectors $\{\boldsymbol{\mu}(s)\}_{s = 0}^t$ over a training horizon $[0, t]$. Following the formulation in \eqref{eqn:logV.as.likelihood2}, we convert the sequence $\{\boldsymbol{\mu}(s)\}_{s = 0}^t$ to an empirical measure
\begin{equation} \label{eqn:empirical.measure}
\frac{1}{t} \sum_{s = 0}^{t - 1} \delta_{(\mathbf{u}(s), \mathbf{r}(s))}
\end{equation}
on $\Delta_{n,\geq} \times \Delta_n$, where $\mathbf{u}(s) = \boldsymbol{\mu}_{(\cdot)}(s)$ is the ranked capital distribution at time $s$ and $\mathbf{r}(s)$ is given by \eqref{eqn:def.r}. As explained in Section \ref{sec:subfamily.fgp}, this device allows us to take advantage of the stability of the capital distribution. More generally, we may take as given a Borel probability measure $\gamma$ on the product set $\Delta_{n,\geq} \times \Delta_n$. Intuitively, $\mathbf{u}$ represents the (ranked) capital distribution and $\mathbf{r}$ represents the rank-based volatility.

Finally, we can give different weights to the diversity component $\mathbf{D}$ and volatility component $\mathbf{L}$ in the decomposition \eqref{eqn:L.decomposition.w.D} of the log relative return $\mathcal{L}$. Specifically, we consider the expression
\begin{equation} \label{eqn:L.decompose.as.div.vol}
\begin{split}
w_0 \mathbf{D}_{\ell}[\mathbf{u} \oplus \mathbf{r} : \mathbf{u}] + w_1\mathbf{L}_{\ell}[\mathbf{u}\oplus\mathbf{r}:\mathbf{u}] &=(w_0-w_1)\mathbf{D}_{\ell}[\mathbf{u} \oplus \mathbf{r}: \mathbf{u}]  +w_1 \mathcal{L} (\mathbf{u}, \mathbf{r} ; \ell) \\
 &=: \eta_0 \mathbf{D}_{\ell}[\mathbf{u} \oplus \mathbf{r}: \mathbf{u}] + \eta_1 \mathcal{L} (\mathbf{u}, \mathbf{r} ; \ell),
\end{split}
\end{equation}
where for notational simplicity (in the proofs) we use the weights $\eta_0 = w_0 - w_1$ and $\eta_1 = w_1$. 
Clearly, letting $\eta_0 = 0$, $\eta_1 = 1$ recovers the relative log return $\mathcal{L}$. Without loss of generality, we may let $\eta_1 = 1$ and think of $\eta_0$ as the extra weight on the diversity component which can be positive or negative. This provides additional flexibility to improve the finite sample performance. For example, over the training period the log return may be dominated by the fluctuation of the diversity, and we can let $\eta_0 < 0$ to downplay its role in the optimization. 

With these preliminaries we are ready to state the general version of the portfolio optimization problem, a special case of which was given in \eqref{eqn:optimization.problem}.

\begin{problem}[Regularized rank-based portfolio optimization] \label{prob:main.optimization}
Consider the set-up described above. Let $\gamma$ be Borel probability measure on $\Delta_{n,\geq} \times \Delta_n$, and let $\eta_0\in\mathbb{R},\eta_1\in (0, \infty)$ be given constant. Our portfolio optimization problem is
\begin{equation} \label{eqn:optim.objective}
\begin{split}
J(\gamma) := \sup_{\ell \in \mathcal{E}_{\beta} \cap C}  J(\ell;\gamma), \quad \quad \quad \quad \quad \quad \quad \quad  \quad \quad \quad \quad\\
\text{where } J(\ell; \gamma):= \int_{\Delta_{n,\geq} \times \Delta_n} \left\{ \eta_0 \mathbf{D}_{\ell}[\mathbf{u} \oplus \mathbf{r}: \mathbf{u}] + \eta_1 \mathcal{L} (\mathbf{u}, \mathbf{r} ; \ell)  \right\} \mathrm{d} \gamma(\mathbf{u}, \mathbf{r}) - \lambda R(\ell).
\end{split}
\end{equation}
\end{problem}

Since $\mathcal{L}$ is concave in $\ell$, $\mathbf{D}$ is linear in $\ell$, $R$ is convex in $\ell$, and the set $\mathcal{E}_{\beta} \cap C$ is convex, \eqref{eqn:optim.objective} is a convex optimization problem. While other objective (or utility) functions can be considered as long as the resulting problem is convex and the relevant technical conditions hold, we focus on the relative logarithmic return (represented by the term $\int \mathcal{L} d\gamma$) for its prevalence in SPT and the analysis of portfolio performance over long horizons. It is clear that \eqref{eqn:optim.objective} is strongly inspired by regularized empirical risk minimization in machine learning.


The regularization $R$ and the constraint set $C$ can be chosen to suit specific needs of the portfolio manager and are included in the formulation of the problem for added flexibility. Here are some examples.

\begin{example}\label{ex:reg.and.const.1}
The $L^2$-regularization
\begin{equation} \label{eqn:regularization.example.1}
R(\ell) = \int_0^1 (\ell'(x))^2 \mathrm{d}x
\end{equation}
penalizes deviation from the market portfolio where $\ell_{\text{market}}(x) \equiv 0$ (Example \ref{eg:decomposable.generators.1}).
\end{example}

\begin{example}\label{ex:reg.and.const.2}
Motivated by \cite{AJ18} which adopts separate benchmarks for outperformance and tracking, we may generalize Example \ref{ex:reg.and.const.1} and let $R$ penalize deviation from a given reference portfolio. For example, we may let
\[
R(\ell) =\int_0^1 (\ell'(x) - \ell_0'(x))^2 \mathrm{d}x
\]
for some $\ell_0 \in \mathcal{E}$. 
\end{example}

\begin{example}\label{ex:reg.and.const.3}
Since in practice most market weights are small, integration with respect to the Lebesgue measure on $[0, 1]$ may not be the most sensible choice. Alternatively, given a (rank-based) portfolio map $\hat{\boldsymbol{\pi}}: \Delta_{n,\geq} \rightarrow \overline{\Delta}_n$ (not necessarily functionally generated), we may let
\begin{equation} \label{eqn:regularization.example.3}
R(\ell) = \int_{\Delta_{n,\geq}} \| \boldsymbol{\pi}(\mathbf{u}) - \hat{\boldsymbol{\pi}}(\mathbf{u}) \|^2 \mathrm{d} ((\Pi_1)_{\#} \gamma) (\mathbf{u}),
\end{equation}
where $(\Pi_1)_{\#} \gamma$ is the pushforward of $\gamma$ under the projection map $\Pi_1 : \Delta_{n,\geq} \times \Delta_n \rightarrow \Delta_{n,\geq}$. An example of $\hat{\boldsymbol{\pi}}$ is the diversity-weighted portfolio \eqref{eqn:diversity.weighted.portfolio}. By \eqref{eqn:fgp.weights}, $R$ is quadratic, hence convex, in the generating function $\ell$. 
\end{example}

The following example of the convex constraint $C$ will be used in the empirical demonstration in Section \ref{sec:empirical}.

\begin{example}[Monotonicity]\label{ex:reg.and.const.4}
Portfolios generated by $\ell \in \mathcal{E}_{\beta}$ may be quite aggressive. For example, it is possible that the portfolio weights of the largest stocks are smaller than $\frac{1}{n}$ and those of the smallest stocks are greater than $\frac{1}{n}$ (i.e., more underweight/overweight than the equal-weighted portfolio). This may not be a desirable feature due to liquidity and transaction costs. Monotonicity of the portfolio weight, in the sense that $p_i \geq p_j \Rightarrow \boldsymbol{\pi}_i(\mathbf{p}) \geq  \boldsymbol{\pi}_j(\mathbf{p})$, can be guaranteed by the following result which is a refinement of \eqref{eqn:weight.ratio.monotone}. Clearly the condition defines a closed set $C \subset \mathcal{E}_{\beta}$. 
\end{example}

\begin{lemma} [Monotonicity of portfolio weights] \label{lem:weight.monotonicity}
	Let $\ell \in \mathcal{E}$. Suppose the function $x \mapsto x \ell'(x)$ is non-decreasing. Then $p_i \geq p_j$ implies $\boldsymbol{\pi}_i(\mathbf{p}) \geq \boldsymbol{\pi}_j(\mathbf{p})$, i.e., the portfolio weights are monotone with respect to the ranks.
\end{lemma}
\begin{proof}
For $\mathbf{p} \in \Delta_{n, \geq}$, we have
\[
\boldsymbol{\pi}_i(\mathbf{p})=p_i\left(1-\frac{1}{n}\sum_{k=1}^n p_k\ell'(p_k)\right)+\frac{1}{n}p_i\ell'(p_i).
\]
Since $p_k\ell'(p_k)\leq 1$ by the exponentially concavity of $\ell$ (see Lemma \ref{lem:exp.conc.property2}), the term in the parentheses above, which is independent of $i$, is non-negative. Hence for $p_i\geq p_j$, our assumption along with the above observation shows that $\boldsymbol{\pi}_i(\mathbf{p}) \geq \boldsymbol{\pi}_j(\mathbf{p})$.
\end{proof}


\subsection{Existence and uniqueness}
In the remainder of this section we will establish some theoretical properties of our optimization problem. We first show that the problem admits an optimal solution which is unique in an appropriate sense.

\begin{theorem}  \label{thm:existence}
	Let $\gamma$ be a Borel probability measure on $\Delta_{n,\geq} \times \Delta_n$. Then:
	\begin{enumerate}
		\item[(i)] Problem \ref{prob:main.optimization} admits an optimal solution $\ell^* \in \mathcal{E}_{\beta} \cap C$.
		\item[(ii)] The portfolio map $\boldsymbol{\pi}^*$ generated by $\ell^*$ is almost unique in the following sense. Suppose $\tilde{\ell}^*$ is another optimal solution and generates the portfolio map $\tilde{\boldsymbol{\pi}}^*$. Then
		\[
		\frac{\boldsymbol{\pi}^*(\mathbf{u}) \cdot \mathbf{r}}{\mathbf{u} \cdot \mathbf{r}}=\frac{\tilde{\boldsymbol{\pi}}^*(\mathbf{u}) \cdot \mathbf{r}}{\mathbf{u} \cdot \mathbf{r}}, \quad \text{for $\gamma$-almost all $(\mathbf{u}, \mathbf{r})$.}
		\]
	\end{enumerate}
\end{theorem}

The proof of Theorem \ref{thm:existence} relies on the following lemma which will be useful later.


\begin{lemma} \label{lem:log.value.estimate}
Let $\beta > 0$. There exists an explicit constant $K = K(n, \beta)$ such that for $\ell, \tilde{\ell} \in \mathcal{E}_{\beta}$ we have
\begin{equation} \label{eqn:ell.Lip.in.ell}
	\begin{split}
		\left| \mathcal{L}(\mathbf{u}, \mathbf{r} ; \ell) - \mathcal{L}(\mathbf{u}, \mathbf{r}; \tilde{\ell}) \right| \leq K d(\ell, \tilde{\ell}),
	\end{split}
\end{equation}
where $d$ is the metric on $\mathcal{E}_{\beta}$ defined by \eqref{eqn:beta.smooth.metric}.
\end{lemma}

\begin{proof}
See the Appendix.
\end{proof}

\begin{proof}[Proof of Theorem \ref{thm:existence}]
(i) Fix $\gamma$ and consider the mapping $\ell \in \mathcal{E}_{\beta} \mapsto J(\ell; \gamma)$, where $J$ is defined by \eqref{eqn:optim.objective}. By Lemma \ref{lem:log.value.estimate}, we have, upon integrating against $\gamma$,
\begin{equation} \label{eqn:objective.estimate}
	\left| \int (\mathcal{L}(\mathbf{u}, \mathbf{r} ; \ell) - \mathcal{L}(\mathbf{u}, \mathbf{r} ; \tilde{\ell})) \mathrm{d}\gamma \right| \leq Kd(\ell, \tilde{\ell}).
\end{equation}
Also, it is easy to see that
\begin{equation} \label{eqn:objective.estimate2}
\left| \int (\mathbf{D}_{\ell}[\mathbf{u} \oplus \mathbf{r}: \mathbf{u}] - \mathbf{D}_{\tilde{\ell}}[\mathbf{u} \oplus \mathbf{r}: \mathbf{u}]) \mathrm{d}\gamma \right| \leq 2d(\ell, \tilde{\ell}).
\end{equation}
So, thanks to our assumptions on $R$, the map $J$ is continuous (with respect to $d$). Also, by Lemma \ref{lem:compactness}, the space $\mathcal{E}_{\beta} \cap C$ is compact. Thus an optimal solution exists by the extreme value theorem.

(ii) This follows from the strict concavity of $\log(\cdot)$ alongside the assumed convexity of $R$ and linearity of $\mathbf{D}$.
\end{proof}

\subsection{Stability estimate}
In practice, the measure $\gamma$ is typically constructed using historical data as in \eqref{eqn:empirical.measure}. Although the (ranked) capital distribution exhibits long term stability, it does not appear to be stationary (see Section \ref{sec:empirical} for more discussion). In \cite{KR18, IL20} it is assumed that the covariance structure as well as the invariant density of the market weight process are known. Here, we do not wish make this assumption. To account for estimation error and nonstationarity of $\gamma$, we provide a stability estimate of our optimization problem with respect to the measure $\gamma$. That is, if $\ell^*$ is optimal for $\gamma$ and $\tilde{\gamma}$ is ``close to'' $\gamma$, then $\ell^*$ is almost optimal for $\tilde{\gamma}$. Also see Remark \ref{rmk:discreitize.measure} below for another motivation regarding numerical implementation.
 
Before stating the stability result we recall the concept of Wasserstein distance. See \cite{V03} for a general overview of optimal transport. Let $(\mathcal{X}, \rho)$ be a metric space. If $P$ and $Q$ are Borel probability measures on $\mathcal{X}$, we define the {\it $1$-Wasserstein distance}, with respect to the metric $\rho$, by
\begin{equation}
\mathbb{W}(P, Q) = \mathbb{W}_{1}(P, Q) = \inf_{R \in \Pi(P, Q)} \int_{\mathcal{X} \times \mathcal{X}} \rho(x, y) \mathrm{d}R(x, y),
\end{equation}
where $\Pi(P, Q)$ is the set of Borel probability measures on $\mathcal{X} \times \mathcal{X}$ whose first and second marginals are $P$ and $Q$ respectively. It is well known that $\mathbb{W}$ defines a metric on the space $\mathcal{W}_{1}(\mathcal{X})$ of Borel probability measures on $\mathcal{X}$ with finite first moments. Convergence in this distance is equivalent to weak convergence plus convergence of the first moments. 

\begin{theorem}[Kantorovich-Rubinstein duality] \cite[Theorem 1.14]{V03}
We have
\begin{equation} \label{eqn:KR.duality}
\mathbb{W}(P, Q) = \sup\left\{ \int_{\mathcal{X}} f \mathrm{d}(P - Q) : f \in \mathrm{Lip}_1(\mathcal{X}) \right\},
\end{equation}
where $\mathrm{Lip}_1(\mathcal{X})$ is the collection of $1$-Lipschitz functions on $\mathcal{X}$.
\end{theorem}

Take $\mathcal{X} = \Delta_{n, \geq} \times \Delta_n$. To construct a suitable metric $\rho$, consider first the {\it Hilbert projective metric} (see for example \cite[Remark 4.12]{PC19}) defined on $\Delta_n$ by
\[
d_\mathcal{H}(\mathbf{p}, \mathbf{q})=\log\left\{\max_{1\leq i \leq n} \frac{p_i}{q_i},  \max_{1 \leq j \leq n} \frac{q_j}{p_j}\right\}.
\]
Also let $\| \cdot \|_1$ be the Euclidean $1$-norm on $\mathbb{R}^n$. Write
\[
\tilde{d}(\mathbf{p}, \mathbf{q}) = \max\{ d_\mathcal{H}(\mathbf{p}, \mathbf{q}), \|\mathbf{p} - \mathbf{q}\|_1\}.
\]
Now we define a metric $\rho$ on $\mathcal{X}$ by
\begin{equation} \label{eqn:metric.rho}
\rho((\mathbf{u}, \mathbf{r}),(\mathbf{u}',\mathbf{r}'))=\tilde{d}(\mathbf{u},\mathbf{u}')+\tilde{d}(\mathbf{r},\mathbf{r}').
\end{equation}
This definition of $\rho$ is motivated by the following result.

\begin{lemma}\label{lem:L.lip}
There exist constants $K_0,K_1 > 0$ depending only on $n$ and $\beta$, so that for any $\ell\in\mathcal{E}_\beta$ the maps $\mathcal{L}(\cdot, \cdot ; \ell)$ and $\mathbf{D}_{\ell}[ \cdot : \cdot ]$ are Lipschitz in $(\mathbf{u}, \mathbf{r}) \in \mathcal{X}$ with constants $K_0$ and $K_1$: 
\[
\left| \mathbf{D}_{\ell}[\mathbf{u} \oplus \mathbf{r}: \mathbf{u}]   -  \mathbf{D}_{\ell}[\mathbf{u}' \oplus \mathbf{r}': \mathbf{u}']\right| \leq K_0\rho((\mathbf{u},\mathbf{r}),(\mathbf{u}',\mathbf{r}')).
\]
\[
\left|\mathcal{L}(\mathbf{u},\mathbf{r} ;\ell)-\mathcal{L}(\mathbf{u}',\mathbf{r}' ;\ell)\right| \leq K_1\rho((\mathbf{u},\mathbf{r}),(\mathbf{u}',\mathbf{r}')),
\]
\end{lemma}

\begin{proof}
See the Appendix.
\end{proof}

The explicit expressions of $K_0$ and $K_1$ are given in the proofs. We also assume that the regularization $R$ satisfies a similar estimate.

\begin{assumption}\label{ass:R.lip}
The regularization, which is convex in $\ell$, has the form	
\[
R(\ell)=\int \Phi(\mathbf{u},\mathbf{r};\ell) \mathrm{d} \gamma(\mathbf{u},\mathbf{r}),
\]	
where $\Phi(\cdot, \cdot; \ell)$ is $K_2$-Lipschitz on $\mathcal{X}$ with respect to $\rho$ for some constant $K_2 > 0$.
\end{assumption}

We are now ready to state the main result. Recall that $J(\gamma)$ is the optimal value of the optimization problem \eqref{eqn:optim.objective}.

\begin{theorem} \label{thm:continuity.Wasserstein}
Suppose the regularization $R$ satisfies Assumption \ref{ass:R.lip}, and let $ K = |\eta_0|K_0  + \eta_1K_1  + K_2$. Then for any $\gamma$, $\gamma'\in \mathcal{W}_{1}(\mathcal{X})$ we have
\begin{equation} \label{eqn:optimal.value.Lipschitz}
|J(\gamma) - J(\tilde{\gamma})| \leq K \mathbb{W}(\gamma, \tilde{\gamma}).
\end{equation}
Furthermore, if $\ell^*$ is optimal for $\gamma$, then
\begin{equation} \label{eqn:stability.estimate.2}
0 \leq J(\tilde{\gamma}) - J(\ell^*;\tilde{\gamma}) \leq 2K \mathbb{W}(\gamma, \tilde{\gamma}).
\end{equation}
\end{theorem}
\begin{proof}
Let $\ell^* \in \mathcal{E}_\beta\cap C$ be optimal for $\gamma$, and similarly let $\tilde{\ell}$ be optimal for $\tilde{\gamma}$. In light of Lemma \ref{lem:L.lip} and Assumption \ref{ass:R.lip}, we have that
\[
\Psi(\mathbf{u},\mathbf{r};\ell):=\eta_0\mathbf{D}_{\ell}[\mathbf{u}\oplus\mathbf{r}:\mathbf{r} ]+\eta_1\mathcal{L}(\mathbf{u},\mathbf{r} ;\ell)-\Phi(\mathbf{u},\mathbf{r}; \ell)
\]
is $K$-Lipschitz in $(\mathbf{u},\mathbf{r})$ with respect to $\rho$. For notational simplicity write $\Psi = \Psi(\mathbf{u},\mathbf{r}; \ell^*)$ and $\tilde{\Psi} = \Psi(\mathbf{u},\mathbf{r}; \tilde{\ell}^*)$, so that $J(\gamma) = \int \Psi \mathrm{d} \gamma$ and $J(\tilde{\gamma}) = \int \tilde{\Psi} \mathrm{d}\tilde{\gamma}$.

Without loss of generality, suppose that $J(\gamma) \geq J(\tilde{\gamma})$. Using the optimality of $\tilde{\ell}^*$ for $\tilde{\gamma}$, we have
\begin{equation*}
0 \leq J(\gamma) - J(\tilde{\gamma}) = \int \Psi \mathrm{d}\gamma - \int \tilde{\Psi}\mathrm{d} \tilde{\gamma} \leq \int \Psi \mathrm{d}\gamma - \int \Psi \mathrm{d}\tilde{\gamma} \leq \left| \int \Psi \mathrm{d} (\gamma - \tilde{\gamma}) \right|.
\end{equation*}
Since $\Psi$ is Lipschitz with constant $K$, we have $\frac{1}{K} \Psi \in \mathrm{Lip}_1(\mathcal{X})$. This and the Kantorovich-Rubinstein duality gives \eqref{eqn:optimal.value.Lipschitz}. The second estimate \eqref{eqn:stability.estimate.2} follows from a similar argument. 
\end{proof}

\begin{remark} \label{rmk:discreitize.measure}
For numerical implementation we may want to discretize $\gamma$ over a grid or approximate a theoretical $\gamma$ (e.g.~from a rank-based model) by simulation. Theorem \ref{thm:continuity.Wasserstein} can be interpreted as a consistency result as the mesh size or approximation error tends to zero.
\end{remark}

\section{Implementation via discretization}  \label{sec:algorithm}
We will implement the portfolio optimization (Problem \ref{prob:main.optimization}) via a finite dimensional discretization which is also convex. As will be shown in Theorem \ref{thm:consistency} below, the discretization error can be rigorously quantified. 

\subsection{Discretizing the problem}
The decision variable of our problem is an exponentially concave function $\ell \in \mathcal{E}_{\beta}$. We discretize $\ell$ over a partition $\mathcal{P}$ of $[0, 1]$: 
\[
\mathcal{P} = \{ 0 = x_1 < x_2 < \cdots < x_d = 1 \}.
\]
For normalization purposes ($\ell(\frac{1}{2}) = 0$) we assume $\frac{1}{2} \in \mathcal{P}$. We let $\delta = \max_i |x_{i + 1} - x_i|$ be the mesh size of $\mathcal{P}$. At each node $x_i$ we associate a function value $\boldsymbol{\ell}_i \in \mathbb{R}$. This gives a vector $\boldsymbol{\ell} = (\boldsymbol{\ell}_i)_{i = 1}^d \in \mathbb{R}^d$. We let $\hat{\boldsymbol{\ell}} : [0, 1] \rightarrow \mathbb{R}$ be the piecewise affine function such that $\hat{\boldsymbol{\ell}}(x_i) = \boldsymbol{\ell}_i$. Note that $\hat{\boldsymbol{\ell}}$ is differentiable except possibly at $x_2, \ldots, x_{d-1}$. At these points we define $\hat{\boldsymbol{\ell}}'(x_i)$ by either the left or right derivative. Below we will state constraints on $\boldsymbol{\ell}$ corresponding to the condition $\ell \in \mathcal{E}_{\beta} \cap C$.

Given $\boldsymbol{\ell}$ we define the weights
\[
\boldsymbol{\pi}_i(\mathbf{p}) = p_i\left(1+\frac{1}{n} \hat{\boldsymbol{\ell}}'(p_i) - \frac{1}{n} \sum_{j = 1}^n p_j \hat{\boldsymbol{\ell}}'(p_j) \right), \quad 1 \leq i \le n,
\]  
and define $\mathcal{L}(\mathbf{u}, \mathbf{r} ; \hat{\boldsymbol{\ell}})$, $\mathbf{D}_{\hat{\boldsymbol{\ell}}}[\mathbf{u} \oplus \mathbf{r} : \mathbf{u}]$ analogously to \eqref{eqn:L.p.r} and \eqref{eqn:div.contr}. Note that $\boldsymbol{\pi}$ is linear in $\boldsymbol{\ell}$ and hence $\mathcal{L}$ is concave in $\boldsymbol{\ell}$. Similarly, $\mathbf{D}_{\hat{\boldsymbol{\ell}}}[\mathbf{u} \oplus \mathbf{r} : \mathbf{u}]$ remains linear in $\boldsymbol{\ell}$. We assume that the measure $\gamma$ is given by an empirical measure in the form of \eqref{eqn:empirical.measure}. We approximate the regularization $R(\ell)$ by a convex function $\hat{R}(\boldsymbol{\ell})$ of $\boldsymbol{\ell}$. This gives the objective function
\[
\hat{J}(\boldsymbol{\ell} ; \gamma) := \frac{1}{t} \sum_{s = 0}^{t - 1}\left\{\eta_0\mathbf{D}_{ \hat{\boldsymbol{\ell}}}[\mathbf{u}\oplus\mathbf{r}(s), \mathbf{r}(s)]+ \eta_1\mathcal{L}(\mathbf{u}(s), \mathbf{r}(s) ; \hat{\boldsymbol{\ell}})\right\} - \lambda \hat{R}(\boldsymbol{\ell}).
\]

Now we formulate the constraints on $\boldsymbol{\ell} = (\boldsymbol{\ell}_i)_{i = 1}^d$. Exponential concavity of $\ell$ leads naturally to the constraint
\begin{equation} \label{eqn:exponential.concavity.constraint}
\frac{e^{\boldsymbol{\ell}_2}-e^{\boldsymbol{\ell}_1}}{x_2-x_1}\geq\cdot\cdot\cdot\geq\frac{e^{\boldsymbol{\ell}_{d}}-e^{\boldsymbol{\ell}_{d-1}}}{x_{d}-x_{d-1}}.
\end{equation}
Note that \eqref{eqn:exponential.concavity.constraint} can be expressed in the form
\begin{equation} \label{eqn:exponential.concavity.constraint2}
-\boldsymbol{\ell}_i + \log (w_i e^{\boldsymbol{\ell}_{i+1}} + (1 - w_i) e^{\boldsymbol{\ell}_{i-1}} ) \leq 0, \quad w_i = \frac{x_{i} - x_{i-1}}{x_{i+1} - x_{i-1}},
\end{equation}
for $i = 2, \ldots, d - 1$. The advantage of \eqref{eqn:exponential.concavity.constraint2} is that the inequality corresponds to a sublevel set of a convex function of $\boldsymbol{\ell}$.

Next, we define the forward differences
\[
\Delta_i \boldsymbol{\ell} = \frac{\boldsymbol{\ell}_{i + 1} - \boldsymbol{\ell}_i}{x_{i + 1} - x_i}, \quad i = 1, \ldots, d - 1,
\]
the backward difference $\Delta_d\boldsymbol{\ell}=\frac{\boldsymbol{\ell}_d-\boldsymbol{\ell}_{d-1}}{x_{d}-x_{d-1}}$, and impose the condition 
\[
\Delta_{i+1}\boldsymbol{\ell}-\Delta_{i}\boldsymbol{\ell} \geq -\frac{\beta}{2}\left(x_{i+2}-x_{i+1}\right)-\frac{\beta}{2}\left(x_{i+1}-x_{i}\right), \quad i = 1, \ldots, d - 2.
\]
This is the discrete analogue of the $\beta$-smooth constraint since the exponential concavity condition ensures that the above is non-positive. Note we necessarily have $\Delta_d\boldsymbol{\ell}=\Delta_{d-1}\boldsymbol{\ell}$. Finally we let $\hat{C} \subset \mathbb{R}^d$ be a closed convex set corresponding to $C \subset \mathcal{E}$ and we assume that $\hat{C} = \{\boldsymbol{\ell}:g_{\hat{C}}(\boldsymbol{\ell}) \leq 0\}$ for some convex function $g_{\hat{C}}$.

Now we are ready to state the discretized problem in the ``standard form'' of convex optimization, i.e., maximization of a concave function subject to inequality constraints given by sublevel sets of convex functions. This is helpful for solving the problem numerically using convex optimization software.

\begin{problem} [Discretized problem] \label{prob:discretized}
The discretized portfolio optimization problem is given by $\max_{\boldsymbol{\ell} \in \hat{C}} \hat{J}(\boldsymbol{\ell} ; \gamma)$ subject to
\begin{align} 
	& -\boldsymbol{\ell}_i + \log (w_i e^{\boldsymbol{\ell}_{i + 1}} + (1 - w_i) e^{\boldsymbol{\ell}_{i-1}} ) \leq 0, \quad i = 2, \ldots, d - 1, \notag \notag \\
	& -\left(\Delta_{i+1}\boldsymbol{\ell}-\Delta_{i}\boldsymbol{\ell}\right) - \frac{\beta}{2}\left(x_{i+2}-x_{i}\right) \leq 0, \quad  i = 1, \ldots, d - 2, \notag \\
	&\left(\Delta_{1}\boldsymbol{\ell}\right)^2-\beta \leq 0, \quad \mathrm{and} \quad \left(\Delta_{d}\boldsymbol{\ell}\right)^2-\beta \leq 0, \label{eqn:constraint.endpoint} \\
	& g_{\hat{C}}(\boldsymbol{\ell}) \leq 0, \notag \\
	& \hat{\boldsymbol{\ell}}(1/2) = 0. \notag
\end{align}
\end{problem}

We note that \eqref{eqn:constraint.endpoint} is a technical condition necessary in the proof of Theorem \ref{thm:consistency} below. It is clear that \eqref{prob:discretized} admits an optimal solution since the feasible set is a compact convex set in $\mathbb{R}^d$. We describe the implementation in Section \ref{sec:implementation}.

\subsection{Analysis of discretization error}
Solving Problem \ref{prob:discretized} gives an optimal $\boldsymbol{\ell}^* \in \mathbb{R}^d$ and the associated piecewise linear generating function $\hat{\boldsymbol{\ell}}^*$. While $\hat{\boldsymbol{\ell}}^*$ is sufficient for practical purposes, strictly speaking the function $\hat{\boldsymbol{\ell}}^*$ is not exponentially concave. In this section we will explicitly treat the approximation error $|J(\gamma) - \hat{J}(\boldsymbol{\ell}^*; \gamma)|$ as the mesh size of the partition tends to zero. Our main result is given by Theorem \ref{thm:consistency} and our approach is to construct a link between the discretized Problem \ref{prob:discretized} and the original Problem \ref{prob:main.optimization}. To facilitate this we use $\boldsymbol{\ell}^*$ to construct approximating functions in $\mathcal{E}_{\beta}$ and relate any $\ell\in\mathcal{E}_\beta$ to a vector $\boldsymbol{\ell}$ satisfying the constraints of Problem \ref{prob:discretized}. Together, these two results will allow us to estimate the approximation error. To arrive at this approximation error we will suppose the partitions $\mathcal{P}$ we consider are ``almost uniform" in the sense that for $\delta:=\max_i|x_{i+1}-x_i|$ and $\underline{\delta}=\min_i|x_{i+1}-x_i|$, we have $|\delta-\underline{\delta}|\leq M\underline{\delta}^3$ for some fixed $M>0$. The choice of exponent is chosen for technical reasons in the proof. When it is necessary to make the dependence on a given partition clear we will use the notation $\boldsymbol{\ell}_{\mathcal{P}}^*$ and $\hat{\boldsymbol{\ell}}^*_{\mathcal{P}}$ for the aforementioned vector and piecewise affine function, respectively. Also, in the analysis we assume that $R \equiv 0$ (hence $\hat{R} \equiv 0$) and do not impose the convex constraint $\ell \in C$ (hence $g_{\hat{C}} \equiv 0$). 


The first lemma is standard and the proof is omitted.

\begin{lemma}\label{lem:beta.smooth.der.approx}
If $\ell \in C^1([0, 1])$ is $\beta$-smooth and $x,y\in[0,1]$ then
\[
0\leq|\ell(y)-\ell(x)-\ell'(x)(y-x)|\leq \frac{\beta}{2}|y-x|^2.
\]
\end{lemma}

\begin{lemma}\label{lem:satisfied.constraints}
Fix a partition $\mathcal{P}$. If $\ell\in\mathcal{E}_\beta$ then $\boldsymbol{\ell} = (\boldsymbol{\ell}_i := \ell(x_i))_i$ satisfies the constraints of Problem \ref{prob:discretized}.
\end{lemma}

\begin{proof}
The first and last constraints follow immediately by exponential concavity and the fact that $\ell(\frac{1}{2})=0$. The second constraint follows from Lemma \ref{lem:beta.smooth.der.approx} since
\begin{align*}
    \left|\frac{\boldsymbol{\ell}_{i+2}-\boldsymbol{\ell}_{i+1}}{x_{i+2}-x_{i+1}}-\frac{\boldsymbol{\ell}_{i+1}-\boldsymbol{\ell}_{i}}{x_{i+1}-x_{i}}\right|&\leq\left|\frac{\boldsymbol{\ell}_{i+2}-\boldsymbol{\ell}_{i+1}}{x_{i+2}-x_{i+1}}-\ell'(x_{i+1})\right|+\left|\ell'(x_{i+1})-\frac{\boldsymbol{\ell}_{i+1}-\boldsymbol{\ell}_{i}}{x_{i+1}-x_{i}}\right|\\
    &\leq \frac{\beta}{2}|x_{i+2}-x_{i+1}|+\frac{\beta}{2}|x_{i+1}-x_{i}|=\frac{\beta}{2}(x_{i+2}-x_{i}).
\end{align*}

It remains to prove \eqref{eqn:constraint.endpoint}. We consider the first inequality (left endpoint) and the other one is similar. Let $\delta_0=x_2-x_1$. Since $\ell$ is continuous and differentiable on $[0,\delta_0]$, by the mean value theorem there exists $c\in(0,\delta_0)$ such that
\[
\ell'(c)=\frac{\ell(\delta_0)-\ell(0)}{\delta_0}=\frac{\boldsymbol{\ell}_2-\boldsymbol{\ell}_1}{\delta_0}.
\]
Applying Lemma \ref{lem:ell.upper.bound} in the Appendix then proves the bound.
\end{proof}


\begin{lemma}\label{lem:vector.bounds}
If $(\boldsymbol{\ell}_i)_{i=1}^d$ (where $d=|\mathcal{P}|$) satisfies the constraints of Problem \ref{prob:discretized} for $\mathcal{P}$ then (i) $|\boldsymbol{\ell}_{i+1}-\boldsymbol{\ell}_i|\leq \sqrt{\beta}|x_{i+1}-x_i|$ and (ii) $-\frac{\sqrt{\beta}}{2}\leq \boldsymbol{\ell}_i\leq \log(2)$.
\end{lemma}
\begin{proof}
For (i) it can be seen via the endpoint constraints \eqref{eqn:constraint.endpoint} that
\[\sqrt{\beta}\geq \frac{\boldsymbol{\ell}_{2}-\boldsymbol{\ell}_1}{x_2-x_1}\geq \frac{\boldsymbol{\ell}_{i+1}-\boldsymbol{\ell}_i}{x_{i+1}-x_i}\geq\frac{\boldsymbol{\ell}_{d}-\boldsymbol{\ell}_{d-1}}{x_d-x_{d-1}}\geq -\sqrt{\beta}\]
for $i=2,\dots,d-2$. The upper bound in (ii) follows by considering the piecewise affine function interpolating $(e^{\boldsymbol{\ell}_i})_{i=1}^d$ and applying the argument of Lemma \ref{lem:exp.conc.property2}. For the lower bound we suppose without loss of generality that $x_i\geq\frac{1}{2}=:x_k$. Then we have
\[|\boldsymbol{\ell}_i|=\left|\boldsymbol{\ell}_k+\sum_{j=0}^{i-k-1}\boldsymbol{\ell}_{k+j+1}-\boldsymbol{\ell}_{k+j}\right|\leq\sum_{j=0}^{i-k-1}\left|\boldsymbol{\ell}_{k+j+1}-\boldsymbol{\ell}_{k+j}\right|\leq \sqrt{\beta}|x_{i}-x_k|\leq\frac{\sqrt{\beta}}{2}\]
since $i\leq d$ and $x_d=1$.
\end{proof}

With this, we may link vectors $\boldsymbol{\ell}_{\mathcal{P}}\in\mathbb{R}^d$ satisfying the constraints of Problem \ref{prob:discretized} to  functions in $\mathcal{E}_\beta$. This is the main technical result required by the proof of Theorem \ref{thm:consistency}. The functions $\ell_{\alpha,\mathcal{P}}$ we construct are explicitly characterized in the proof and this may be of some independent interest. The proofs of the following two results are given in the Appendix.

\begin{proposition}\label{prop:construction}
Fix $\beta>0$ and $0<\alpha<1$. For each partition $\mathcal{P}$ satisfying $|\delta-\underline{\delta}|\leq M\underline{\delta}^3$, where $M>0$ is fixed, let $\boldsymbol{\ell}_{\mathcal{P}}=(\boldsymbol{\ell}_{\mathcal{P},i})_{i=1}^d$ be a vector satisfying the constraints of Problem \ref{prob:discretized}. Furthermore, let $\hat{\boldsymbol{\ell}}_{\mathcal{P}}$ be the piecewise affine function interpolating $\boldsymbol{\ell}_{\mathcal{P}}$. 
Then there exists a $\delta_{0,\alpha}>0$ depending only on $\alpha$, $M$ and $\beta$ such that for all such partitions with $\underline{\delta}<\delta_{0,\alpha}$ we can find a $\ell_{\alpha,\mathcal{P}}\in\mathcal{E}_\beta$ satisfying:
\[|\ell'_{\alpha,\mathcal{P}}(x)-\hat{\boldsymbol{\ell}}'_{\mathcal{P}}(x)|\leq K\underline{\delta}^{\alpha}\leq K\delta^\alpha, \ \ \ x\in[0,1],\]
where $K>0$ is a constant depending only on $\beta$ and $M$.
\end{proposition}



\begin{lemma} \label{lem:approx.error.J}
If $\ell\in\mathcal{E}_\beta$ and $(\ell_\delta)_{\delta > 0}$ are univariate functions on $[0, 1]$ satisfying
\[|\ell'(x)-\ell_\delta'(x)|\leq K_0\delta^\alpha, \ \ \ x\in[0,1],\]
for some $K_0,\alpha>0$, then there exists a $\delta_{0,\alpha}>0$ depending only on $\alpha$, $\beta$ and $n$ such that if $\delta<\delta_{0,\alpha}$ then for $R\equiv 0$ and a given probability measure $\gamma$, we have
\[
\left|J(\ell;\gamma)-J(\ell_\delta;\gamma)\right|\leq K_1\delta^\alpha,
\]
where the constant $K_1>0$ depends only on $\beta, n$ and $K_0$. (Note $J$ was originally defined on $\mathcal{E}$, but the evaluation of $J(\cdot ; \gamma)$ at $\ell_\delta$ makes sense.)
\end{lemma}


With this we are now ready to state and prove the main theorem.

\begin{theorem}\label{thm:consistency}
Consider Problem \ref{prob:discretized} for a given probability measure $\gamma$ when there is no constraint set $C$ and $R\equiv 0$. For each partition $\mathcal{P}$ satisfying $|\delta-\underline{\delta}|\leq M\underline{\delta}^3$, where $M>0$ is fixed, consider an associated optimizer $\boldsymbol{\ell}_{\mathcal{P}}^*$. For any $0<\alpha<1$ there exists a $\delta_{0,\alpha}>0$ depending only on $\alpha$, $\beta$, $M$ and $n$ such that if $\underline{\delta}<\delta_{0,\alpha}$ then:
\[
\left|J(\gamma)-\hat{J}(\boldsymbol{\ell}_{\mathcal{P}}^*;\gamma)\right|\leq K\underline{\delta}^{\alpha} \leq K \delta^{\alpha},
\]
where $K>0$ is a constant depending only on $\beta$, $M$ and $n$.
\end{theorem}

\begin{proof} Consider the piecewise affine function $\hat{\boldsymbol{\ell}}_{\mathcal{P}}^*$ interpolating $\boldsymbol{\ell}^*_{\mathcal{P}}$, and an optimizer $\ell^*$ for the original problem with measure $\gamma$ over $\mathcal{E}_\beta$. We have established via Proposition \ref{prop:construction} that for any $0<\alpha<1$ there is a $\delta_{1,\alpha}$ such that if $\underline{\delta}<\delta_{1,\alpha}$ then we can find $\ell_{\alpha,\mathcal{P}}\in\mathcal{E}_\beta$ satisfying \[|\ell'_{\alpha,\mathcal{P}}(x)-\hat{\boldsymbol{\ell}}_{\mathcal{P}}^*{}'(x)|\leq K_0\underline{\delta}^{\alpha}, \quad x\in[0,1],\]
for some $K_0>0$ depending only on $\beta$ and $M$. Then by the arguments in the proof of Lemma \ref{lem:approx.error.J} we can find a $\delta_{2,\alpha}<\delta_{1,\alpha}$ depending on $\alpha,\beta$ and $n$ such that if $\underline{\delta}<\delta_{2,\alpha}$ then
\[\left|\hat{J}(\boldsymbol{\ell}_{\mathcal{P}}^*;\gamma)-J(\ell_{\alpha,\mathcal{P}};\gamma)\right|=
  \left|J(\hat{\boldsymbol{\ell}}_{\mathcal{P}}^*;\gamma)-J(\ell_{\alpha,\mathcal{P}};\gamma)\right|\leq K_1\underline{\delta}^\alpha,\]
for some $K_1$ depending only on $\beta,M$ and $n$. Similarly, let $\hat{\boldsymbol{l}}_{\mathcal{P}}^{*}$ be the piecewise affine function interpolating $(\boldsymbol{l}^*_{\mathcal{P},i})_{i=1}^d$, given by $\boldsymbol{l}^*_{\mathcal{P},i}=\ell^*(x_i)$ on the partition $\mathcal{P}$. By our definition of the derivative at the mesh points, for any $x\in[0,1]$ there is an $x_i\in\mathcal{P}$ such that $|x_i-x|\leq \delta$ (here $\delta=\max_i|x_{i+1}-x_i|$) and $\hat{\boldsymbol{l}}^*_{\mathcal{P}}{}'(x)=\hat{\boldsymbol{l}}_{\mathcal{P}}^*{}'(x_i)$. This gives \[\left|\ell^*{}'(x)-\hat{\boldsymbol{l}}_{\mathcal{P}}^*{}'(x)\right|\leq\left|\ell^*{}'(x)-\ell^*{}'(x_i)\right|+\left|\ell^*{}'(x_i)-\hat{\boldsymbol{l}}_{\mathcal{P}}^*{}'(x_i)\right|\leq\beta\delta+\frac{\beta}{2}\delta=\frac{3\beta}{2}\delta,\] 
where we have used Lemma \ref{lem:beta.smooth.der.approx} and the $\beta$-smoothness of $\ell^*$. This is in turn bounded by $\frac{3\beta}{2}(M+1)\underline{\delta}$ by our assumptions. Then again by Lemma \ref{lem:approx.error.J} there exists a $\delta_3>0$ depending only on $\beta$, $M$ and $n$ such that if $\underline{\delta}<\delta_3$ then
\[\left|J(\ell^*;\gamma)-\hat{J}(\boldsymbol{l}_{\mathcal{P}}^{*};\gamma)\right|=\left|J(\ell^*;\gamma)-J(\hat{\boldsymbol{l}}_{\mathcal{P}}^{*};\gamma)\right|\leq K_2\underline{\delta},\]
for a constant $K_2>0$ depending only on $\beta$, $M$ and $n$. We will now choose $\delta_{0,\alpha}<\min\{\delta_{2,\alpha},\delta_3\}$ and fix $\underline{\delta}<\delta_{0,\alpha}$ alongside an associated partition $\mathcal{P}$. Lemma \ref{lem:satisfied.constraints} tells us that $\boldsymbol{l}_{\mathcal{P}}^*$ satisfies the constraints of the discretized problem and we have already noted that $\ell_{\alpha,\mathcal{P}}\in\mathcal{E}_\beta$ for our choice of $\underline{\delta}$. Hence by the definition of the supremum we have
\[J(\ell^*;\gamma)+K_1\underline{\delta}^\alpha\geq  J(\ell_{\alpha,\mathcal{P}};\gamma)+K_1\underline{\delta}^\alpha \geq \hat{J}(\boldsymbol{\ell}^*_{\mathcal{P}};\gamma)\geq \hat{J}(\boldsymbol{l}^*_{\mathcal{P}};\gamma). \]
As a result
\[K_1\underline{\delta}^\alpha\geq \hat{J}(\boldsymbol{\ell}^*_{\mathcal{P}};\gamma)-J(\ell^*,\gamma)\geq \hat{J}(\boldsymbol{l}^*_{\mathcal{P}};\gamma)-J(\ell^*;\gamma)\geq- K_2\underline{\delta}\]
by lower bounding the right hand side. Taken together we find
\[
\left|J(\ell^*;\gamma)-\hat{J}(\boldsymbol{\ell}^*_{\mathcal{P}};\gamma)\right|\leq K_1\underline{\delta}^\alpha+K_2\underline{\delta} \leq K \underline{\delta}^{\alpha},
\]
where $K =K_1+K_2$. Noting that $J(\ell^*;\gamma)=J(\gamma)$ completes the proof.
\end{proof}

\subsection{Implementation}\label{sec:implementation}
The implementation of the examples to follow in Section \ref{sec:empirical} was performed using the software CVX in MATLAB \cite{gb08,cvx}. CVX is a modeling system for convex optimization that supports disciplined convex programming and graph implementations. The convex program (Problem \ref{prob:discretized}) we specify is transformed by CVX to a suitable form that can be passed to the MOSEK optimization software \cite{mosek}, which is compatible with CVX. In particular, we make use of MOSEK's ability to solve convex problems of conic type and its recently introduced support for problems involving exponential cones. This is relevant to the formulation of our problem since the exponential concavity constraints can be expressed in terms of exponential cones. Finally, we remark that it is computationally convenient to use a non-uniform grid in order to exploit the fact that the market weights are typically clustered at the smaller values in the interval $[0,1]$ (i.e.~values $\leq \frac{1}{n}$). Sample code implementing the optimization has been made available on the first author's website at \url{www.stevenacampbell.com}.

\section{Empirical examples}  \label{sec:empirical}
In this section we illustrate our framework of functional portfolio optimization using CRSP data from the US stock market. Section \ref{sec:dataset} describes the data set and explains how the stock returns are computed. In Section \ref{sec:rank.based.stability} we provide evidence of the relative stability of rank-based measures $\gamma$ when compared to their name-based counterparts. In Section \ref{sec:SPT.setting.empirical} we work under the classic SPT setting (i.e., closed market as in Section \ref{sec:market}). We illustrate the behaviors of the optimized portfolios arising from different regularizations and constraints; in particular, under the monotonicity constraint (Lemma \ref{lem:weight.monotonicity}) we obtain portfolios that are qualitatively similar to the diversity-weighted portfolio (Remark \ref{rem:div.weighted}). Our empirical examples also elucidate the results of Lemma \ref{lem:beta.interpretation} by demonstrating our portfolio's sensitivity to the parameter $\beta$. Finally, in Section \ref{sec:open.mkt.setting.empirical} we consider the more realistic open market setting, allowing for changes in constituent stocks, delisting events, defaults and transaction costs. We show that the optimized portfolios still outperform the market under low to moderate transaction costs. Nevertheless, we stress that the general open market set-up involves many considerations that are beyond the framework of this paper. Some of these challenges, that we plan to address in future research, are highlighted in the discussion.

\subsection{Description of the dataset} \label{sec:dataset}
The market data for the following examples was obtained from The Center for Research in Security Prices (CRSP) which contains traded stocks on all major US stock exchanges \cite{CRSPdata}. The database contains data from 1926, but for simplicity and relevance we focus on the most recent 50 year window spanning January 1971 to December 2020. For this period, we compute the daily market capitalization $X_i$ of each stock in the CRSP database by considering the share price and multiplying by the number of shares outstanding. This allows us to compute the market weight of any given stock. Unless otherwise specified, transactions occur every $5$ (trading) days. All transactions are made at the prevailing market prices on the first day of the 5-day period. (Thus the $s$ as in e.g.~$\boldsymbol{\mu}(s)$ indexes the 5-day period.)


To construct a closed market to be used in Section \ref{sec:SPT.setting.empirical}, we restrict to the largest $n$ stocks (where $n = 100$) at the beginning of a given time window, say $[t_1, t_2]$, and compute the market sequence $\{\boldsymbol{\mu}(s)\}_{s=0}^t$ using the renormalized market weights. Also, in this subsection we neglect transaction costs and dividends. Thus the relative return of a portfolio is given by \eqref{eqn:portfolio.relative.value}.

In Section \ref{sec:open.mkt.setting.empirical} the stock returns will incorporate dividends and delisting events. Note that this allows us to correct for the selection bias against stocks leaving the market. Here we explain how the returns are computed, and leave the discussion of the construction of the trading universe and computation of transaction costs to Section \ref{sec:open.market.construction}. Following the analysis of \cite{ruf2020impact}, for each trading day we collect return data $\{\mathfrak{r}_i\}_i$ from the CRSP database which includes the contribution of dividends. We estimate the dividend contribution as in \cite{ruf2020impact} by
\[
\mathfrak{r}_i^D=\max\left\{1+\mathfrak{r}_i-\frac{X_i'}{X_i},0\right\},
\]
where $X_i$ and $X_i'$ are respectively old and new (closing) market capitalizations of stock $i$. The realized daily return is then defined by $\mathfrak{r}_i^R=\mathfrak{r}_i-\mathfrak{r}_i^D$. These returns are then compounded to get the return over each $5$-day period. We include the maximum above to ensure the dividend yield is non-negative. One situation where this can be required is when a company issues extra stocks and so we treat such scenarios as though there are no dividends paid. When there is a delisting event we use the delisting returns that are available through the CRSP database. In the (extremely) rare circumstance where return data is missing, we assume a return of $0\%$ in that period.


\subsection{Stability of rank-based dynamics}
\label{sec:rank.based.stability}

Rank-based models have been a consistent feature of stochastic portfolio theory \cite{banner2005atlas, chatterjee2010phase, F02, pal2011analysis}. As mentioned previously, it is worth examining the stability of rank-based distributions in practice since we exploit their properties. Recent work by Banner et al.~\cite{banner2018diversification} supports our hypothesis that there is a stability in rank-based models that is missing from their name-based counterparts. Specifically, they find that the ranked average log returns are approximately the same, while the variance of the log returns is positively associated with rank. This is used to explain the outperformance of particular portfolios by appealing to their excess growth rates.

The following simple example further illustrates the rationale for using rank-based data. Consider five disjoint 5-year periods from 1996--2020. After ranking the stocks in the market at the beginning of 1996 we select the largest 100 stocks that survive until 2020. In each of these time periods the rank-based data $(\mathbf{u},\mathbf{r})$ and the name-based data $(\mathbf{p},\mathbf{q})$ (sampled every $5$ days) have associated empirical distributions. We will denote the empirical distribution of time period $i=1,...,5$ by $\gamma_i$ for the rank-based data
and $\tilde{\gamma}_i$ for the name-based data. 
In Table \ref{tab:distances} we show the Wasserstein distances $\mathbb{W}(\gamma_i, \gamma_j)$ and $\mathbb{W}(\tilde{\gamma}_i, \tilde{\gamma}_j)$, where $\rho$ is the metric on $\Delta_n \times \Delta_n$ used in Theorem \ref{thm:continuity.Wasserstein}. These computations were performed using the package \cite{flamary2017pot}. From the table, we see that the rank-based quantities are more stable in the sense that that the distributions of $(\mathbf{u},\mathbf{r})$ stay ``closer" to each other across time than their name-based counterparts. It also appears that $\gamma$ is not stationary over time. 

\begin{table}[t!]
	\begin{center}
		\begin{tabular}{|c | c c c c c|}
			\hline
			$\mathbb{W}$ & $\gamma_1$ & $\gamma_2$ & $\gamma_3$ & $\gamma_4$ & $\gamma_5$ \\ [0.5ex] 
			\hline
			$\gamma_1$ & 0 & 1.66 & 1.49 & 1.53 & 2.60\\ 
			\hline
			$\gamma_2$ & 1.66 & 0 & 1.07& 1.14 & 1.78\\
			\hline
			$\gamma_3$ & 1.49 & 1.07 & 0 & 0.91 & 1.75\\
			\hline
			$\gamma_4$ & 1.53 & 1.14 & 0.91 & 0 & 1.74\\
			\hline
			$\gamma_5$ & 2.60 & 1.78 & 1.75 & 1.74 & 0\\ 
			\hline
		\end{tabular}
		\hspace{0.5ex}
		\begin{tabular}{|c | c c c c c|}
			\hline
			$\mathbb{W}$ & $\tilde{\gamma}_1$ & $\tilde{\gamma}_2$ & $\tilde{\gamma}_3$ & $\tilde{\gamma}_4$ & $\tilde{\gamma}_5$ \\ [0.5ex] 
			\hline
			$\tilde{\gamma}_1$ & 0 & 5.96 & 8.15 & 7.04 & 10.0\\ 
			\hline
			$\tilde{\gamma}_2$ & 5.96 & 0 & 7.35 & 6.13 & 8.03\\
			\hline
			$\tilde{\gamma}_3$ & 8.15 & 7.35 & 0 & 5.44 & 6.96 \\
			\hline
			$\tilde{\gamma}_4$ & 7.04 & 6.13 & 5.44 & 0 & 4.76\\
			\hline
			$\tilde{\gamma}_5$ & 10.0 & 8.03 & 6.96 & 4.76 & 0\\ 
			\hline
		\end{tabular}
	\end{center}
	\caption{Wasserstein distance (with respect to the metric $\rho$) of empirical distributions for the rank-based data $(\mathbf{u},\mathbf{r})$ (left) and the name-based $(\mathbf{p},\mathbf{q})$ (right) across five disjoint 5-year periods and in terms of the largest surviving 100 stocks ranked at the start of period 1. Generally $\mathbb{W}(\gamma_i, \gamma_j)$ is much smaller than $\mathbb{W}(\tilde{\gamma}_i, \tilde{\gamma}_j)$.}
	\label{tab:distances}
\end{table}

We note here that although using the surviving stocks in this analysis introduces some bias, it is unavoidable if we want to compare the name and rank-based distributions directly. This comparison would not be possible if the names in the market were allowed to vary. In future research we plan to use ideas from optimal transport to carry out a deeper study of the statistical properties of the capital distribution. 

\subsection{Sample portfolios and their performance in the closed market setting}\label{sec:SPT.setting.empirical}
We now return to our portfolio optimization (Problem \ref{prob:main.optimization}) and its discrete approximation via Problem \ref{prob:discretized}. 
In all cases the optimization was implemented using the empirical measure $\gamma$ (see \eqref{eqn:empirical.measure}) corresponding to the training period under consideration.

\subsubsection{Data Considerations}
Throughout this section we want to consider the classic SPT setting of a closed market implicitly assumed in Problems \ref{prob:main.optimization} and \ref{prob:discretized}. While this is unrealistic, it reflects the theoretical formulation of our problem and allows us to test its performance in the traditional setting before considering extensions in Section \ref{sec:empirical.performance.open.mkt}. To this end, we will break up our data into $5$ year intervals. In each 5 year period we form a closed market of the top $n$ stocks ($n = 100$) as determined at the beginning of the time interval. To ensure $(\mathbf{u},\mathbf{r})\in\Delta_{n,\geq}\times\Delta_n$, we require that the top stocks selected survive the entire 5 years. (Delisting events, including bankruptcy, occur very infrequently in the top 100 stocks.) When testing the performance of our trained generating function out of sample we form a new closed market of $100$ stocks at the beginning of the 5-year testing period.



\subsubsection{Effects of constraints and penalties on portfolio weights}
We first consider the portfolio map induced by the optimized function. We consider the average market vector $\bar{\mathbf{p}}$ in a given period and the portfolio weight vector $\boldsymbol{\pi}(\bar{\mathbf{p}})$ recommended by our optimized generating function at that market weight. Throughout we include for reference the market, equal-weighted and diversity-weighted portfolios (where we take $\theta=0.5$, see Remark \ref{rem:div.weighted}) and use the weights $\eta_0=0,\eta_1=1$ in \eqref{eqn:optim.objective}. Figure \ref{fig:std.v.mon.} compares the weights of different portfolios with and without the monotone weight constraint given by Example \ref{ex:reg.and.const.4}. The optimized portfolios were fit over the 5 year period from 2003--2007 using different values of $\beta$. Specifically, $\beta=10^4,10^5,10^8$ and $\beta=10^7,5\times10^7,10^8$ for the unconstrained and monotone problems, respectively. Also we set $R \equiv 0$ so there is no further regularization. Throughout this section the choice of $\beta$ was made primarily to illustrate the range of behaviours available to the portfolios. Notably, the unconstrained portfolio can in fact be more underweight the largest stocks than the equal-weighted portfolio which does not belong to $\mathcal{E}_{\beta}$ for any $\beta > 0$. 

\begin{figure}[t!]
\begin{center}
    \includegraphics[width=0.49\textwidth]{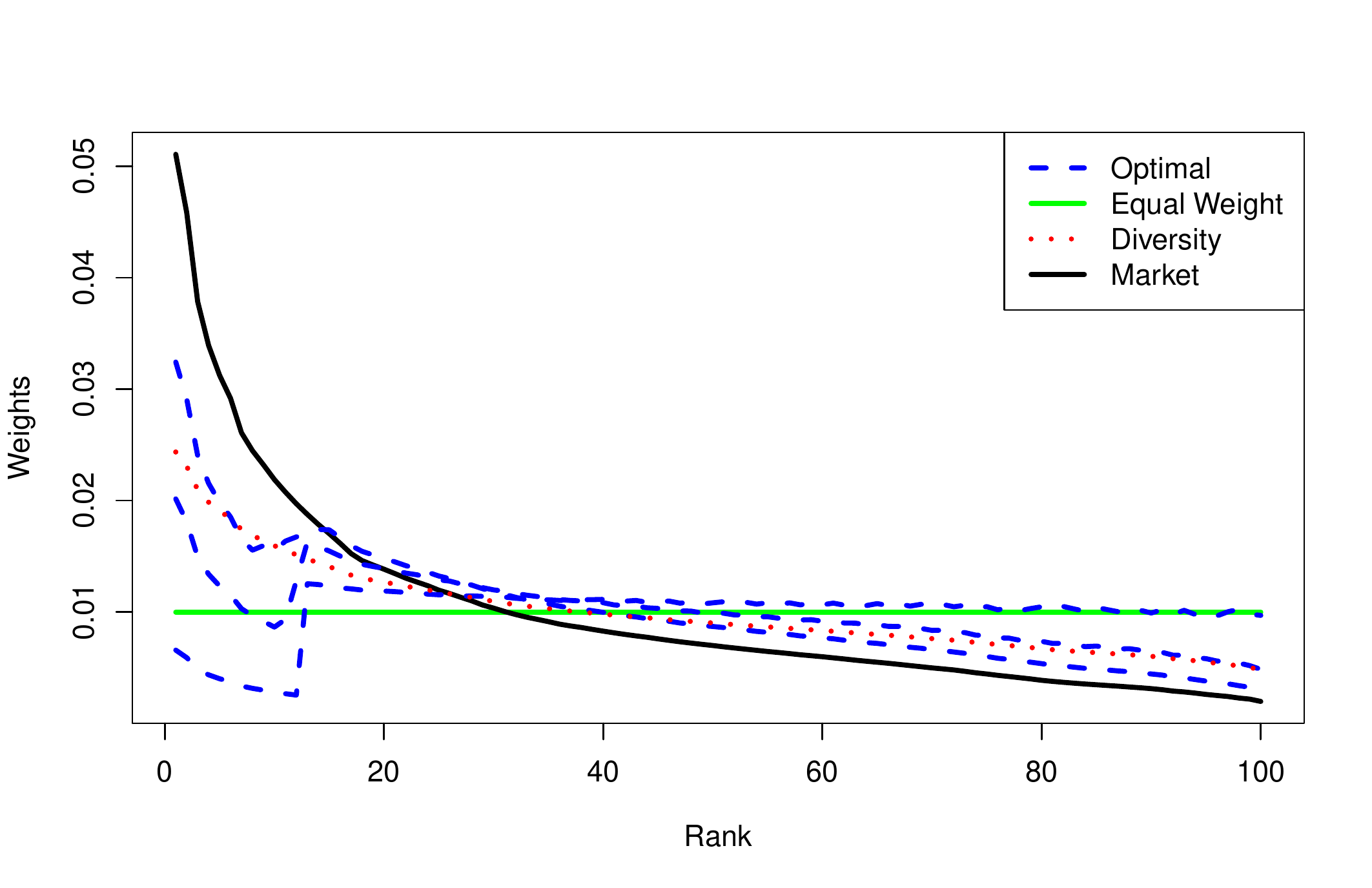}
    \includegraphics[width=0.49\textwidth]{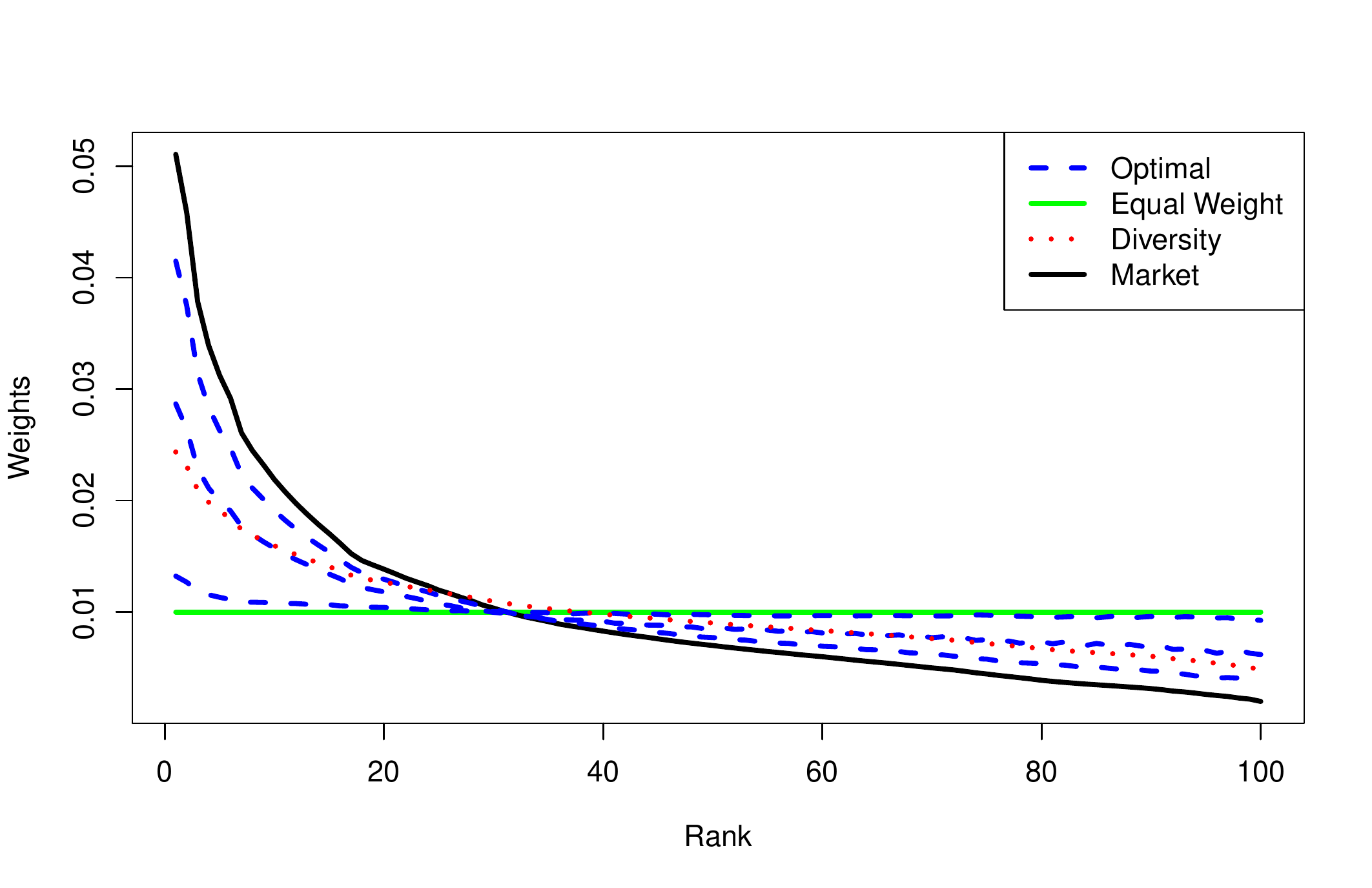}
\end{center}
\caption{Portfolio weights at the average market vector (arranged by rank) for the unconstrained (left) and monotone weight (right) versions of the problem with $\eta_0=0,\eta_1=1$ when $\beta=10^4,10^5,10^8$ and $\beta=10^7,5\times10^7,10^8$, respectively. Also $R = 0$. The equal-weighted, diversity-weighted and market portfolios are included for reference. As $\beta$ decreases the recommended portfolio moves closer to  the market portfolio.}
\label{fig:std.v.mon.}
\end{figure}

\begin{remark}\label{rem:div.similar}
In Figure \ref{fig:std.v.mon.} the portfolios with the monotonicity constraint are qualitatively similar to the diversity-weighted portfolio. This might be seen as a theoretical justification for the latter portfolio. Namely, it approximately optimizes the objective \eqref{eqn:optim.objective} under the monotone weight constraints.
\end{remark} 

Figure \ref{fig:penalties} considers the regularizations $R$ from \eqref{eqn:regularization.example.1} and \eqref{eqn:regularization.example.3}, where the reference portfolio is the market portfolio. The same 5 year period (2003--2007) is considered, but now $\beta$ is kept fixed at $10^8$. For the tuning parameter we use the values $\lambda=4\times10^{-8},4\times10^{-7},2\times10^{-6}$ for penalty \eqref{eqn:regularization.example.1} and $\lambda=10^{-3},2.5\times10^{-3},4\times10^{-3}$ for penalty \eqref{eqn:regularization.example.3}. As expected, in both cases as the penalty increases the optimal portfolio moves closer to that of the market. We also note that without the monotonicity constraint, the portfolio weight, as a function of rank, has a jump whose size and location depends on the tuning parameter. One intuition for the jump in the portfolio weight is that this arises when there is a material change in the value of the derivative of the generating function at the given weight. For example, this can be more pronounced when the derivative changes in sign from positive to negative.

\begin{figure}[h]
    \begin{center}
        \includegraphics[width=0.49\textwidth]{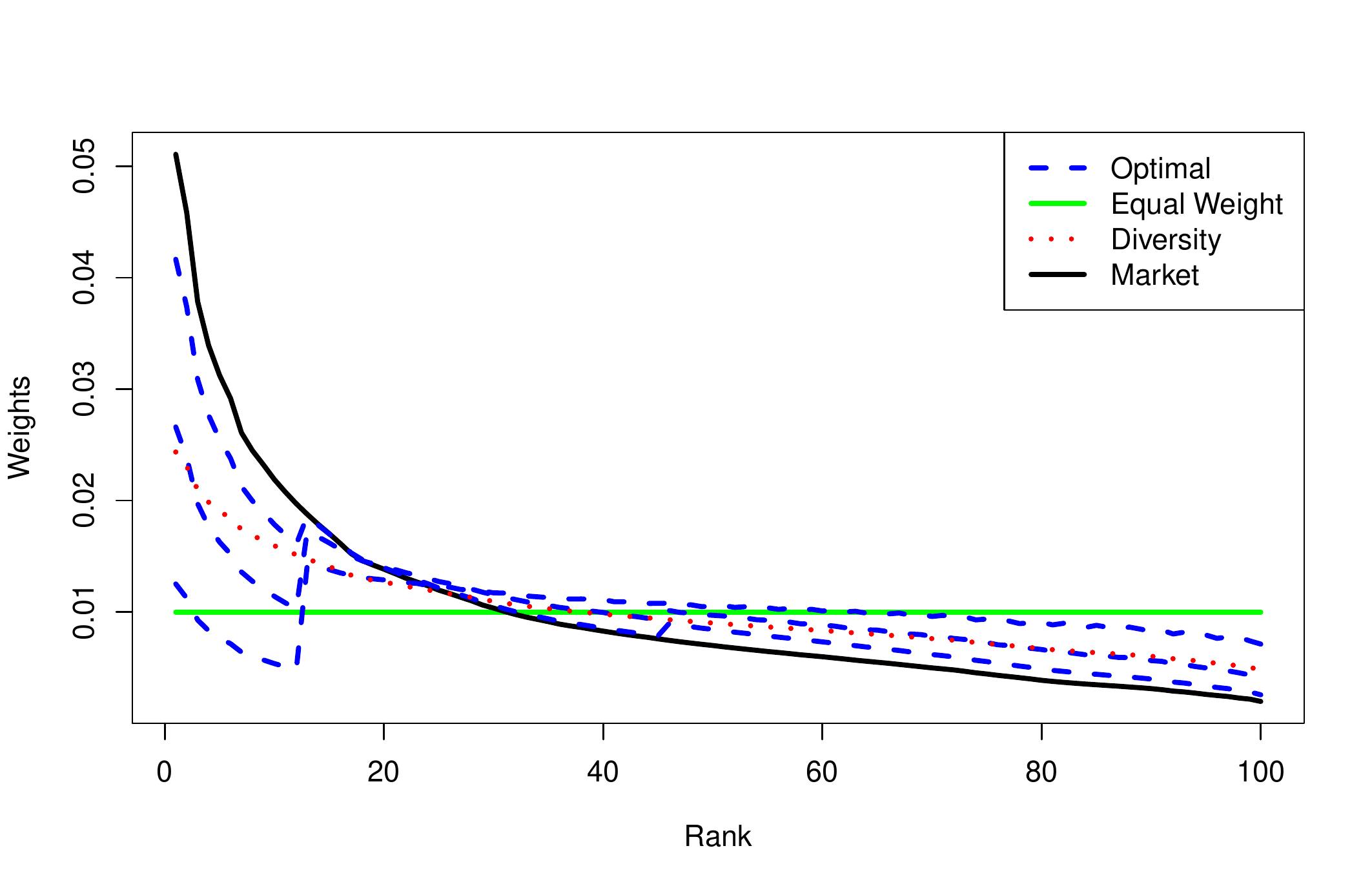}
        \includegraphics[width=0.49\textwidth]{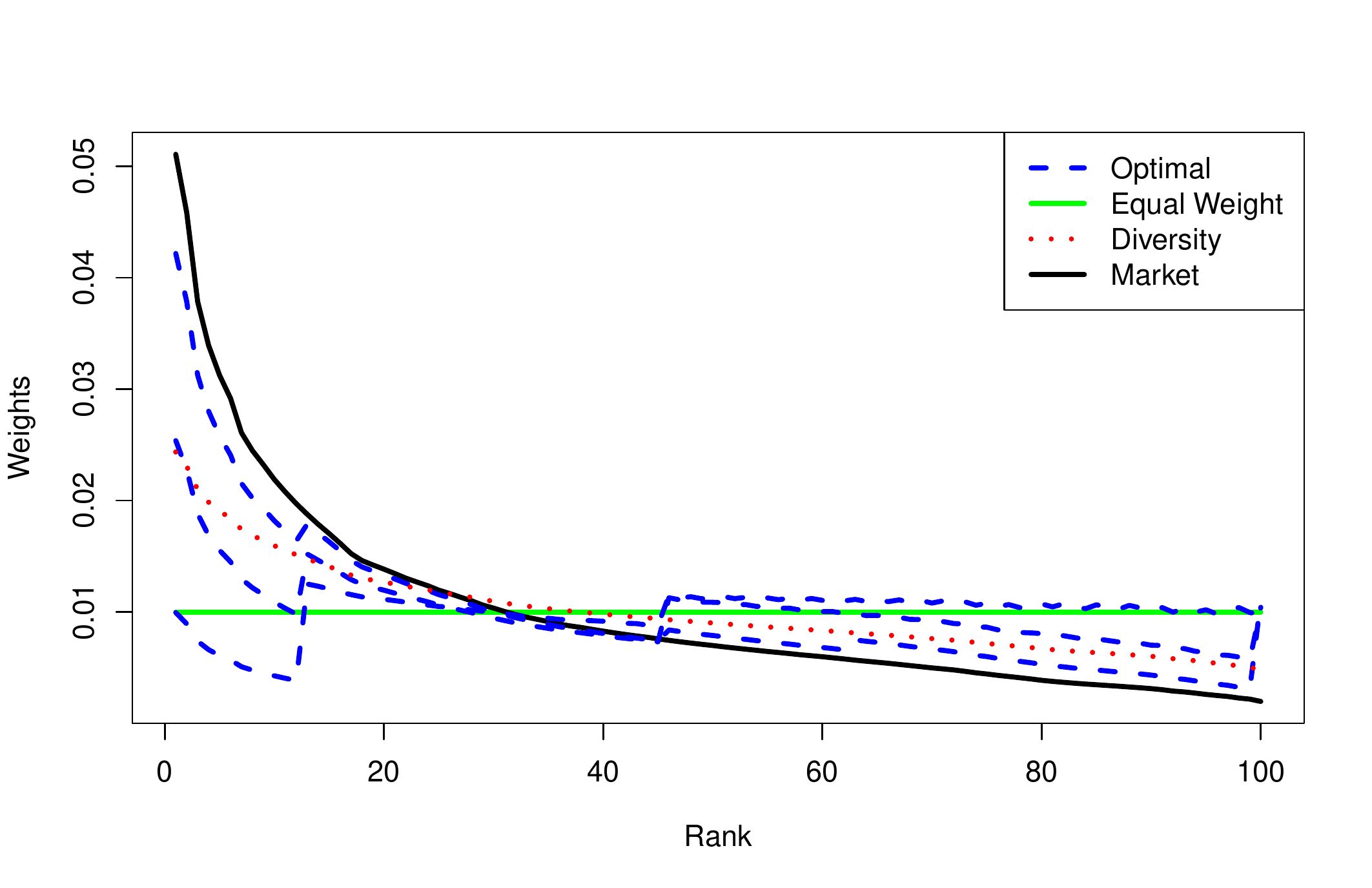}
    \end{center}
    \caption{Portfolio weights at the average market vector for $\beta=10^8$ and $\eta_0=0,\eta_1=1$ under a regularization penalty (left) with $\lambda=4\times10^{-8},4\times10^{-7},2\times10^{-6}$ , and when penalizing deviation from the market portfolio (right) with $\lambda=10^{-3},2.5\times10^{-3},4\times10^{-3}$. As $\lambda$ increases the recommended portfolio moves closer to the market portfolio.}
    \label{fig:penalties}
\end{figure}

\begin{figure}[h]
\begin{center}
    \includegraphics[width=0.49\textwidth]{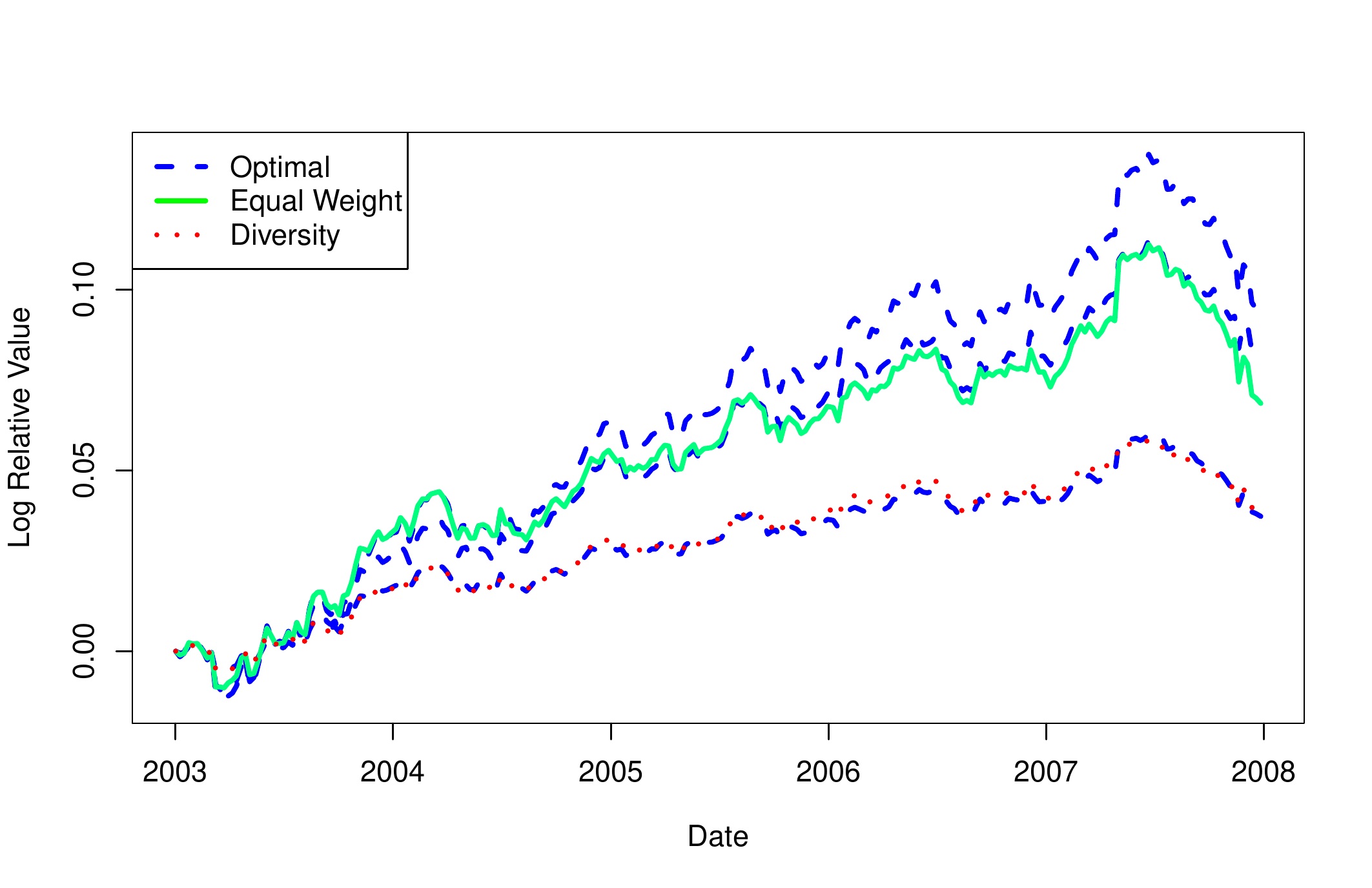}
    \includegraphics[width=0.49\textwidth]{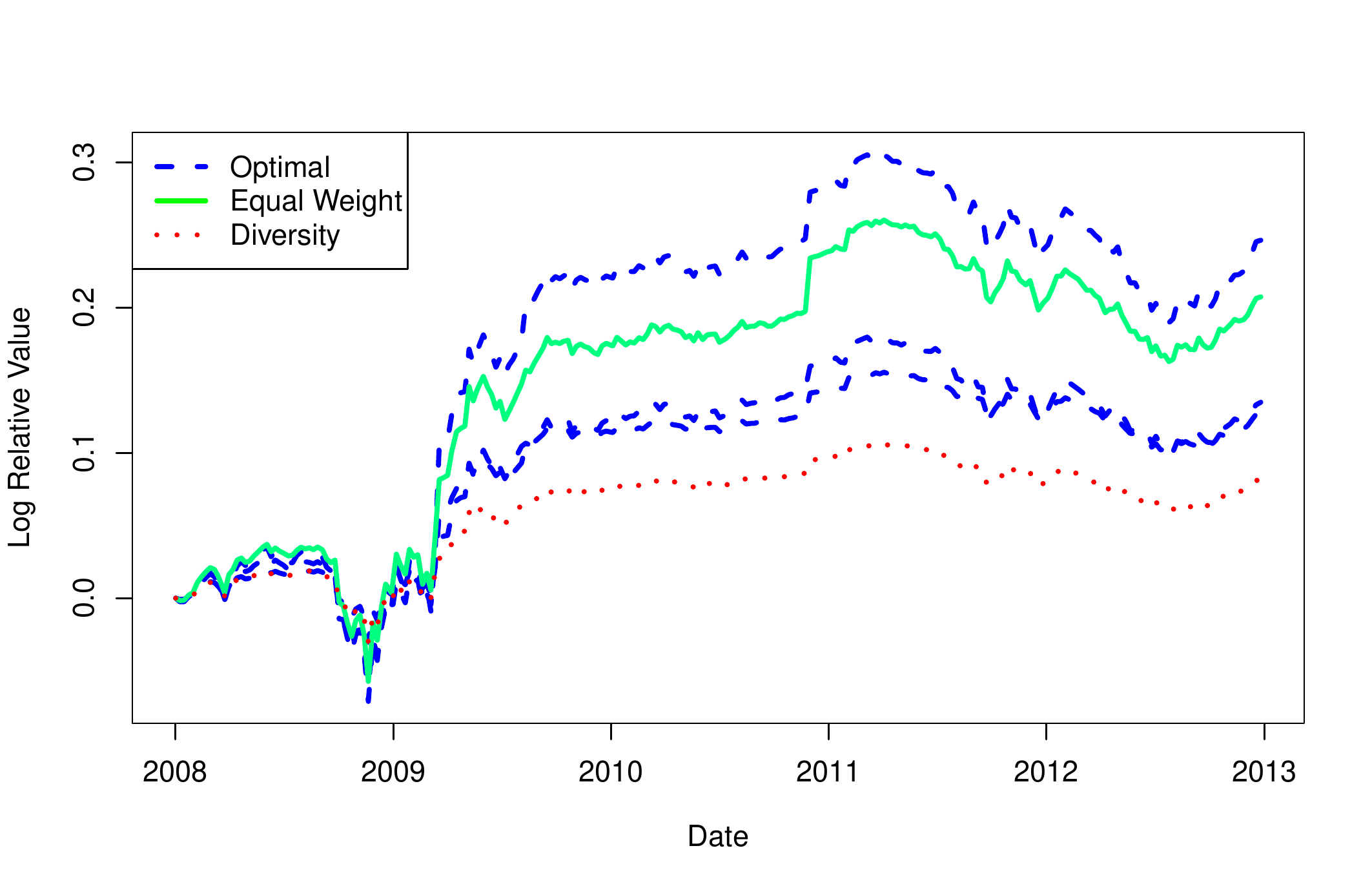}
\end{center}
\caption{Sample relative value performance  in the training period 2003--2007 (left) and in the testing period 2008--2013 (right). The three optimized portfolios involve either no additional constraints, the regularization constraint with $\lambda=4\times10^{-8}$, or the monotonicity constraint. The latter uses $\beta=5\times10^{7}$ and the former two use $\beta=10^{8}$. In all cases the optimization was conducted with $\eta_0=0,\eta_1=1$. The upper dashed curve is the unconstrained portfolio, the middle dashed curve is the regularized portfolio, and the lower dashed curve is the monotone portfolio.}
\label{fig:rel.val.perf}
\end{figure}

\subsubsection{Empirical performance} \label{sec:empirical.performance}
Finally we illustrate the performance of several of the above portfolios. Figure \ref{fig:rel.val.perf} shows the relative log value process of three optimized portfolios under the closed market setting. The training period is 2003--2007 (just before the financial crisis) and the testing period is 2008--2012. We show the performance of the portfolios in both periods. The portfolios we choose are the unconstrained version, the monotone version from Example \ref{ex:reg.and.const.4} and the regularized version from Example \ref{ex:reg.and.const.1} with $\lambda=4\times 10^{-8}$. The monotone portfolio was fit using $\beta=5\times 10^7$ and the others were fit using $\beta=10^8$. As mentioned above, each of the portfolios updates its holdings every $5$ trading days and transaction costs are neglected. Qualitatively we see a range in performance with the unconstrained portfolio outperforming the market in both the in-sample and out of sample settings. Also, the unconstrained portfolio outperforms the equal-weighted portfolio. Additionally, in keeping with Remark \ref{rem:div.similar} we see that the monotone weight portfolio performs very similarly to the diversity-weighted portfolio during the training period.

\subsection{Sample portfolio performance in the open market setting}\label{sec:open.mkt.setting.empirical}
Now we go beyond the setting of closed market considered in Section \ref{sec:SPT.setting.empirical}. The most important difference is that the universe of stocks is no longer kept fixed over time. Also, we incorporate dividends and proportional transaction costs. Since the portfolios studied in this paper are rank-based, we can consider, at a given time, the largest top $n$ (again $n = 100$ for illustrative purposes) stocks and apply $\boldsymbol{\pi}(\cdot)$ to the renormalized market weights. This setting is called an {\it open market} and is studied in the recent work \cite{karatzas2020open}. An interesting phenomenon that arises in the open market setting is the ``leakage effect"; see \cite{karatzas2020open,xie2020leakage} for detailed treatments. This setting takes into account features that are outside of the scope of the original optimization problems in Sections \ref{sec:optimization.problem} and \ref{sec:algorithm}. As a result, strictly speaking, our portfolios are not ``optimized" to this extended formulation. Here, we adapt our portfolios to the open market setting while keeping the spirit of our framework, and leave further modeling and optimization (such as prediction of market diversity and optimal rebalancing frequency) to future research. Our analysis follows the approach of \cite{ruf2020impact} which considered the performance of classic portfolios arising in SPT (e.g. equal-weighted and diversity-weighted portfolio) in the presence of these additional market features.

\subsubsection{Data Considerations} \label{sec:open.market.construction}
To construct an open market of ranked stocks we select a renewal frequency of 6 months for the constituent list. Every 6 months the largest 100 stocks are selected and delisting events or defaults are permitted in the interim. In the event of a delisting, the CRSP database provides a delisting return which we use to determine the closing value of a portfolio position. In the trading period following a delisting, all portfolios will distribute the funds amongst the remaining stocks, or the new stocks in the event that the constituent list is updated. When a stock is delisted we assume that there are no dividends paid. Our implementation of the portfolio transactions and trading costs exactly follows the approach of \cite{ruf2020impact} which the reader is referred to for further details.

When trading in the open market, we estimate our portfolio generating function in a rolling fashion. Specifically, we solve Problem \ref{prob:discretized} on the top 100 stocks over a 5 year period as in Section \ref{sec:empirical.performance}. In effect, we train our portfolio using a closed market and then test it by transacting in the open market setting. Specifically, the trained generating function is used to compute our portfolio weights for the 2 years following this training period. In this way we have a 5 year training/2 year testing split. For comparison purposes, the relative value is now calculated with respect to the value of the index tracking portfolio. That is, the portfolio with target holdings equal to the (relative) market weights of the top $100$ stocks.

\subsubsection{Empirical performance} \label{sec:empirical.performance.open.mkt}

We begin by comparing the performance of two optimized portfolios - unconstrained and monotone (see Example \ref{ex:reg.and.const.4})  - in the presence of various proportional transaction costs. The results can be found in Figure \ref{fig:mkt.performance}. The former portfolio was solved using $\beta=5\times10^8$ and the latter using $\beta=10^8$. In this first example, both optimization problems used $\eta_0=-0.5$ and $\eta_1=1$, meaning that we reduce the exposure of the optimized portfolio to the (realized) change in diversity during the training period (see \eqref{eqn:L.decompose.as.div.vol}). In Figure \ref{fig:varying.eta}, discussed below, we also illustrate the impact of varying these parameters.

\begin{figure}[h]
    \begin{center}
        \includegraphics[width=0.49\textwidth]{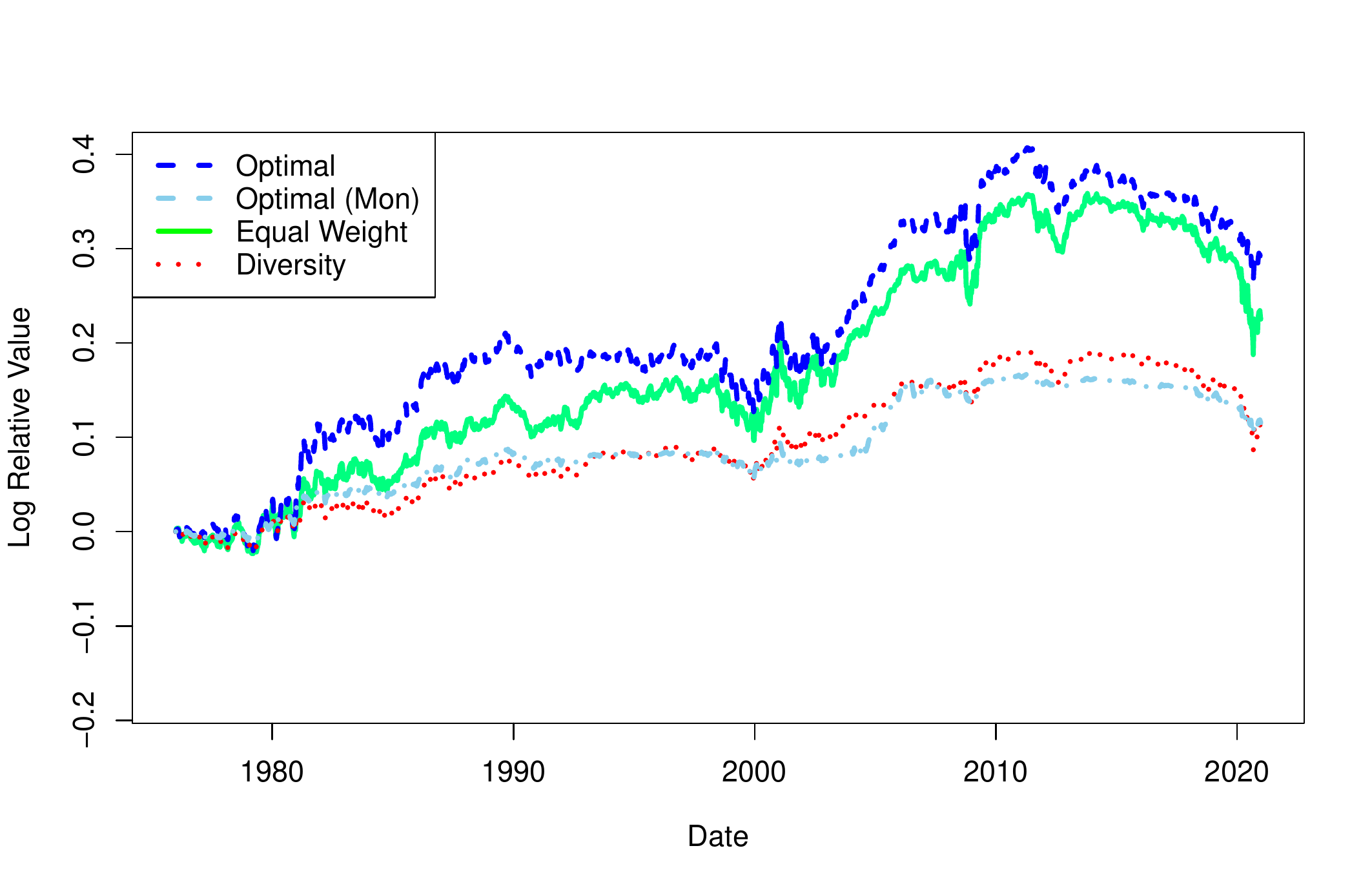}
        \includegraphics[width=0.49\textwidth]{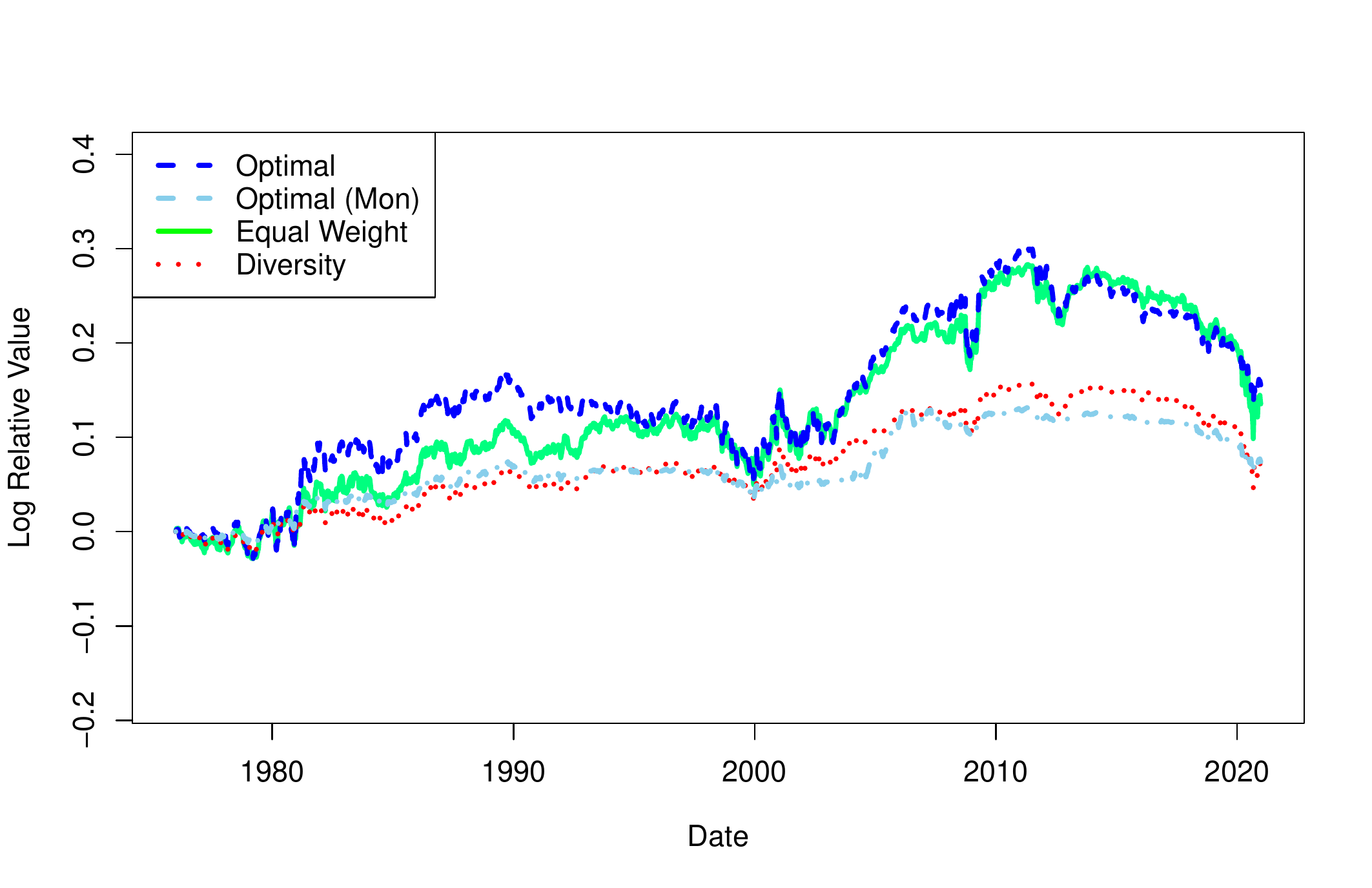}
    \end{center}
    \begin{center}
        \includegraphics[width=0.49\textwidth]{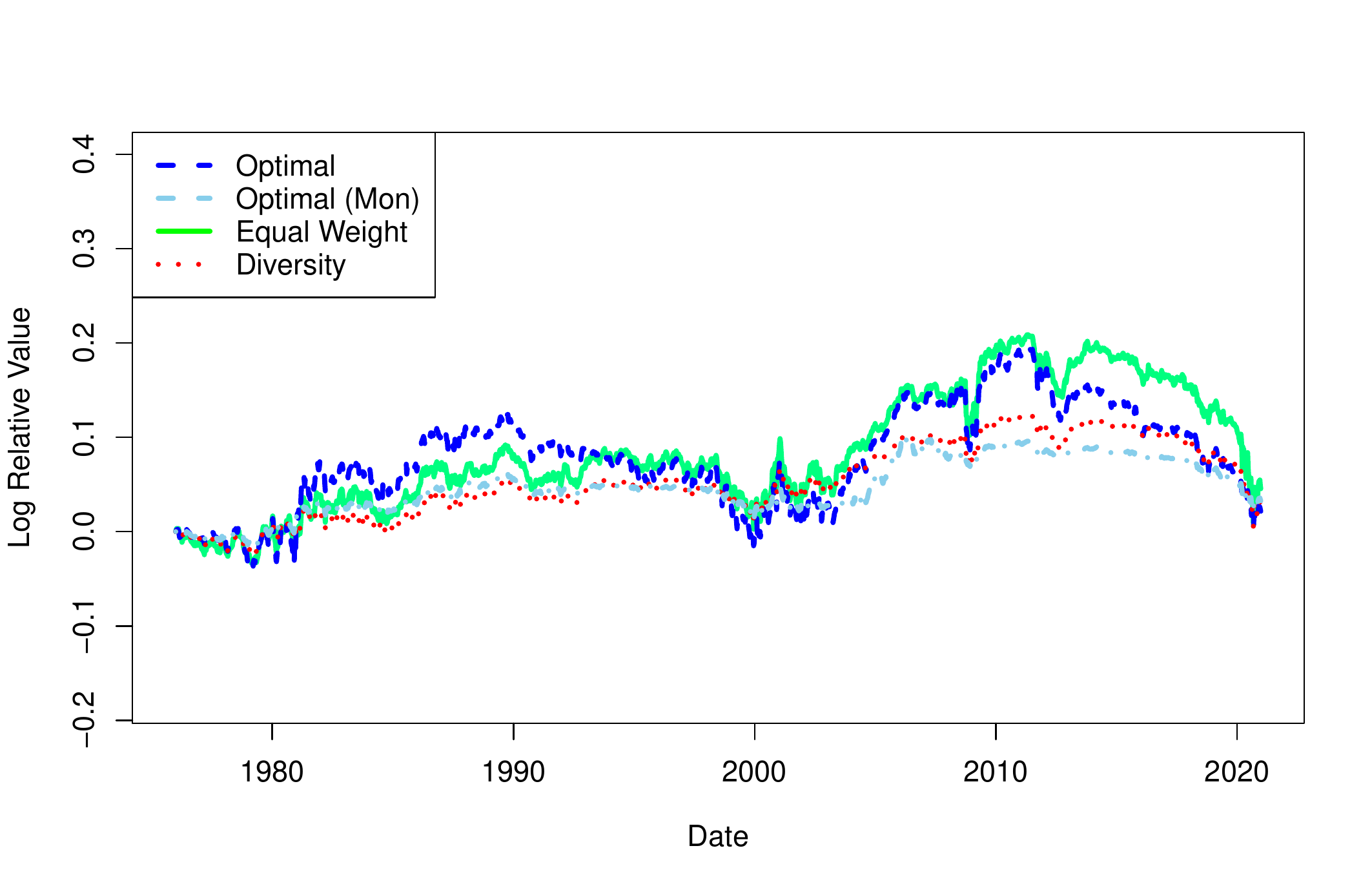}
        \includegraphics[width=0.49\textwidth]{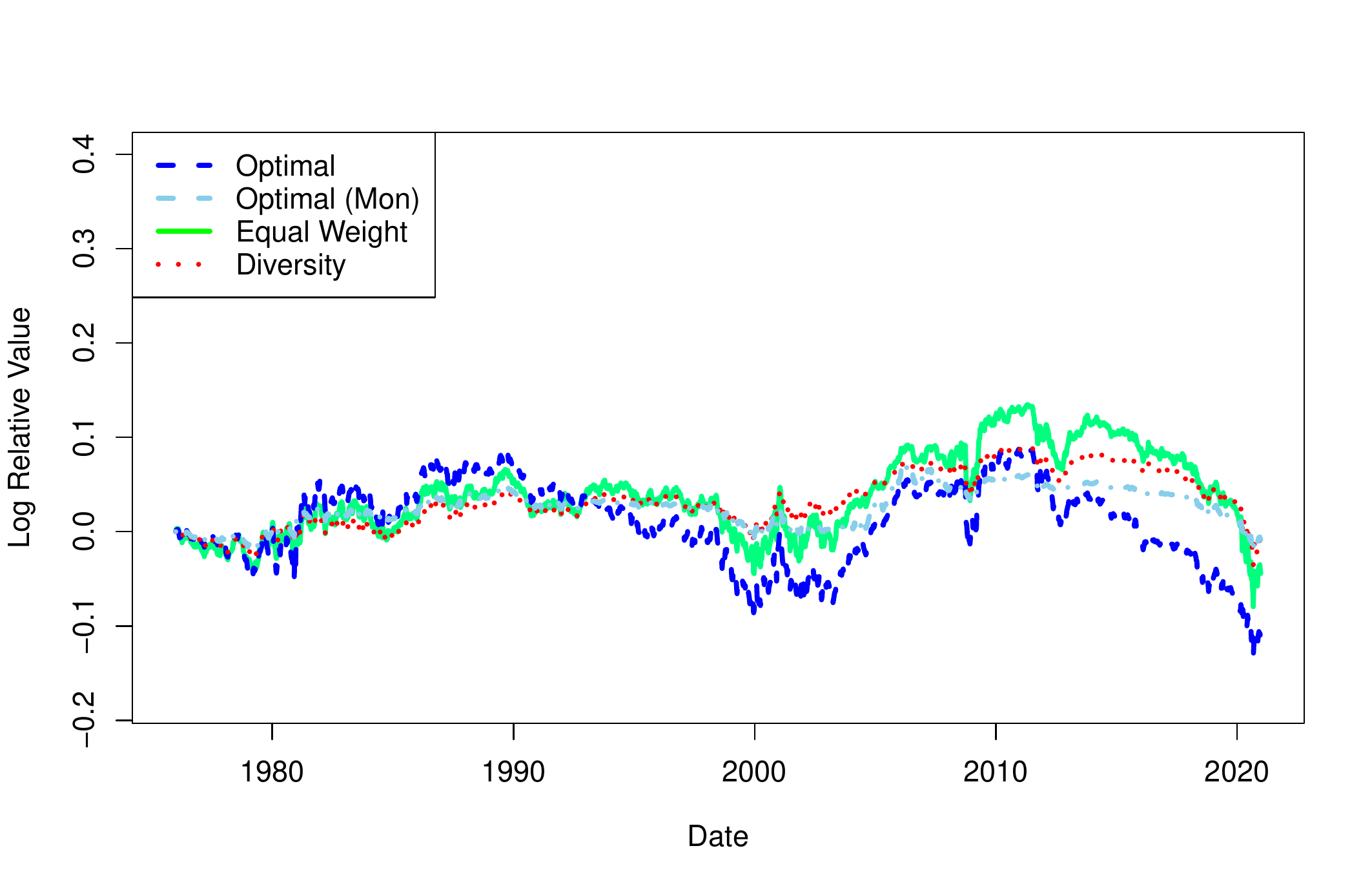}
    \end{center}
    \caption{Performance of the monotone and unconstrained portfolios in the open market setting relative to the index tracking portfolio. The diversity and equal weight portfolios are included for reference. The values of $\beta$ are $10^8$ and $5\times 10^8$, respectively. In both cases we take $\eta_1=1$, and $\eta_0=-0.5$. Proportional transaction costs are increasing from left to right, top to bottom ($tc=0\%,0.15\%,0.3\%,0.45\%$).}
    \label{fig:mkt.performance}
\end{figure}

As in \cite{ruf2020impact} we observe the deleterious impact of including transaction costs and the disproportionate effect on the equal-weighted, diversity-weighted and optimized portfolios. Note that the reference index tracking portfolio is also impacted by the imposition of transaction costs. Indeed, we see that as the proportional transaction costs (denoted by $tc$) increase, the relative performance of each of these portfolios decreases. Of these three portfolios in the low transaction cost regime, the optimized unconstrained portfolio, whose weights are the most aggressive, fares best, yet its terminal performance following the crash in February and March of 2020 is not materially different than that of the index tracking portfolio when the transaction costs are $0.3\%$. As before, the performance of the monotone weight and diversity-weighted portfolios is similar. For the highest transaction cost considered, $tc=0.45\%$, the optimized portfolio fares worst and is the most severely affected. 

\begin{figure}[h]
    \begin{center}
        \includegraphics[width=0.49\textwidth]{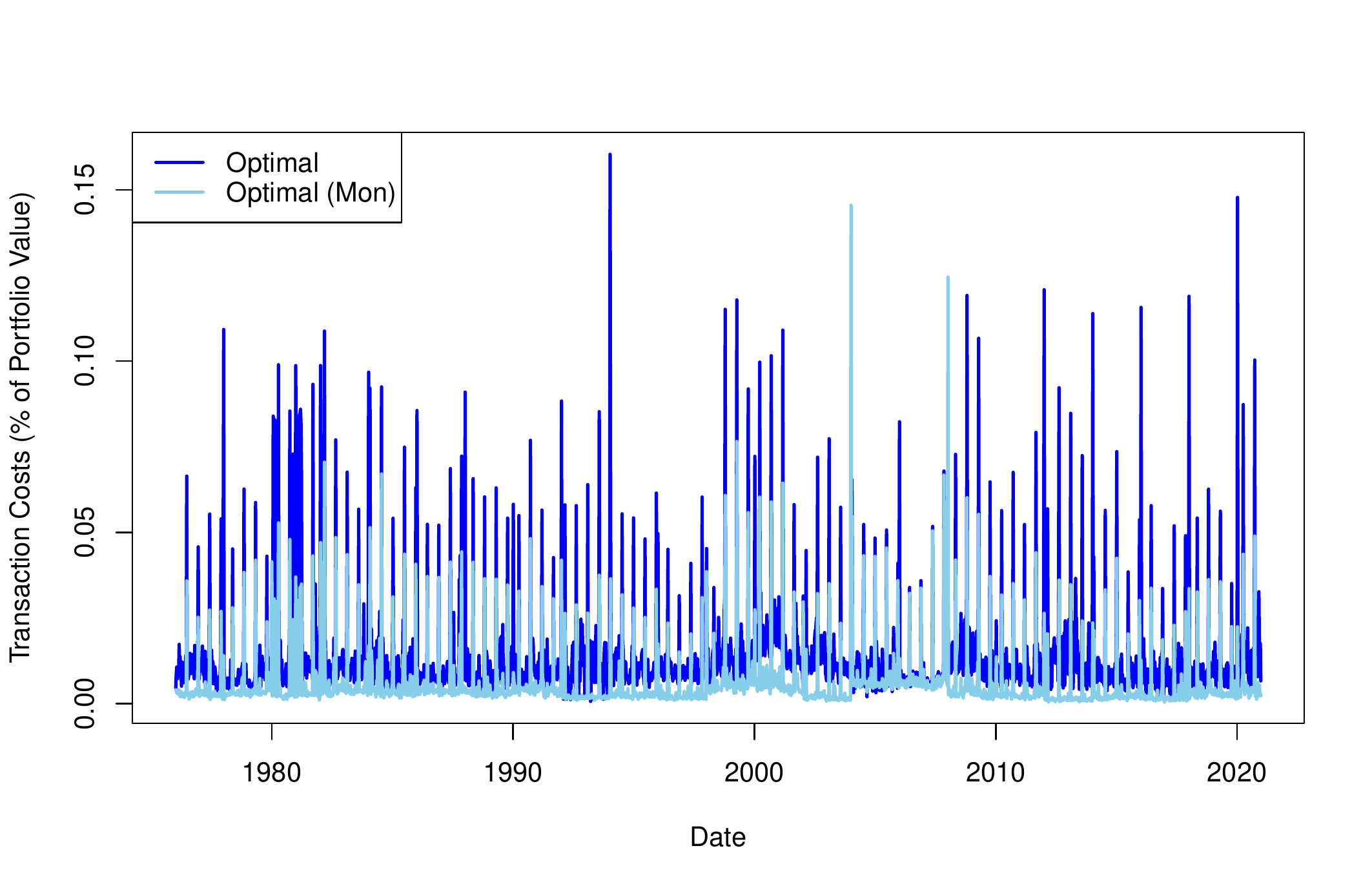}
        \includegraphics[width=0.49\textwidth]{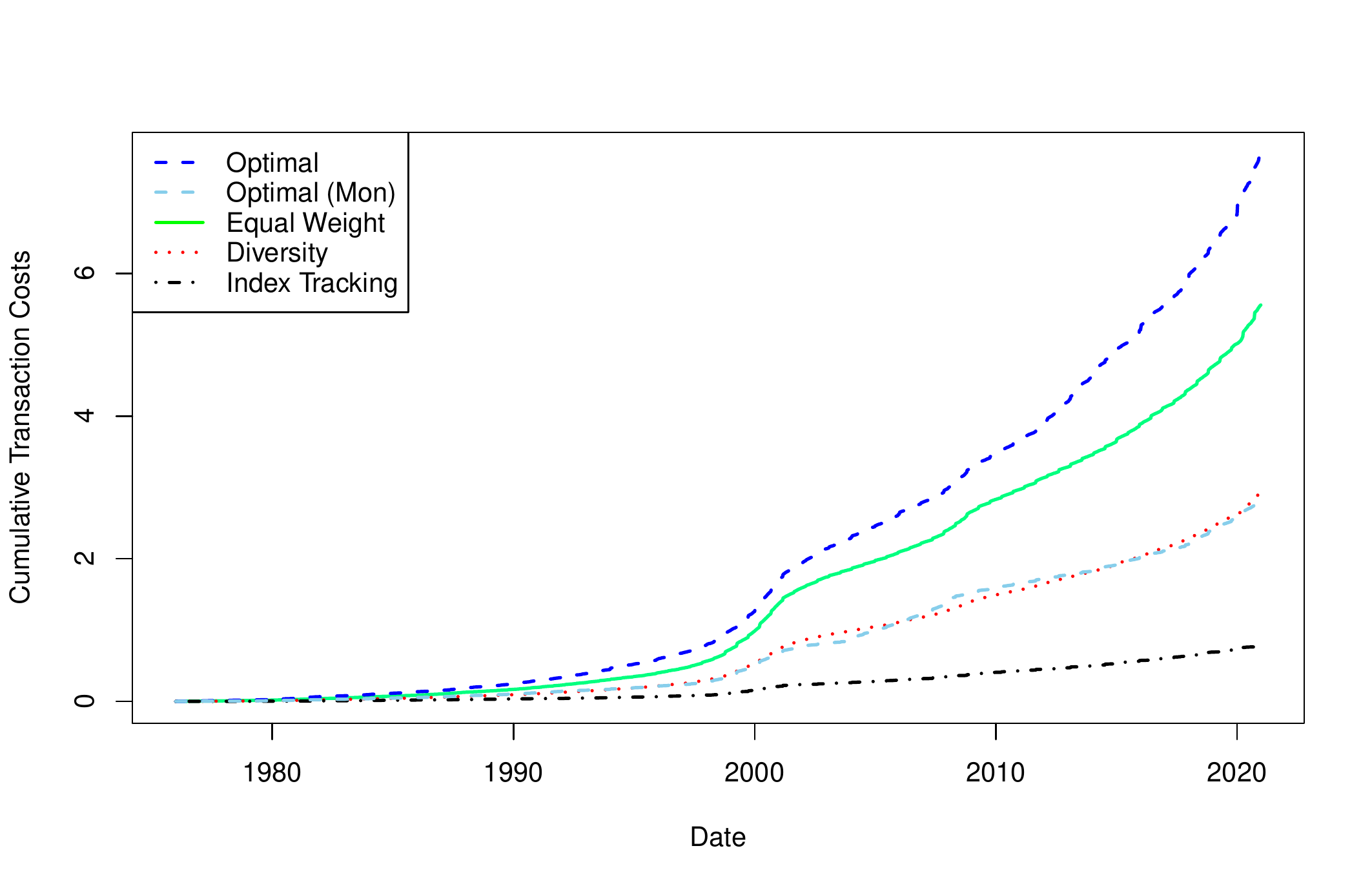}
    \end{center}
    \caption{Transaction costs incurred by the optimal portfolios from Figure \ref{fig:mkt.performance} when $tc=0.3\%$ as a percentage of the portfolio value (left). Cumulative transaction costs for all portfolios in Figure \ref{fig:mkt.performance} when $tc=0.3\%$ when the initial portfolio value is taken to be $1$ (right).}
    \label{fig:trans.costs}
\end{figure}

An illustration of the relative transaction costs is given in Figure \ref{fig:trans.costs}. We observe that a majority of the costs is due to changes in the constituent stocks. In this figure we can also see that the optimal monotone weight portfolio is less impacted than the unconstrained portfolio. In particular, there are once again similarities in the susceptibility of the monotone and diversity portfolios to transaction costs. A reference for the (nominal) value process (started at $1$) of the portfolios is given in Figure \ref{fig:div.and.value} where we also plot the time series of a measure of market diversity (defined by $D(\mathbf{\mu}) = (\sum_i \mu_i^{\theta})^{1/\theta}$ where $\theta = 0.5$).

\begin{figure}[h]
    \begin{center}
        \includegraphics[width=0.49\textwidth]{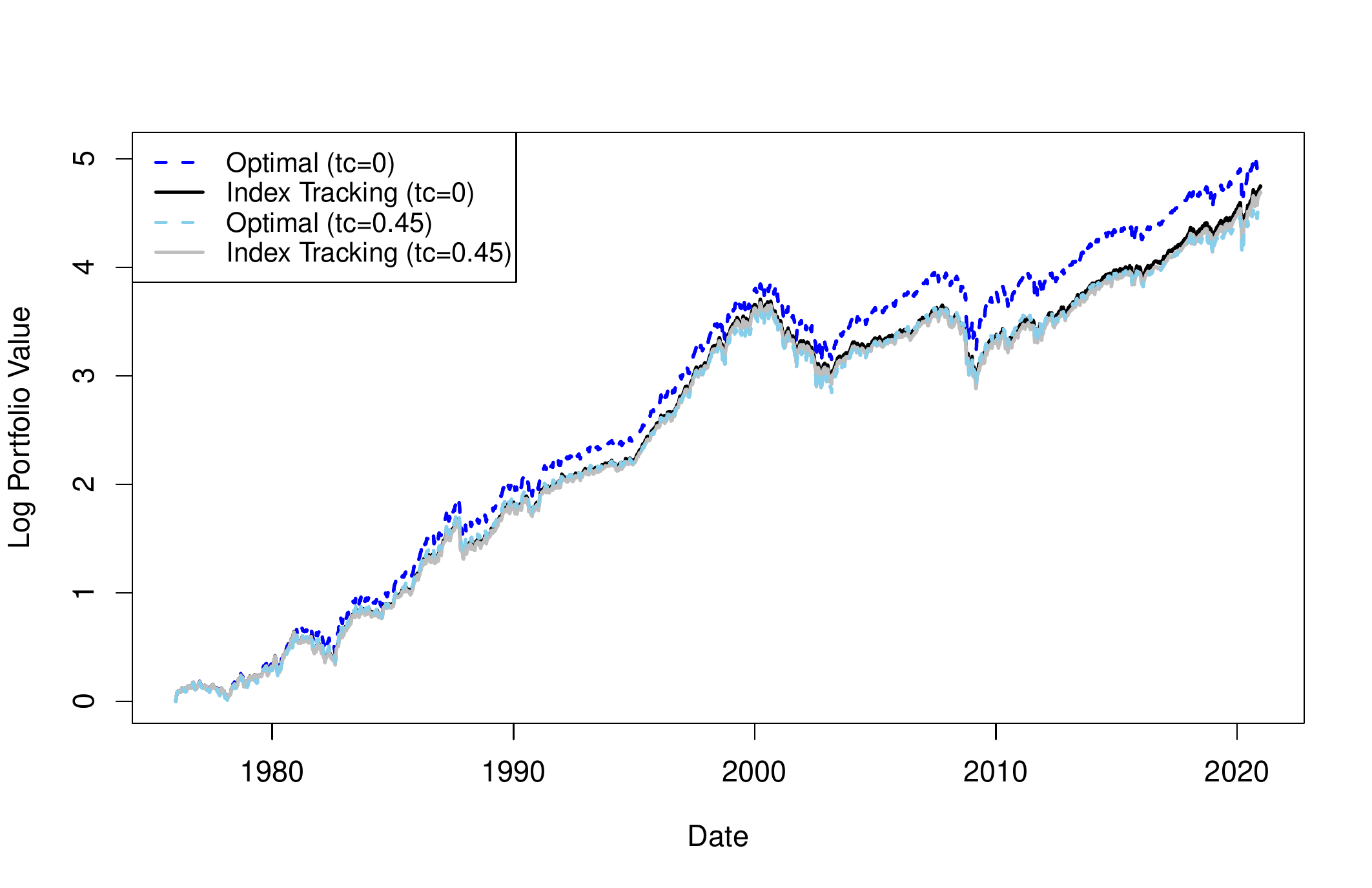}
        \includegraphics[width=0.49\textwidth]{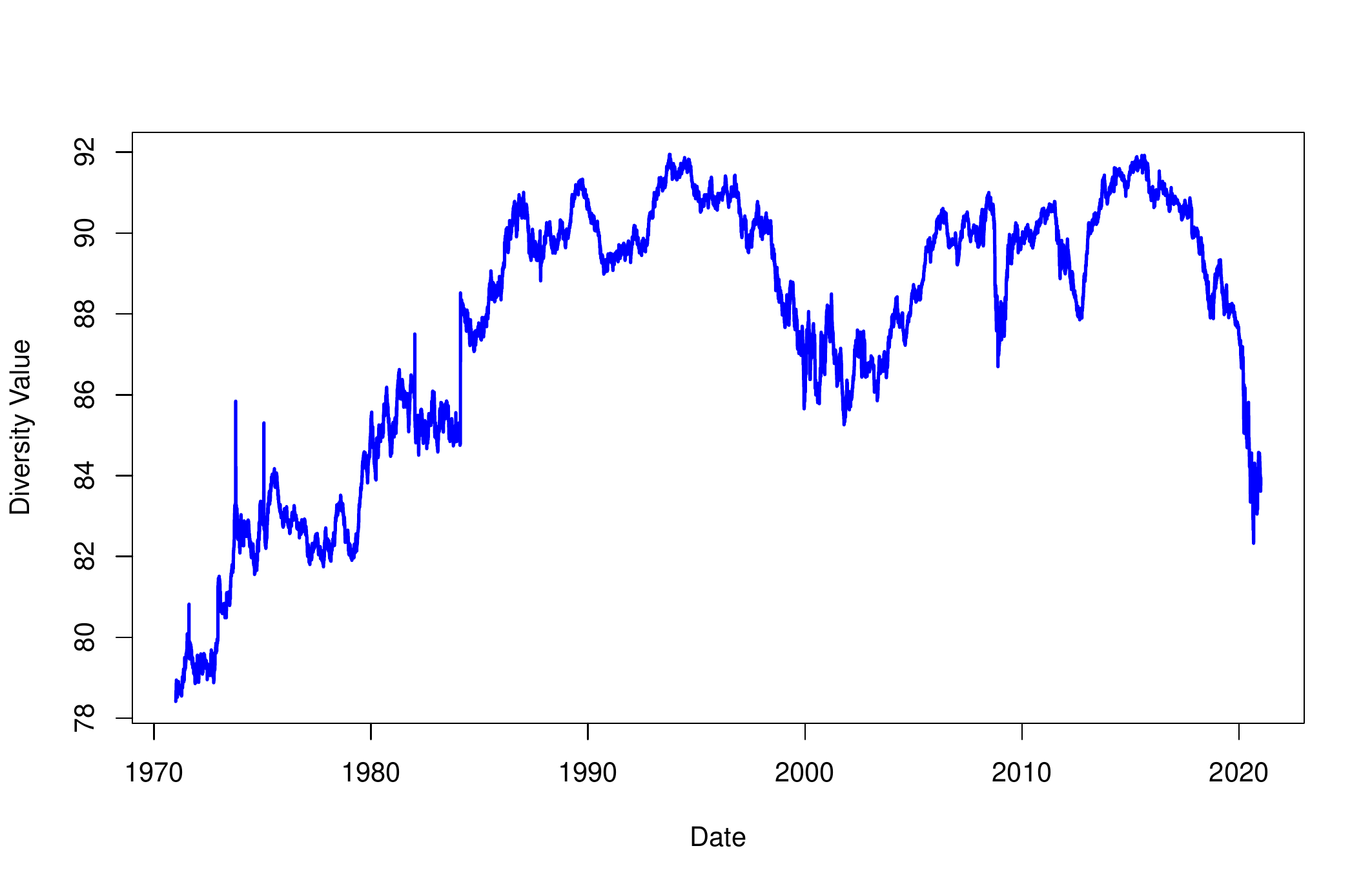}
    \end{center}
    \caption{Portfolio values (started at $1$) for the unconstrained and index tracking portfolios from Figure \ref{fig:mkt.performance} when $tc=0\%$ and $tc=0.45\%$ (left). Trajectory of market diversity for $\theta=1/2$ since $1971$ of the top 100 stocks ranked daily (right).}
    \label{fig:div.and.value}
\end{figure}

A separate investigation of market diversity indicates that one explanation for the underperformance of all portfolios (relative to the index tracking benchmark) after 2015 is that diversity has declined considerably in this latter part of the trading window. This is reflected in Figure \ref{fig:div.and.value}. This is in keeping with the findings of \cite{taljaard2021has} which documents the same phenomenon.  

\begin{figure}[h]
    \begin{center}
        \includegraphics[width=0.49\textwidth]{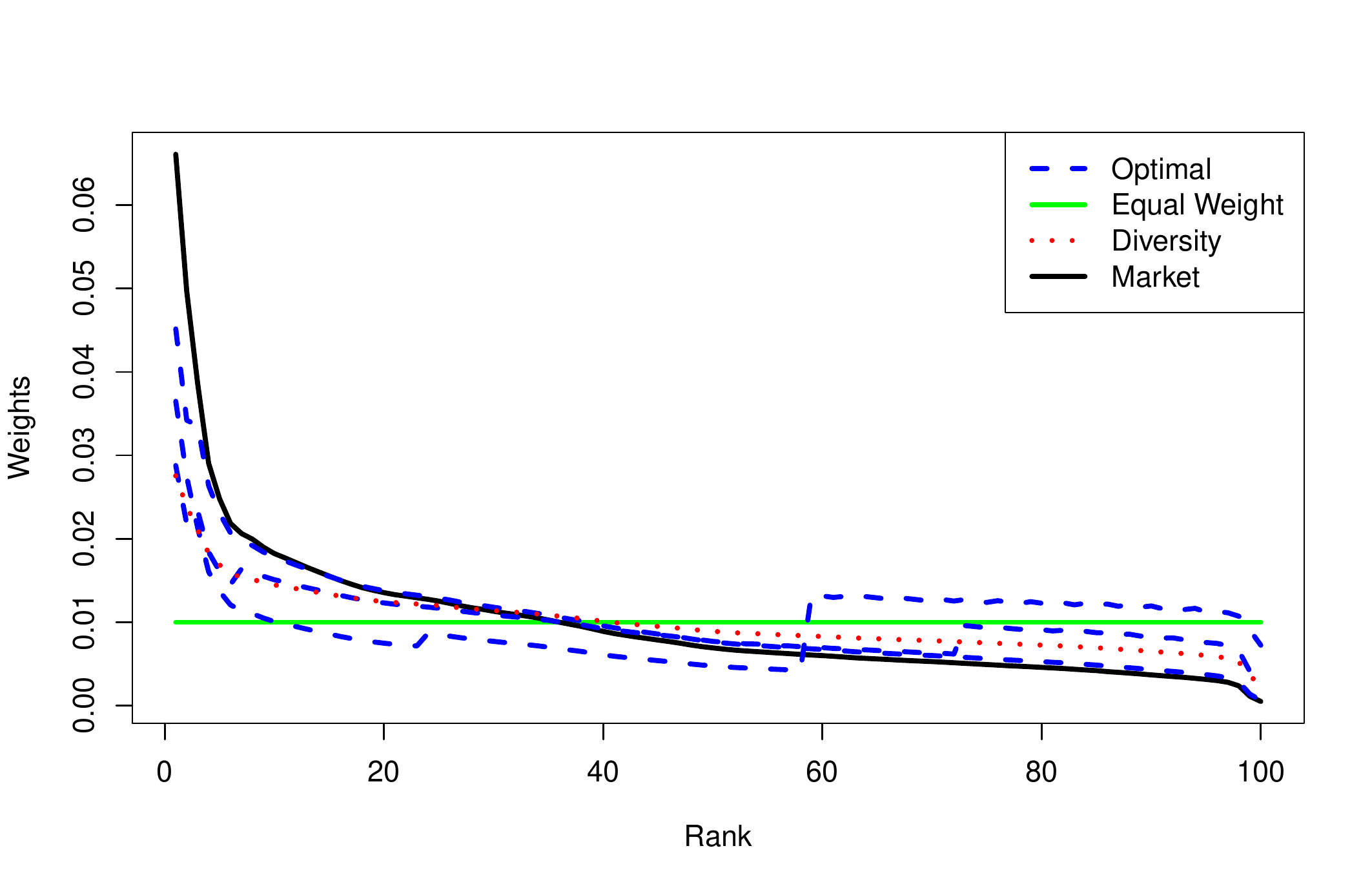}
        \includegraphics[width=0.49\textwidth]{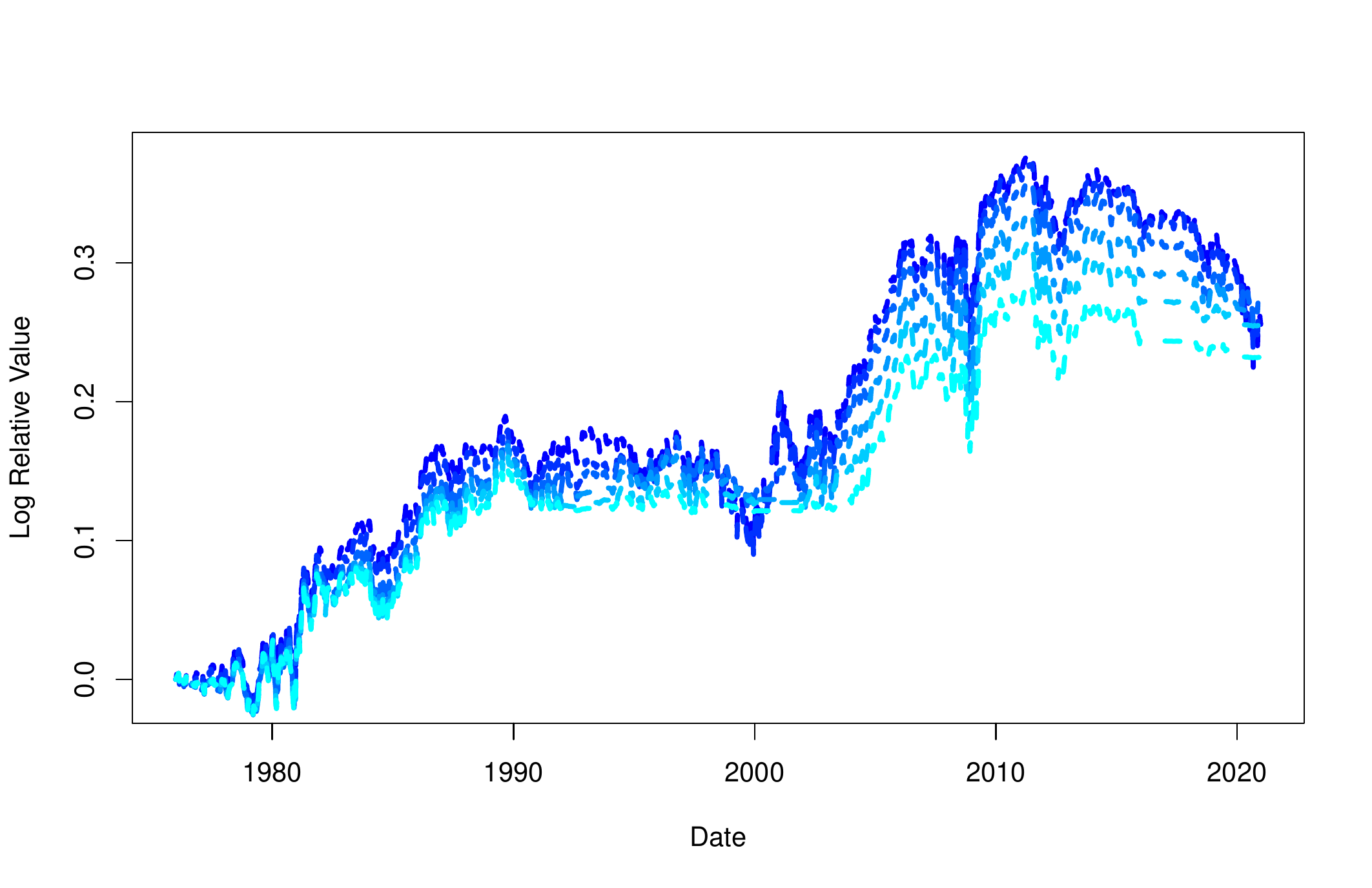}
    \end{center}
    \caption{Impact of varying $\eta_0\in\{-0.75,0,0.5\}$ in the period from 1986 to 1990 on the proposed portfolio (left) and varying $\eta_0\in\{-0.75,-0.5,-0.25,0,0.25,0.5\}$ in trading from 1976 to 2020 (right) with no transaction costs. The former is a period where the proposed portfolio for $\eta_0=0$ is less aggressive. In both cases, as $\eta_0$ decreases the plots show that the portfolio becomes more aggressive. Here $\eta_1=1$ and $\beta=10^8$ with no additional portfolio constraints.}
    \label{fig:varying.eta}
\end{figure}

We also take this opportunity to note the impact of fluctuations in diversity on the optimized portfolio in shorter time horizons. From Figure \ref{fig:div.and.value} we can see that there are extended periods where diversity is generally decreasing. This generally favors less aggressive portfolios in our optimization (when $\eta_0 > -1$ in \eqref{eqn:optim.objective}). As illustrated in Figure \ref{fig:varying.eta} the portfolio trained in one of these  periods (1986--1990) (for $\eta_0\geq0$) is very close to the market portfolio. This is unfavorable if diversity increases in the testing period. The portfolio manager can reduce the exposure to changes in diversity by decreasing $\eta_0$. From this figure we see that as $\eta_0$ decreases, the recommended portfolio moves away from the marker portfolio and becomes more aggressive. This general intuition is corroborated by the second image in Figure \ref{fig:varying.eta} which shows the overall relative performance of the unconstrained portfolio when $tc=0\%$ as $\eta_0$ is varied. In this example, the lower the value of $\eta_0$, the better the overall performance. While choosing a longer time horizon might allow the manager to capture longer term market cycles (in diversity, etc.) this comes at the cost of less recent and representative training data. Thus, this ability to change the parameters $\eta_0,\eta_1$ gives the portfolio manager an extra lever with which to improve finite sample performance.

Evidently, there are many questions raised by the open market setting. These include how to choose the training window, rebalancing frequency, and the modeling of market diversity. Overcoming these limitations and constructing portfolios whose performance is robust to these effects is naturally an interesting research direction.

\section{Conclusion}  \label{sec:conclusion}
In this paper we introduced an implementable optimization problem for a class of rank-based functionally generated portfolios and illustrated its use using empirical data. At a higher level, our work contributes to applications of exponentially concave functions in optimization. Our inquiry leads naturally to several directions which we describe below. 
\begin{enumerate}
	\item[(i)] Our optimization problem, as well as some of those cited in Section \ref{sec:intro}, rely crucially on the stability of the (rank-based) capital distribution. While there is a well-developed literature (see e.g.~\cite{chatterjee2010phase,pal2011analysis}) on rank-based models based on interacting particle systems, there is still a strong need for rigorous empirical and statistical analyses. In Section \ref{sec:rank.based.stability} we briefly considered the stability using concepts from optimal transport and we plan to carry out a deeper study in future research. The recent paper \cite{taljaard2021has} has also noted the potential utility of such a study.
	\item[(ii)] In Section \ref{sec:empirical} we find that there are myriad factors that can influence the performance of portfolios formed in the classic closed market SPT setting. In principle, these should be explicitly treated in portfolio construction.
	\item[(iii)] While the family $\mathcal{E}_{\beta}$ is mathematically and computationally tractable, it is natural to extend the framework to more general families of generating functions. In particular, an interesting problem is to study optimization of functionally generated portfolios that depend not only on the market weight process but also on additional processes such as volatility; see for example \cite{M21,X20} and the references therein. Additionally, the summation in the definition of $\varphi$ (see \eqref{eqn:varphi.decomposable}) may be generalized in the spirit of linear basis expansion to have a separate function $\ell_k$ for each rank. It is then natural to investigate the ``optimal'' balance of bias and variance. 
	\item[(iv)] In \cite{BW19,CSW19,W15b} the authors considered Cover's universal portfolio, a form of online learning algorithm, for functionally generated portfolios in the context of stochastic portfolio theory (also see \cite[Section 4.5]{karatzas2020open}). In these papers a major challenge is to obtain explicit finite time regret bounds. Will the function classes and constructions of this paper -- when combined with suitable learning algorithms -- lead to tractable results and explicit finite time bounds?  
\end{enumerate}

\appendix
\addcontentsline{toc}{section}{}
\section*{Appendix}
\stepcounter{section}

\subsection{General analytical results}

We first state and prove some elementary properties of exponentially concave functions on the unit interval.

\begin{lemma} \label{lem:exp.conc.property2}
If $\ell\in\mathcal{E}$ then for $x \in (0, 1)$ we have
\begin{equation} \label{eqn:ell.derivative.bound}
\frac{-1}{1-x}\leq \ell'(x)\leq \frac{1}{x}
\end{equation}
and
\begin{equation} \label{eqn:ell.global.bound}
\log(2-2x)\wedge\log(2x)\leq \ell(x)\leq \log(2-2x)\vee\log(2x).
\end{equation}
\end{lemma}
\begin{proof}
Let $x, y \in (0, 1)$. By exponential concavity of $\ell$, we have
\[
\ell(y) - \ell(x) \leq \log (1 + \ell'(x) (y - x)) \Rightarrow 1 + \ell'(x) (y - x) > 0. 
\]
Rearranging, we have $\ell'(x) < \frac{1}{x} + \frac{\ell'(x)}{x} y$. Letting $y \downarrow 0$ gives the upper bound in \eqref{eqn:ell.derivative.bound}. Note that we also have $\ell'(x) \frac{y - x}{1 - x} > \frac{-1}{1 - x}$. Letting  $y \uparrow 1$ gives the lower bound. 

To prove \eqref{eqn:ell.global.bound}, consider $f = e^{\ell}$ which is positive, concave and satisfies $f(\frac{1}{2}) = 1$. From the concavity of $f$ we have the elementary bound 
\[
(2 - 2x) \wedge 2x \leq f(x) \leq (2 - 2x) \vee 2x, \quad 0 \leq x \leq 1.
\]
Taking log gives the desired result. 	
\end{proof}

\begin{lemma} \label{lem:ell.upper.bound}
	Let $\ell \in \mathcal{E}_{\beta}$. Then the following statements hold for $x \in [0, 1]$.
	\begin{enumerate}
		\item[(i)] $|\ell'(x)| \leq \sqrt{\beta}$. In particular, $\ell$ is Lipschitz with constant $\sqrt{\beta}$.
		\item[(ii)] $|\ell(x)| \leq \sqrt{\beta} \left|x - \frac{1}{2} \right| \leq \frac{\sqrt{\beta}}{2}$.
	\end{enumerate}
\end{lemma}
\begin{proof}
Pick $x_0\in(0,1)$. By the monotonicity of $\ell'(x)$, $\ell''(x)$ exists almost everywhere on $(0,1)$. Hence, there exists a $c\in(0,x_0)$ and $\bar{c}\in(x_0,1)$ such that $\ell''(c)$ and $\ell''(\bar{c})$ exist. Moreover, by exponential concavity
$\ell'(c)\geq \ell'(x_0)\geq \ell'(\bar{c})$ and $(\ell'(c))^2+\ell''(c)\leq0$. Then by $\beta$-smoothness we have
\[(\ell'(c))^2\leq-\ell''(c)\leq \beta.\]
The same bound holds for $(\ell'(\bar{c}))^2$ so we conclude $|\ell'(x_0)|\leq \sqrt{\beta}$. Since $x_0$ was arbitrary the above bound holds on all of $(0,1)$. Since $\lim_{x\downarrow0}\ell'(x)$ and $\lim_{x\uparrow1}\ell'(x)$ exist and coincide with $\ell'(0)$ and $\ell'(1)$, respectively, the bound above extends to $[0,1]$. The second statement is immediate from (i).
\end{proof}

\subsection{Lemma \ref{lem:beta.interpretation}}
\begin{proof}[Proof of Lemma \ref{lem:beta.interpretation}]
Let $\mathbf{p}, \mathbf{q} \in \Delta_n$. By \eqref{eqn:pathwise.decomposition} and Lemma \ref{lem:ell.upper.bound}(i), we have
\[
\log \left(\sum_{i = 1}^n \boldsymbol{\pi}_i(\mathbf{p}) \frac{q_i}{p_i} \right) \geq \frac{1}{n} \sum_{i = 1}^n (\ell(q_i) - \ell(p_i)) \geq \frac{-1}{n} \sqrt{\beta} \|\mathbf{q} - \mathbf{p}\|_1.
\]
Sending $\mathbf{q}$ to $\mathbf{e}_i$, the $i$-th vertice of $\Delta_n$, we have $\log \frac{\boldsymbol{\pi}_i(\mathbf{p})}{p_i} \geq -\frac{2\sqrt{\beta}}{n}$. Exponentiating gives the lower bound in \eqref{eqn:weight.ratio.bound}.

Next we turn to the upper bound. By relabeling the coordinates if necessary, we may assume without loss of generality that $p_n = p_{(n)}$ where $p_{(n)}=\min_{1\leq i\leq n}p_i$. By \eqref{eqn:weight.ratio.monotone} and \eqref{eqn:fgp.weights}, we have
\[
\frac{\boldsymbol{\pi}_i(\mathbf{p})}{p_i} - 1 \leq \frac{\boldsymbol{\pi}_{\sigma(n)}(\mathbf{p})}{p_{\sigma(n)}} -1 =  \frac{1}{n}  \ell'(p_{(n)}) - \frac{1}{n}  \sum_{j = 1}^n p_j \ell'(p_j) = \frac{1}{n} \sum_{j = 1}^n p_j (\ell'(p_{(n)}) - \ell'(p_j)).
\]
Using $\beta$-smoothness of $\ell$, the last expression is bounded above by
\begin{equation} \label{eqn:weight.ratio.upper.bound.derivation}
\begin{split}
&\frac{\beta}{n} \sum_{j = 1}^n p_j (p_{(n)} - p_j)= \beta \left[ (1 - p_{(n)}) p_{(n)} - \sum_{j = 1}^{n-1} p_j^2 \right] \\
&\leq \frac{\beta}{n} \left[ (1 - p_{(n)}) p_{(n)} - \frac{(1 - p_{(n)})^2}{n-1}  \right] \leq \frac{\beta}{n} (1 - p_{(n)}) p_{(n)} \leq \frac{\beta}{n^2},
\end{split}
\end{equation}
where in the last inequality we used the fact that $p_{(n)} \leq \frac{1}{n}$. 
\end{proof}

\subsection{Lemma \ref{lem:density}}
\begin{proof}[Proof of Lemma \ref{lem:density}]
Fix $\ell\in\mathcal{E}$. For $\epsilon>0$ given, we will construct $\ell_{\epsilon}\in\bigcup_{\beta\geq0} \mathcal{E}_\beta$ such that $d(\ell,\ell_\epsilon)<\epsilon$. Let $f=e^{\ell}$. Since $\ell\in\mathcal{E}$, $f$ is concave, $f>0$, and $f$ is continuously differentiable with $f'=\ell'e^{\ell}$. Define the Bernstein operator $B_m$ for $m = 1, 2, \ldots$ by:
\[
B_m(f)(x)=\sum_{k=0}^nf(\frac{k}{m})\binom{m}{k}x^k(1-x)^{m-k}, \quad x\in[0,1],
\]
and consider the function $B_m(f)$. By \cite[Theorems 6.2.2 and 6.3.2]{davis1975interpolation}, $B_m(f)$ converges uniformly to $f$ and $B_m'(f)$ converges uniformly to $f'$ as $m \rightarrow \infty$. Moreover, by \cite[Theorem 6.3.3]{davis1975interpolation} $B_m(f)$ is concave. Since the approximation error of $B_m(f)$ is $0$ at the endpoints, we have that $B_m(f)>0$ for all $n$. Finally, since $e^{\ell}$ and $\ell'e^{\ell}$ are also continuous on $[0,1]$, we see that $\frac{B_m'(f)}{B_m(f)}$ converges uniformly to $\frac{\ell'e^{\ell}}{e^{\ell}}=\ell'$ on $[0,1]$. Note now that $(\log(B_m(f)))'=\frac{B_m'(f)}{B_m(f)}$. Hence for any $\epsilon>0$ we can find a corresponding $m_0$ such that for all $n\geq m_0$ we have $d(\log(B_m(f)),\ell)<\epsilon$. Fix such an $m \geq m_0$ and define $\ell_\epsilon:=\log(B_m(f))+a_m$, where $a_m$ is chosen so that $\ell_\epsilon(1/2)=0$. It remains to show that $\ell_\epsilon$ is exponentially concave whose derivative is Lipschitz. Clearly, $e^{\ell_\epsilon}$ is concave since it is a positive multiple of $B_m(f)$ which is concave. Moreover, there exists a constant $K_\epsilon$ such that:
\[
|\ell''_\epsilon|=|(\log(B_m(f)))''|=\left|\frac{B''_m(f)B_m(f)-(B'_m(f))^2}{(B_m(f))^2}\right|\leq K_\epsilon, \quad  x\in[0,1].
\]
The above follows from continuity on $[0,1]$ since the numerator and denominator are continuous with the latter strictly positive. This ensures the Lipschitz property of the derivative and taken together we conclude $\ell_\epsilon\in\bigcup_{\beta\geq0} \mathcal{E}_\beta$.
\end{proof}

\subsection{Lemma \ref{lem:log.value.estimate}}
\begin{proof}[Proof of Lemma \ref{lem:log.value.estimate}]
Using \eqref{eqn:fgp.weights} and \eqref{eqn:L.p.r}, we have that
\[
\frac{\boldsymbol{\pi}(\mathbf{u}) \cdot \mathbf{r}}{\mathbf{u} \cdot \mathbf{r}} = 1 + \frac{1}{n} \sum_{i, j = 1}^n \frac{u_i r_i}{\mathbf{u} \cdot \mathbf{r}} (\ell'(u_i) - \ell'(u_j)).
\]
By Lemma \ref{lem:beta.interpretation}, this quantity always lies in the interval $\left[e^{-\frac{2\sqrt{\beta}}{n}},1+\frac{\beta}{n^2}\right]$. Now let $\ell, \tilde{\ell} \in \mathcal{E}_{\beta}$ and let $\boldsymbol{\pi}$ and $\tilde{\boldsymbol{\pi}}$ be the corresponding portfolios. Since the map $x \mapsto \log (x)$ is Lipschitz on this interval with constant $e^{\frac{2\sqrt{\beta}}{n}} $, we have
\begin{equation*}
\begin{split}
&|\mathcal{L}(\mathbf{u},\mathbf{r} ; \ell) - \mathcal{L}(\mathbf{u},\mathbf{r} ; \tilde\ell)|\\
&\leq \frac{1}{n} e^{\frac{2\sqrt{\beta}}{n}} \left| \sum_{i, j = 1}^n \frac{u_i r_i}{\mathbf{u} \cdot \mathbf{r}}  [(\ell'(u_i) - \tilde{\ell}'(u_i)) - (\ell'(u_j) - \tilde{\ell}'(u_j))] \right| \\
&\leq \frac{1}{n} e^{\frac{2\sqrt{\beta}}{n}} \sum_{i, j = 1}^n \frac{u_i r_i}{\mathbf{u} \cdot \mathbf{r}} 2 d (\ell, \tilde{\ell}) = 2 e^{\frac{2\sqrt{\beta}}{n}}  d(\ell, \tilde{\ell}).
\end{split}
\end{equation*} 

\end{proof}

\subsection{Lemma \ref{lem:L.lip}}

The following result will be used in the proof of Lemma \ref{lem:L.lip}.

\begin{lemma}\label{lem:hilbert.log.dot.product}
	Let $\mathbf{x},\mathbf{y},\mathbf{x}',\mathbf{y}'\in\Delta_n$. Then
	\[
	\left|\log\left(\frac{\mathbf{x}\cdot \mathbf{y}}{\mathbf{x}'\cdot \mathbf{y}'}\right)\right|\leq d_\mathcal{H}(\mathbf{x},\mathbf{x}')+d_\mathcal{H}(\mathbf{y},\mathbf{y}').
	\]
\end{lemma}

\begin{proof}
	Since the components of $\mathbf{x},\mathbf{y},\mathbf{x}',\mathbf{y}'$ are positive, we have
	\[
	\min_i\frac{x_i}{x_i'}\min_i\frac{y_i}{y_i'}\leq \min_i\frac{x_iy_i}{x_i'y_i'}\leq \frac{\sum_ix_iy_i}{\sum_ix_i'y_i'}\leq \max_i\frac{x_iy_i}{x_i'y_i'}\leq \max_i\frac{x_i}{x_i'}\max_i\frac{y_i}{y_i'}.
	\]
	Now since $\mathbf{x},\mathbf{x}'\in\Delta_n$, we have $\min_i\frac{x_i}{x'_i}\leq 1\leq \max_i\frac{x_i}{x'_i}$. As a consequence
	\[\max_i\frac{x_i}{x_i'}\max_i\frac{y_i}{y_i'}\leq \max_i\frac{x_i}{x_i'}\max_i\frac{x_i'}{x_i}\max_i\frac{y_i}{y_i'}\max_i\frac{y_i'}{y_i}\]
	and the analogous result holds for the lower bound. Since  $\min_i\frac{x_i}{x_i'}=\left(\max_i\frac{x_i'}{x_i}\right)^{-1}$, we have
	\[
	\left(\max_i\frac{x_i}{x_i'}\max_i\frac{x_i'}{x_i}\max_i\frac{y_i}{y_i'}\max_i\frac{y_i'}{y_i}\right)^{-1}\leq\frac{x\cdot y}{x'\cdot y'}\leq\max_i\frac{x_i}{x_i'}\max_i\frac{x_i'}{x_i}\max_i\frac{y_i}{y_i'}\max_i\frac{y_i'}{y_i}.
	\]
	The proof is completed by taking the logarithm.
\end{proof}

\begin{proof}[Proof of Lemma \ref{lem:L.lip}]
Let $(\mathbf{u}, \mathbf{r}), (\mathbf{u}', \mathbf{r}') \in \Delta_{n,\geq} \times \Delta_n$. We begin with
\begin{equation*}
\begin{split}
| \mathcal{L} (\mathbf{u}, \mathbf{r} ; \ell) - \mathcal{L} (\mathbf{u}', \mathbf{r}';\ell)| &= \left| \log \left( \frac{\boldsymbol{\pi}(\mathbf{u}) \cdot \mathbf{r}}{\mathbf{u} \cdot \mathbf{r}}\right) - \log \left( \frac{\boldsymbol{\pi}(\mathbf{u}') \cdot \mathbf{r}'}{\mathbf{u}'\cdot \mathbf{r}'}\right) \right| \\
  &\leq \left|  \log \left( \frac{\boldsymbol{\pi}(\mathbf{u}) \cdot \mathbf{r}}{\boldsymbol{\pi}(\mathbf{u}') \cdot \mathbf{r}'}  \right) \right| + \left| \log \left( \frac{\mathbf{u}\cdot \mathbf{r}}{\mathbf{u}' \cdot \mathbf{r}'} \right) \right| \\
  &\leq \left|  \log \left( \frac{\boldsymbol{\pi}(\mathbf{u}) \cdot \mathbf{r}}{\boldsymbol{\pi}(\mathbf{u}') \cdot \mathbf{r}'}  \right) \right|+d_\mathcal{H}(\mathbf{r},\mathbf{r}')+ d_\mathcal{H}(\mathbf{u},\mathbf{u}'),
\end{split}
\end{equation*}
where the last inequality follows from Lemma \ref{lem:hilbert.log.dot.product}. To deal with the first term we apply the log-sum inequality:
\begin{equation*}
\begin{split}
\log \left( \frac{\boldsymbol{\pi}(\mathbf{u}) \cdot \mathbf{r}}{\boldsymbol{\pi}(\mathbf{u}' \cdot \mathbf{r}'} \right) &\leq \frac{1}{\boldsymbol{\pi}(\mathbf{u}) \cdot \mathbf{r}} \sum_{k = 1}^n \boldsymbol{\pi}_k(\mathbf{u}) r_k \log \frac{\boldsymbol{\pi}_k(\mathbf{u}) r_k}{\boldsymbol{\pi}_k(\mathbf{u}') r_k'} \\  &\leq \sum_{k = 1}^n \frac{\boldsymbol{\pi}_k(\mathbf{u}) r_k}{\boldsymbol{\pi}(\mathbf{u}) \cdot \mathbf{r}} \left( | \log \boldsymbol{\pi}_k(\mathbf{u}) - \log \boldsymbol{\pi}_k(\mathbf{u}')  | + |\log r_k - \log r_k'| \right) \\
&\leq \| \log \boldsymbol{\pi}(\mathbf{u}) - \log \boldsymbol{\pi}(\mathbf{u}') \|_{\infty} + d_\mathcal{H}(\mathbf{r},\mathbf{r}'),
\end{split}
\end{equation*}
where $\log$ is applied to each component. A symmetric argument gives the corresponding lower bound. To see the last equality note that
\[
|\log r_i - \log r'_i|  \leq \log\max_i\frac{r_i}{r'_i}\max_j\frac{r'_j}{r_j}=d_\mathcal{H}(\mathbf{r},\mathbf{r}').
\]

Now fix $k$ and consider
\[
|\log \boldsymbol{\pi}_k (\mathbf{u}) - \log \boldsymbol{\pi}_k (\mathbf{u}')| \leq \left| \log \frac{\boldsymbol{\pi}_k (\mathbf{u})}{u_k} - \log \frac{\boldsymbol{\pi}_k (\mathbf{u}')}{u_k'} \right| + | \log u_k - \log u_k'|.
\]
For the latter term we again note that $\left|\log u_i - \log u_i'\right|\leq d_\mathcal{H}(\mathbf{u},\mathbf{u}')$. Using the Lipschitz continuity argument as in the proof of Lemma \ref{lem:log.value.estimate}, we have
\begin{equation*}
\begin{split}
&\left| \log \frac{\boldsymbol{\pi}_k (\mathbf{u})}{u_k} - \log \frac{\boldsymbol{\pi}_k (\mathbf{u}')}{u_k'} \right| \\
&\leq e^{\frac{2\sqrt{\beta}}{n}} \left| \frac{\boldsymbol{\pi}_k (\mathbf{u})}{u_k} - \frac{\boldsymbol{\pi}_k (\mathbf{u}')}{u_k'} \right| \\ &\leq \frac{1}{n}e^{\frac{2\sqrt{\beta}}{n}} \left( \left| \ell'(u_k) - \ell'(u_k') \right| + \sum_{m = 1}^n |u_m \ell'(u_m) - u_m' \ell'(u_m')| \right). 
\end{split}
\end{equation*}
Since $\ell$ is $\beta$-smooth, we have
\[
|\ell'(u_k) - \ell'(u_k')| \leq \beta | u_k - u_k' |  \leq \beta \|\mathbf{u} - \mathbf{u}'\|_1.
\]
Using Lemma \ref{lem:ell.upper.bound}, we can estimate the other term by
\begin{equation*}
\begin{split}
\sum_{m = 1}^n |u_m \ell'(u_m) - u_m' \ell'(u_m')| &\leq \sum_{m = 1}^n \left( u_m |\ell'(u_m) - \ell'(u_m')| + |\ell'(u_m')| |u_m - u_m' | \right) \\
  &\leq \beta \| \mathbf{u} - \mathbf{u}'\|_1 + \sqrt{\beta} \| \mathbf{u} - \mathbf{u}'\|_1.
\end{split}
\end{equation*}
Gathering all pieces together, we have
\[
| \mathcal{L} (\mathbf{u}, \mathbf{r} ; \ell) - \mathcal{L} (\mathbf{u}', \mathbf{r}'; \ell)| \leq 2\tilde{d}(\mathbf{r},\mathbf{r}') +
\left(2+\frac{\sqrt{\beta}+2\beta}{n}e^{\frac{2\sqrt{\beta}}{n}}\right) \tilde{d}(\mathbf{u},\mathbf{u}'),
\]
from which we may take $K_1(\beta,n):=2+\frac{\sqrt{\beta}+2\beta}{n}e^{\frac{2\sqrt{\beta}}{n}}$. 

Next, recall $\overline{\mathbf{e}}=(1/n,\dots,1/n)^\top$ and for $\ell(\mathbf{x}):=(\ell(x_1),\dots,\ell(x_n))^\top$ we have that
\begin{align*}
    \left|\mathbf{D}_\ell[\mathbf{u}\oplus\mathbf{r}: \mathbf{r}] - \mathbf{D}_\ell[\mathbf{u}'\oplus\mathbf{r}': \mathbf{r}']\right|&\leq \left|\overline{\mathbf{e}}\cdot\left(\ell(\mathbf{r}\oplus\mathbf{u})-\ell(\mathbf{r}'\oplus\mathbf{u}')\right)\right|+\left|\overline{\mathbf{e}}\cdot\left(\ell(\mathbf{u}')-\ell(\mathbf{u})\right)\right|\\
    &\leq \left|\overline{\mathbf{e}}\cdot\left(\ell(\mathbf{r}\oplus\mathbf{u})-\ell(\mathbf{r}'\oplus\mathbf{u}')\right)\right|+\frac{\sqrt{\beta}}{n}||\mathbf{u}-\mathbf{u}'||_1,
\end{align*}
where the last inequality follows from Lemma \ref{lem:ell.upper.bound}. Isolating the first term in the above and letting $\varphi_{\ell}$ be the exponentially concave function on $\Delta_n$ defined by $\ell$ we have
\begin{align*}
    \overline{\mathbf{e}}\cdot\left(\ell(\mathbf{r}\oplus\mathbf{u})-\ell(\mathbf{r}'\oplus\mathbf{u}')\right)&=\varphi_\ell(\mathbf{r}\oplus\mathbf{u})-\varphi_\ell(\mathbf{r}'\oplus\mathbf{u}')\\
    &\leq \log\left(1+\nabla\varphi_\ell(\mathbf{r}'\oplus\mathbf{u}')\cdot (\mathbf{r}\oplus\mathbf{u}-\mathbf{r}'\oplus\mathbf{u}')\right)\\
    &=\log\left(\sum_{i=1}^n\boldsymbol{\pi}_i(\mathbf{r}'\oplus\mathbf{u}')\frac{(\mathbf{r}\oplus\mathbf{u})_i}{(\mathbf{r}'\oplus\mathbf{u}')_i}\right),
\end{align*}
where $(\cdot)_i$ selects the $i$th index of the vectors. Here we have used that $\mathbf{r}\oplus\mathbf{u},\mathbf{r}'\oplus\mathbf{u}'\in\Delta_n$, the non-negativity of \eqref{eqn:L.divergence}, and \eqref{eqn:rank.based.return}. Now,
\begin{align*}
    \log\left(\sum_{i=1}^n\boldsymbol{\pi}_i(\mathbf{r}'\oplus\mathbf{u}')\frac{(\mathbf{r}\oplus\mathbf{u})_i}{(\mathbf{r}'\oplus\mathbf{u}')_i}\right)&=\log\left(\sum_{i=1}^n\boldsymbol{\pi}_i(\mathbf{r}'\oplus\mathbf{u}')\frac{r_iu_i}{r_i'u_i'}\frac{\mathbf{r}'\cdot\mathbf{u}'}{\mathbf{r}\cdot\mathbf{u}}\right)\\
    &= \log\left(\sum_{i=1}^n\boldsymbol{\pi}_i(\mathbf{r}'\oplus\mathbf{u}')\frac{r_iu_i}{r_i'u_i'}\right)+\log\left(\frac{\mathbf{r}'\cdot\mathbf{u}'}{\mathbf{r}\cdot\mathbf{u}}\right)\\
    &\leq \log\left(\max_i\frac{r_iu_i}{r_i'u_i'}\right)+\log\left(\frac{\mathbf{r}'\cdot\mathbf{u}'}{\mathbf{r}\cdot\mathbf{u}}\right)\\
    &\leq\log\left(\max_i\frac{r_i}{r_i'}\max_j\frac{r_j'}{r_j}\max_k\frac{u_k}{u_k'}\max_l\frac{u_l'}{u_l}\right)\\
    & \quad \quad \quad \quad \quad \quad \quad \quad+ d_\mathcal{H}(\mathbf{r},\mathbf{r}') +d_\mathcal{H}(\mathbf{u},\mathbf{u}')\\
    &=2d_\mathcal{H}(\mathbf{r},\mathbf{r}')+2d_\mathcal{H}(\mathbf{u},\mathbf{u}').
\end{align*}
The first inequality follows from the monotonicity of $\log(\cdot)$ and bounding the average using the positive weights $\boldsymbol{\pi}_i(\cdot)$ from above with the maximum. The second inequality follows by an application of
Lemma \ref{lem:hilbert.log.dot.product} and the arguments therein. A symmetric argument gives the identical lower bound
\begin{align*}\overline{\mathbf{e}}\cdot\left(\ell(\mathbf{r}\oplus\mathbf{u})-\ell(\mathbf{r}'\oplus\mathbf{u}')\right)&=\varphi_\ell(\mathbf{r}\oplus\mathbf{u})-\varphi_\ell(\mathbf{r}'\oplus\mathbf{u}')\\
&=-(\varphi_\ell(\mathbf{r}'\oplus\mathbf{u}')-\varphi_\ell(\mathbf{r}\oplus\mathbf{u}))\\
&\geq-\log\left(\sum_{i=1}^n\boldsymbol{\pi}_i(\mathbf{r}\oplus\mathbf{u})\frac{(\mathbf{r}'\oplus\mathbf{u}')_i}{(\mathbf{r}\oplus\mathbf{u})_i}\right)\\
&\geq-2d_\mathcal{H}(\mathbf{r},\mathbf{r}')-2d_\mathcal{H}(\mathbf{u},\mathbf{u}').
\end{align*}
Collecting terms as before allows us to take $K_0(\beta,n):=2+\frac{\sqrt{\beta}}{n}$.
\end{proof}

\subsection{Proposition \ref{prop:construction}}
To prove Proposition \ref{prop:construction} we will need several preliminary results. We begin by introducing some notations and constructions. Recall $\mathcal{P}$ is a partition (containing $\frac{1}{2}$) with mesh size $\delta=\max_i|x_{i+1}-x_i|$ and minimal mesh size $\underline{\delta}:=\min_i|x_{i+1}-x_i|$. Throughout we make the assumption that there exists a fixed $M>0$ such that $|\delta-\underline{\delta}|\leq M\underline{\delta}^3$. When needed we will refer to the length of the $i$th interval as $\delta_i:=|x_{i+1}-x_{i}|$ for $i=1,...,d-1$. Let $\hat{\boldsymbol{\ell}} : [0, 1] \rightarrow \mathbb{R}$ be a concave piecewise affine function over $\mathcal{P}$ (i.e., $\hat{\boldsymbol{\ell}}|_{[x_i, x_{i + 1}]}$ is affine for all $i$), such that the vector $\boldsymbol{\ell}=(\boldsymbol{\ell}_i)_{i=1}^d=(\hat{\boldsymbol{\ell}}(x_i))_{i=1}^d$ satisfies the constraints of Problem \ref{prob:discretized} on $\mathcal{P} = \{x_i\}_{i = 1}^d$. We define $\boldsymbol{f}$ by $\boldsymbol{f}=e^{\boldsymbol{\ell}} = (e^{\boldsymbol{\ell}_i})$ and let $f$ be the concave and piecewise affine function that interpolates $\boldsymbol{f}$ over $\mathcal{P}$. Note that $f$,  $\log(f)$ and $\hat{\boldsymbol{\ell}}$ are differentiable except possibly at the mesh points. At the mesh points $x_i$ for concreteness we define the derivative to be the right derivative for $i < d$ and the left derivative for $i = d$. 

Additionally, given $\boldsymbol{f}$ we define the piecewise quadratic function $s(x)$ by
\begin{equation}\label{eqn:s.def}
\begin{split}
s(x) &:= \frac{f(x_i)+ f(x_i-\underline{\delta})}{2} + \frac{f(x_i) - f(x_i-\underline{\delta})}{\underline{\delta}} \left(x - \frac{(x_i-\underline{\delta}) + x_i}{2}\right) + \\
&\quad + \frac{f(x_i+\underline{\delta}) - 2f(x_i) +f(x_i-\underline{\delta})}{2\underline{\delta}^2} \left(x - \frac{(x_i-\underline{\delta})+ x_i}{2}\right)^2, \quad x \in I_i,
\end{split}
\end{equation}
and
\[
I_i:=\left[\frac{(x_{i}-\underline{\delta})+x_i}{2},\frac{x_{i}+(x_{i}+\underline{\delta})}{2}\right]=\left[x_{i}-\frac{\underline{\delta}}{2},x_{i}+\frac{\underline{\delta}}{2}\right],
\]
for $i=2,...,d-1$ and $s(x)=f(x)$ otherwise (i.e. it is linear on the end intervals $[0,x_2-\frac{\underline{\delta}}{2}],[x_{d-1}+\frac{\underline{\delta}}{2},1]$ and the intermediate intervals $J_i:=(x_{i}+\frac{\underline{\delta}}{2},x_{i+1}-\frac{\underline{\delta}}{2})$ for $i=2...,d-2$). It is easy to verify that $s$ is $C^1$, concave and matches the function value and derivative of $f$ at the endpoints of $I_i$. Moreover, $s(x)=f(x)$ on the intermediate intervals $J_i$ not accounted for by $I_i$, and the end intervals. As a result, it is also strictly positive. Note here that $J_i=\emptyset$ if and only if $|x_{i+1}-x_i|=\underline{\delta}$.

Where appropriate we will make the dependence on $\mathcal{P}$ explicit by including a subscript (as in $\boldsymbol{\ell}_{\mathcal{P}}$ or $s_{\mathcal{P}}(x)$) for the functions and vectors above. Finally, for $0<\alpha<1$ we define the $\ell_{\alpha,\mathcal{P}}$ (candidate for Proposition \ref{prop:construction}) on a given partition $\mathcal{P}$ by 
\[
\ell_{\alpha,\mathcal{P}}(x) =\log(s_{\mathcal{P}}(x)+\underline{\delta}^\alpha)+a_{\mathcal{P}},
\]
where $a_{\mathcal{P}}$ is a constant chosen so that $\ell_{\alpha,\mathcal{P}}(1/2)=0$. Clearly $\ell_{\alpha, \mathcal{P}}$ is exponentially concave. As will be seen below, the term $\underline{\delta}^{\alpha}$ is added to ensure $\beta$-smoothness when $\underline{\delta}$ is sufficiently small. Much of the analysis below will require estimates of the derivatives of $\log(f(x))$, $\log(s(x))$ and $\log(s(x)+\underline{\delta}^\alpha)$. For this we will make regular use of the following standard Taylor expansion as well as Lemma \ref{lem:vector.bounds}.

\begin{lemma} For $z\in\mathbb{R}$ and fixed $k_0 \in\mathbb{N}$ there exists a $\xi\in[-|z|,|z|]$ such that
\[e^z=\sum_{k=0}^n\frac{z^k}{k!}+\frac{e^{\xi}}{(n+1)!} z^{n+1}.\]
Similarly, for $r\geq1$ and $|z|<1$ there exists a $\xi\in[-|z|,|z|]$ such that
\[\frac{1}{(1+z)^r}=1-rz+\frac{(r+1)r}{2(1+\xi)^{r+2}}z^2. \]
\end{lemma}

To be concise these tedious but elementary computations will be omitted where possible. We will also occasionally mention that derivatives (and their bounds) can be extended to one-sided derivatives at mesh points. Where applicable, the justification for these statements is the following two standard lemmas.

\begin{lemma} \label{lem:extension.pt1}
For a function $f$ if $f^{(k)}$ is defined on $[x_0,x_1]$ with $f^{(k)}$ differentiable on the open interval $(x_0,x_1)$ then
$\lim_{x\downarrow x_0}f^{(k+1)}(x)=f^{(k+1)}_+(x_0)$
and 
$\lim_{x\uparrow x_1}f^{(k+1)}(x)=f^{(k+1)}_-(x_1)$ if the former limits exist. Here $f^{(k+1)}_\pm(x)$ denotes the one-sided $(k+1)$-th derivatives at $x$.
\end{lemma}

\begin{lemma}\label{lem:extension.pt2}
If a function $f$ defined on $[x_0,x_1]$ has a bounded derivative of order $k+1$ on the open interval $(x_0,x_1)$ then
$\lim_{x\downarrow x_0}f^{(k)}(x)$
and 
$\lim_{x\uparrow x_1}f^{(k)}(x)$ exist.
\end{lemma}

\begin{proof}
By the bound on $f^{(k+1)}(x)$ we have that $f^{(k)}(x)$ is Lipshitz on $(x_0,x_1)$ and hence, uniformly continuous on the open interval. The uniform continuity allows us to attain a unique extension of $f^{(k)}$ to $[x_0,x_1]$.
\end{proof}

Finally, we will be able to extend smoothness properties where only the left/right second derivatives exist by the following.

\begin{lemma}\label{lem:beta.smooth.twicediff.a.e} If $\ell$ is concave and continuously differentiable on $[0,1]$ with $\ell''_+(x)\geq -\beta$ and $\ell''_-(x)\geq -\beta$ for all $x$ then $\ell$ is $\beta$-smooth.
\end{lemma}

\begin{proof}
Since $\ell$ is concave, $\ell'$ is monotone decreasing. We also have by assumption that $\ell'$ is continuous on $[0,1]$. The rest of the proof proceeds along standard lines. Namely, we choose any $0\leq x_0<x_1\leq 1$ and define $h'(x)=\frac{\ell'(x_1)-\ell'(x_0)}{x_1-x_0}x-\ell'(x)$. Since $h'(x_0)=h'(x_1)$, a generalization of Rolle's theorem gives the result after rearranging.
\end{proof}

We will now establish an approximation error result that will readily translate to our analysis of $\ell_{\alpha,\mathcal{P}}$.

\begin{lemma}\label{lem:approx.of.l.by.salpha}
There exists a $K>0$ depending only on $\beta$ and $M$ such that for $0<\alpha<1$ and any ${\underline{\delta}}>0$
\begin{equation}\label{eqn:approx.logsalpha.l}
    \left|(\log(s_{\mathcal{P}}(x)+\underline{\delta}^\alpha))'-\hat{\boldsymbol{\ell}}'_\mathcal{P}(x)\right|\leq K\underline{\delta}^\alpha.
\end{equation}
\end{lemma}

\begin{proof}
In what follows the subscript $\mathcal{P}$ will be omitted and it is understood that the functions are all defined with respect to the same fixed partition $\mathcal{P}$. 

We now build up the approximation error \eqref{eqn:approx.logsalpha.l} in parts. We first aim to show 
\begin{equation}\label{eqn:approx.logf.l}|(\log(f(x)))'-\hat{\boldsymbol{\ell}}'(x)|\leq \frac{1}{2}e^{\sqrt{\beta}}\beta(M+1)\underline{\delta}, \ \ \ x\in[0,1]. 
\end{equation}
Note that on $(x_i,x_{i+1})$ we have $\hat{\boldsymbol{\ell}}'(x)=\frac{\boldsymbol{\ell}_{i+1}-\boldsymbol{\ell}_i}{x_{i+1}-x_i}$. By our constraints $f(x)$ is concave and then so is $\log(f(x))$. Hence the derivatives on $(x_i,x_{i+1})$ satisfy
\[(\log(f(x_i)))'_+\geq (\log(f(x)))'\geq (\log(f(x_{i+1}))'_-, \ \ \ x\in(x_i,x_{i+1}).\]
Now by Taylor's remainder theorem there is a $\xi_1\in[-|\boldsymbol{\ell}_{i+1}-\boldsymbol{\ell}_i|,|\boldsymbol{\ell}_{i+1}-\boldsymbol{\ell}_i|]$ such that
\[(\log(f(x_i)))'_+=\frac{e^{\boldsymbol{\ell}_{i+1}-\boldsymbol{\ell}_i}-1}{x_{i+1}-x_i}=\frac{\boldsymbol{\ell}_{i+1}-\boldsymbol{\ell}_i}{x_{i+1}-x_i}+\frac{1}{2}e^{\xi_1}\frac{(\boldsymbol{\ell}_{i+1}-\boldsymbol{\ell}_i)^2}{x_{i+1}-x_i},\]
and a $\xi_2\in[-|\boldsymbol{\ell}_{i+1}-\boldsymbol{\ell}_i|,|\boldsymbol{\ell}_{i+1}-\boldsymbol{\ell}_i|]$ such that
\[(\log(f(x_{i+1})))'_-=\frac{1-e^{-(\boldsymbol{\ell}_{i+1}-\boldsymbol{\ell}_i)}}{x_{i+1}-x_i}=\frac{\boldsymbol{\ell}_{i+1}-\boldsymbol{\ell}_i}{x_{i+1}-x_i}-\frac{1}{2}e^{\xi_2}\frac{(\boldsymbol{\ell}_{i+1}-\boldsymbol{\ell}_i)^2}{x_{i+1}-x_i}.\]
Hence for $x\in(x_i,x_{i+1})$, we have
\[|(\log(f(x)))'-\hat{\boldsymbol{\ell}}'(x)|\leq \frac{1}{2}e^\xi\frac{(\boldsymbol{\ell}_{i+1}-\boldsymbol{\ell}_i)^2}{x_{i+1}-x_i}, \ \ \ \mathrm{for \ some} \   \xi\in[-|\boldsymbol{\ell}_{i+1}-\boldsymbol{\ell}_i|,|\boldsymbol{\ell}_{i+1}-\boldsymbol{\ell}_i|].\]
By Lemma \ref{lem:vector.bounds}, $|\boldsymbol{\ell}_{i+1}-\boldsymbol{\ell}_i|\leq\sqrt{\beta}\delta_i$ and
\[
|(\log(f(x)))'-\hat{\boldsymbol{\ell}}'(x)|\leq \frac{1}{2}e^{\sqrt{\beta}\delta_i}\beta\delta_i\leq \frac{1}{2}e^{\sqrt{\beta}}\beta(M+1)\underline{\delta}, \ \ \ x\in(x_i,x_{i+1}),
\]
where we have used $\delta_i\leq 1$. Given our choice of derivative at the mesh points the above applies to the entire interval $[x_i,x_{i+1}]$ and since the arguments were independent of $i$ we get that this extends to $[0,1]$ as claimed in \eqref{eqn:approx.logf.l}.

We next want to show
\begin{equation} \label{eqn:approx.s.f}
\|s' - f'\|_{\infty} \leq 2\beta(1+e^{\sqrt{\beta}})(M+1)\underline{\delta} \quad \text{and} \quad \|s - f\|_{\infty} \leq \beta(1+e^{\sqrt{\beta}})(M+1)\underline{\delta}^2.
\end{equation}
Since $f=s$ on $J_i$ and the end intervals it suffices to consider the intervals $I_i$ to show \eqref{eqn:approx.s.f}. First let us note that by concavity and Taylor's remainder theorem alongside Lemma \ref{lem:vector.bounds} and our constraints on $\boldsymbol{\ell}$ that we have:
\begin{equation}\label{eqn:f.der.diff.estimate}
0\geq\frac{\frac{f_{i+1}-f_i}{\delta_{i}}-\frac{f_{i}-f_{i-1}}{\delta_{i-1}}}{\frac{\delta_{i}+\delta_{i-1}}{2}}\geq -2\beta(1+e^{\sqrt{\beta}}).
\end{equation}
Let $I_i$ be fixed and $x \in I_i$. By concavity and the enforced equality of $s'$ and $f'$ at the endpoints, we have
\begin{equation*}
\begin{split}
s'(x) - f'(x) &\leq s'\left(x_i-\frac{\underline{\delta}}{2}\right) - f'\left(x_i+\frac{\underline{\delta}}{2}\right) \\
  &= f'\left(x_i-\frac{\underline{\delta}}{2}\right) - f'\left(x_i+\frac{\underline{\delta}}{2}\right) \\
  &\leq 2\beta(1+e^{\sqrt{\beta}})(M+1)\underline{\delta}.
\end{split}
\end{equation*}
where the last inequality is from \eqref{eqn:f.der.diff.estimate} and the bound on $\delta_i,\delta_{i-1}$.
By a similar argument on the lower bound, we have $\|s' - f'\|_{\infty} \leq 2\beta(1+e^{\sqrt{\beta}})(M+1)\underline{\delta}$. Finally, for $x \in I_i$, let $z_i$ be an endpoint of $I_i$ such that $|x - z_i| \leq \frac{\delta}{2}$. Since $s(z_i) = f(z_i)$ by construction we can use our first estimate to get the second estimate in \eqref{eqn:approx.s.f} by integration.

Our third claim is
\begin{equation}\label{eqn:approx.logs.logf}|(\log(s(x)))'-(\log(f(x)))'|\leq 2\beta(1+e^{\sqrt{\beta}})(M+1)\left( e^{\frac{\sqrt{\beta}}{2}}\underline{\delta}+\sqrt{\beta}e^{2\sqrt{\beta}}\underline{\delta}^2\right).
\end{equation}
If $x\not\in\mathcal{P}$ then:
\begin{align*}
    |(\log(s(x)))'-(\log(f(x)))'|&=\frac{|s'(x)f(x)-f'(x)s(x)|}{s(x)f(x)}\\
    &\leq\frac{|s'(x)-f'(x)|}{s(x)}+\frac{|f'(x)||f(x)-s(x)|}{s(x)f(x)}\\
    &\leq \frac{2\beta(1+e^{\sqrt{\beta}})(M+1)\underline{\delta}}{s(x)}+\frac{|f'(x)|\beta(1+e^{\sqrt{\beta}})(M+1)\underline{\delta}^2}{s(x)f(x)}.
\end{align*}
Now since $f(x)$ attains a minimum at $x_1=0$ or $x_d=1$ by concavity and $f(x)=s(x)$ at the endpoints we have by $f(x_i)=e^{\boldsymbol{\ell}_i}\in \left[e^{-\frac{\sqrt{\beta}}{2}},2\right]$ (see Lemma \ref{lem:vector.bounds}) that:
\begin{equation}\label{eqn:partial.approx.logs.logf}
    |(\log(s(x)))'-(\log(f(x)))'|\leq (M+1) \beta(1+e^{\sqrt{\beta}}) \left(2e^{\frac{\sqrt{\beta}}{2}}\underline{\delta}+|f'(x)|e^{\sqrt{\beta}}\underline{\delta}^2\right).
\end{equation}
We now treat the term involving $f'(x)$. We have that for a given $x\not\in\mathcal{P}$ that there exists an $i$ such that
\[f'(x)=\frac{e^{\boldsymbol{\ell}_{i+1}}-e^{\boldsymbol{\ell}_{i}}}{x_{i+1}-x_i}.\]
By Taylor's remainder theorem and Lemma \ref{lem:vector.bounds} we then get
\begin{equation}\label{eqn:diff.f.bd}
|f'(x)|\leq 2e^{\sqrt{\beta}\delta_i}\sqrt{\beta}\leq 2e^{\sqrt{\beta}}\sqrt{\beta},
\end{equation}
which gives \eqref{eqn:approx.logs.logf} on $\mathcal{P}$ in light of \eqref{eqn:partial.approx.logs.logf}. The above is extended to the mesh points $x\in\mathcal{P}$ by noting that the left (right) limits of the above derivatives exist and so this will coincide with the left (right) derivative at the mesh points (see Lemmas \ref{lem:extension.pt1} and \ref{lem:extension.pt2}). Finally we show
\begin{equation}\label{eqn:approx.logsalpha.logs}
\begin{split}
    |(\log(s(x)+\underline{\delta}^\alpha))'&-(\log(s(x)))'|\\
    &\leq 2\sqrt{\beta}e^{2\sqrt{\beta}}\left(\underline{\delta}^\alpha+\sqrt{\beta}(1+e^{-\sqrt{\beta}})(M+1)\underline{\delta}^{1+\alpha}\right).
\end{split}
\end{equation}
By direct computation we have
\[|(\log(s(x)+\underline{\delta}^\alpha))'-(\log(s(x)))'|=\left|\frac{s'(x)}{(s(x)+\underline{\delta}^\alpha)s(x)}\right|\underline{\delta}^\alpha\leq \left|\frac{s'(x)}{s(x)^2}\right|\underline{\delta}^\alpha\]
since $s(x)>0$. Arguing as above, by Lemma \ref{lem:vector.bounds} we have that
\[|(\log(s(x)+\underline{\delta}^\alpha))'-(\log(s(x)))'|\leq e^{\sqrt{\beta}}\left|s'(x)\right|\underline{\delta}^\alpha .\]
It remains to bound $s'(x)$. By \eqref{eqn:diff.f.bd} and \eqref{eqn:approx.s.f} we have
\[|s'(x)|\leq 2e^{\sqrt{\beta}}\sqrt{\beta}+2\beta(1+e^{\sqrt{\beta}})(M+1)\underline{\delta},\]
so taking the above together gives \eqref{eqn:approx.logsalpha.logs} after collecting terms. To complete the proof of \eqref{eqn:approx.logsalpha.l} we can collect the estimates \eqref{eqn:approx.logf.l}, \eqref{eqn:approx.logs.logf}, \eqref{eqn:approx.logsalpha.logs} and apply the triangle inequality after using the fact that if $k\geq 1$, $\underline{\delta}^\alpha\geq \underline{\delta}^k$ for $\underline{\delta}\in[0,1]$ since $0<\alpha<1$. 

\end{proof}

We will now introduce a function $g=\log(f)$ and get some error estimates that will be used later in Lemma \ref{lem:interior.est} by writing the form of $s(x)$ given in \eqref{eqn:s.def} equivalently in terms of $e^g$.

\begin{lemma}\label{lem:g.approx.errors}
Let $g=\log(f)$. If $|\delta-\underline{\delta}|\leq M\underline{\delta}^{3}$ then
\[\left|\frac{g(x_i+\underline{\delta})-g(x_i)}{\underline{\delta}}-\frac{\boldsymbol{\ell}_{i+1}-\boldsymbol{\ell}_i}{\delta_{i}}\right| \leq \left(\sqrt{\beta}+ 2\sqrt{\beta}e^{\frac{3\sqrt{\beta}}{2}}\right)M\underline{\delta}^{2},\]
\[\left|\frac{g(x_i)-g(x_i-\underline{\delta})}{\underline{\delta}}-\frac{\boldsymbol{\ell}_{i}-\boldsymbol{\ell}_{i-1}}{\delta_{i-1}}\right| \leq\left(\sqrt{\beta}+ 2\sqrt{\beta}e^{\frac{3\sqrt{\beta}}{2}}\right)M\underline{\delta}^{2}\]
and
\[\left|\frac{g(x_i+\underline{\delta})-2g(x_i)+g(x_i-\underline{\delta})}{\underline{\delta}^2}-\frac{\frac{\boldsymbol{\ell}_{i+1}-\boldsymbol{\ell}_i}{\delta_{i}}-\frac{\boldsymbol{\ell}_{i}-\boldsymbol{\ell}_{i-1}}{\delta_{i-1}}}{\frac{\delta_{i}+\delta_{i-1}}{2}}\right| \leq\beta M\underline{\delta}^{2}+2\sqrt{\beta}\left(1+ 2e^{\frac{3\sqrt{\beta}}{2}}\right)M\underline{\delta}.\]
\end{lemma}
\begin{proof}
For the first two estimates we treat the forward difference and the backward difference is similar. Since $\boldsymbol{\ell}_{i}=\log(f(x_i))=g(x_i)$ we have:
\[\left|g(x_i+\underline{\delta})-g(x_i)-(\boldsymbol{\ell}_{i+1}-\boldsymbol{\ell}_i)\right|=\left|g(x_i+\underline{\delta})-\boldsymbol{\ell}_{i+1}\right|\leq e^{\frac{\sqrt{\beta}}{2}}|f(x_i+\underline{\delta})-f(x_{i+1})|\]
since $f\in [e^{-\frac{\sqrt{\beta}}{2}},2]$ (see Lemma \ref{lem:vector.bounds}) and $\log(\cdot)$ is Lipschitz on this domain. Then by estimate \eqref{eqn:diff.f.bd} we get
\[e^{\frac{\sqrt{\beta}}{2}}|f(x_i+\underline{\delta})-f(x_{i+1})|\leq2\sqrt{\beta}e^{\frac{3\sqrt{\beta}}{2}}|\underline{\delta}-\delta_i|\leq 2\sqrt{\beta}e^{\frac{3\sqrt{\beta}}{2}}M\underline{\delta}^{3}.\]
Hence,
\begin{align*}
    \left|\frac{g(x_i+\underline{\delta})-g(x_i)}{\underline{\delta}}-\frac{\boldsymbol{\ell}_{i+1}-\boldsymbol{\ell}_i}{\delta_{i}}\right|&\leq\left|\frac{\boldsymbol{\ell}_{i+1}-\boldsymbol{\ell}_i}{\underline{\delta}}-\frac{\boldsymbol{\ell}_{i+1}-\boldsymbol{\ell}_i}{\delta_{i}}\right|+ 2\sqrt{\beta}e^{\frac{3\sqrt{\beta}}{2}}M\underline{\delta}^{2}\\
    &\leq|\boldsymbol{\ell}_{i+1}-\boldsymbol{\ell}_i|\left|\frac{\delta_i-\underline{\delta}}{\underline{\delta}\delta_i}\right|+ 2\sqrt{\beta}e^{\frac{3\sqrt{\beta}}{2}}M\underline{\delta}^{2}\\
    &\leq \sqrt{\beta}M\underline{\delta}^{2}+ 2\sqrt{\beta}e^{\frac{3\sqrt{\beta}}{2}}M\underline{\delta}^{2}.
\end{align*}
For the third estimate we have by using the earlier estimates and our constraints on $\boldsymbol{\ell}$ that
\begin{align*}
&\left|\frac{g(x_i+\underline{\delta})-2g(x_i)+g(x_i-\underline{\delta})}{\underline{\delta}^2}-\frac{\frac{\boldsymbol{\ell}_{i+1}-\boldsymbol{\ell}_i}{\delta_{i}}-\frac{\boldsymbol{\ell}_{i}-\boldsymbol{\ell}_{i-1}}{\delta_{i-1}}}{\frac{\delta_{i}+\delta_{i-1}}{2}}\right|\\
&\quad \quad \leq \left|\frac{\boldsymbol{\ell}_{i+1}-\boldsymbol{\ell}_i}{\delta_{i}}-\frac{\boldsymbol{\ell}_{i}-\boldsymbol{\ell}_{i-1}}{\delta_{i-1}}\right|\left|\frac{\delta_i+\delta_{i-1}-2\underline{\delta}}{\underline{\delta}(\delta_i+\delta_{i-1})}\right|+2\left(\sqrt{\beta}+ 2\sqrt{\beta}e^{\frac{3\sqrt{\beta}}{2}}\right)M\underline{\delta}\\
&\quad \quad \leq \beta M\underline{\delta}^{2}+2\left(\sqrt{\beta}+ 2\sqrt{\beta}e^{\frac{3\sqrt{\beta}}{2}}\right)M\underline{\delta}.
\end{align*}
\end{proof}

Next, we turn to establishing $\beta$-smoothness of $\log(s_{\mathcal{P}}(x)+\underline{\delta}^\alpha)$ on the interior of the intervals $I_i$, $J_i$ and the interior of the end intervals when the minimal grid spacing is sufficiently small. We begin by getting estimates of $s_{\mathcal{P}}(x),s_{\mathcal{P}}'(x)$ and $s_{\mathcal{P}}''(x)$ (the latter of which exists on the interior of these intervals).

\begin{lemma} \label{lem:interior.est}
Consider a partition $\mathcal{P}$ and any of the intervals $I_i$. Let $g=\log(f)$ and define $\mathbf{g}^i_{+}=g(x_{i}+\underline{\delta})$, $\mathbf{g}^i_0=g(x_{i})$, and $\mathbf{g}^i_{-}=g(x_{i}-\underline{\delta})$. For any $x\in \mathring{I}_i$ we reparameterize $x=x_{i}+(k-1)\underline{\delta}$ for $k\in(\frac{1}{2},\frac{3}{2})$. With this we have:
\begin{equation*}
\begin{split}
s(x_{i}+(k-1)\underline{\delta})&=e^{\mathbf{g}^i_{-}}\left(1+\mathcal{H}_{i,0}(k,\underline{\delta})\right),\\
s'(x_{i}+(k-1)\underline{\delta})&=e^{\mathbf{g}^i_{-}}\bigg(\frac{\left(\mathbf{g}^i_{0}-\mathbf{g}^i_{-}\right)}{\underline{\delta}}+\mathcal{H}_{i,1}(k,\underline{\delta})\bigg)\\
&=e^{\mathbf{g}^i_{-}}\bigg(\frac{\left(\boldsymbol{\ell}_i-\boldsymbol{\ell}_{i-1}\right)}{\delta_{i-1}}+\mathcal{R}_{i,1}(k,\underline{\delta})\bigg),\\
s''(x_{i}+(k-1)\underline{\delta})&=e^{\mathbf{g}^i_{-}}\left(\frac{\mathbf{g}^i_{+}-2\mathbf{g}^i_{0}+\mathbf{g}^i_{-}}{\underline{\delta}^2}+\left(\frac{\mathbf{g}^i_{0}-\mathbf{g}^i_{-}}{\underline{\delta}}\right)^2+\mathcal{H}_{i,2}(k,\underline{\delta})\right)\\
&=e^{\mathbf{g}^i_{-}}\left(\frac{\frac{\boldsymbol{\ell}_{i+1}-\boldsymbol{\ell}_i}{\delta_{i}}-\frac{\boldsymbol{\ell}_{i}-\boldsymbol{\ell}_{i-1}}{\delta_{i-1}}}{\frac{\delta_{i}+\delta_{i-1}}{2}}+\left(\frac{\boldsymbol{\ell}_i-\boldsymbol{\ell}_{i-1}}{\delta_{i-1}}\right)^2+\mathcal{R}_{i,2}(k,\underline{\delta})\right)
\end{split}
\end{equation*}
for functions $\mathcal{H}_{i,j}:\left(\frac{1}{2},\frac{3}{2}\right)\times [0,1]\mapsto \mathbb{R}$, $\mathcal{R}_{i,j}:\left(\frac{1}{2},\frac{3}{2}\right)\times [0,1]\mapsto \mathbb{R}$  $j=0,1,2$ satisfying:
\begin{equation*}
\begin{split}
|\mathcal{R}_{i,0}(k,\underline{\delta})|&=|\mathcal{H}_{i,0}(k,\underline{\delta})|\leq \left(\frac{3}{4}e^{\sqrt{\beta}}+e^{2\sqrt{\beta}}\right)\sqrt{\beta}\underline{\delta}, \\
|\mathcal{H}_{i,1}(k,\underline{\delta})|&\leq \left(\beta\left[1+\frac{1}{2}e^{\sqrt{\beta}}+2e^{2\sqrt{\beta}}+M\right]+2\sqrt{\beta}\left[1+2e^{\frac{3}{2}\sqrt{\beta}}\right]M\right)\underline{\delta}, \\
|\mathcal{R}_{i,1}(k,\underline{\delta})|&\leq|\mathcal{H}_{i,1}(k,\underline{\delta})|+\left(\sqrt{\beta}+ 2\sqrt{\beta}e^{\frac{3\sqrt{\beta}}{2}}\right)M\underline{\delta}^{2},\\
|\mathcal{H}_{i,2}(k,\underline{\delta})|&\leq \left(\frac{1}{2}A(\beta,M)^2+2\sqrt{\beta}A(\beta,M)+\frac{1}{6}e^{2\sqrt{\beta}}\beta^{\frac{3}{2}}+\frac{1}{3}e^{\sqrt{\beta}}\beta^{\frac{3}{2}}\right)\underline{\delta},\\
|\mathcal{R}_{i,2}(k,\underline{\delta})|&\leq|\mathcal{H}_{i,2}(k,\underline{\delta})|+\beta M\underline{\delta}^{2}+2\sqrt{\beta}\left(1+ 2e^{\frac{3\sqrt{\beta}}{2}}\right)M\underline{\delta}\\
& \quad + 2\sqrt{\beta}\left(\sqrt{\beta}+ 2\sqrt{\beta}e^{\frac{3\sqrt{\beta}}{2}}\right)M\underline{\delta}^{2} +\left(\sqrt{\beta}+ 2\sqrt{\beta}e^{\frac{3\sqrt{\beta}}{2}}\right)^2M^2\underline{\delta}^{4}.
\end{split}
\end{equation*}
for all $i$ and $k\in\left(\frac{1}{2},\frac{3}{2}\right)$ where $A(\beta,M):=\beta +\beta M+2\left(\sqrt{\beta}+ 2\sqrt{\beta}e^{\frac{3\sqrt{\beta}}{2}}\right)M$.
\end{lemma}

\begin{proof}
Taylor's remainder theorem alongside Lemma \ref{lem:vector.bounds} and Lemma \ref{lem:g.approx.errors}.
\end{proof}

Next we similarly address the two endpoint intervals and $J_i$ where $s(x)$ is linear. On these intervals $s(x)$ takes the form
\[s(x)=\frac{e^{\boldsymbol{\ell}_{i+1}}-e^{\boldsymbol{\ell}_i}}{\delta_i}(x-x_i)+e^{\boldsymbol{\ell}_i}\]
for some $i=1,...,d-1$ where $x\in J_i$ or $x$ is in the end intervals if $i=1,d-1$.

\begin{lemma}
Consider a partition $\mathcal{P}$. We reparameterize $x=x_i+k\delta_i$ for $k\in(0,1)$. With this we have:
\begin{equation*}
\begin{split}
s\left(x_i+k\delta_i\right) &=e^{\boldsymbol{\ell}_i}\left(1+\mathcal{R}_{i,4}(k,\underline{\delta})\right), \\
s'\left(x_i+k\delta_i\right) &=e^{\boldsymbol{\ell}_i}\left(\left(\frac{\boldsymbol{\ell}_{i+1}-\boldsymbol{\ell}_i}{\delta_i}\right)+\mathcal{R}_{i,5}(k,\underline{\delta})\right),
\end{split}
\end{equation*}
for $i=1,d-1$ and functions $\mathcal{R}_{i,j}:\left(0,1\right)\times [0,1]\mapsto \mathbb{R}$,  $j=4,5$ satisfying:
\[
|\mathcal{R}_{i,4}(k,\underline{\delta})|\leq e^{\sqrt{\beta}}\sqrt{\beta}(\underline{\delta}+M\underline{\delta}^{3}), \quad |\mathcal{R}_{i,5}(k,\underline{\delta})|\leq \frac{1}{2}e^{\sqrt{\beta}}\beta(\underline{\delta}+M\underline{\delta}^{3}),
\]
for all $i$ and $k\in(0,1)$.
\end{lemma}

\begin{proof}
Taylor's remainder theorem alongside Lemma \ref{lem:vector.bounds}.
\end{proof}

With this we establish the claimed $\beta$-smoothness on the interior of these intervals.

\begin{lemma}\label{lem:interior.estimate}
Consider a partition $\mathcal{P}$ and the (open) interior intervals $\mathring{I}_i$. We have that there exists a $\delta_{1,\alpha}>0$ depending only on $\alpha$, $\beta$ and $M$ such that if $\underline{\delta}<\delta_{1,\alpha}$ then for any $i$ and any $x\in \mathring{I}_i$ we have that:
\[0\geq (\log(s_{\mathcal{P}}(x)+\underline{\delta}^\alpha))''\geq -\beta.\]
Similarly, for the (open) end intervals $\left(0,x_2-\frac{\underline{\delta}}{2}\right)$ and $\left(x_{d-1}+\frac{\underline{\delta}}{2},1\right)$ and $J_i$, $i=2,...,d-1$, there exists a $\delta_{2,\alpha}>0$ depending only on $\alpha$, $\beta$ and $M$ such that if $\underline{\delta}<\delta_{2,\alpha}$ then for any $x$ in these intervals the above estimate also holds.
\end{lemma}

\begin{proof}
We will treat here the former case as the latter is argued analogously. Note first that
\[(\log(s_{\mathcal{P}}(x)+{\underline{\delta}}^\alpha))''=\frac{s_{\mathcal{P}}''(x)}{s_{\mathcal{P}}(x)+{\underline{\delta}}^\alpha}-\left(\frac{s_{\mathcal{P}}'(x)}{s_{\mathcal{P}}(x)+{\underline{\delta}}^\alpha}\right)^2\]
and by concavity it suffices to show the lower bound. Reparameterizing $x=x_{i}+(k-1){\underline{\delta}}$ for $k\in(\frac{1}{2},\frac{3}{2})$ so that $x\in\left(x_{i}-\frac{\underline{\delta}}{2},x_{i}+\frac{\underline{\delta}}{2}\right)$ we have using Lemma \ref{lem:interior.est}:
\begin{align*}
    &(\log(s_{\mathcal{P}}(x)+\underline{\delta}^\alpha))'' \\&=\frac{\frac{\frac{\boldsymbol{\ell}_{i+1}-\boldsymbol{\ell}_i}{\delta_{i}}-\frac{\boldsymbol{\ell}_{i}-\boldsymbol{\ell}_{i-1}}{\delta_{i-1}}}{\frac{\delta_{i}+\delta_{i-1}}{2}}+\left(\frac{\boldsymbol{\ell}_i-\boldsymbol{\ell}_{i-1}}{\delta_{i-1}}\right)^2+\mathcal{R}_{i,2}(k,{\underline{\delta}})}{1+\mathcal{R}_{i,0}(k,{\underline{\delta}})+e^{-\mathbf{g}^i_{-}}{\underline{\delta}}^\alpha} -\left(\frac{\frac{\left(\boldsymbol{\ell}_i-\boldsymbol{\ell}_{i-1}\right)}{\delta_{i-1}}+\mathcal{R}_{i,1}(k,{\underline{\delta}})}{1+\mathcal{R}_{i,0}(k,{\underline{\delta}})+e^{-\mathbf{g}^i_{-}}{\underline{\delta}}^\alpha}\right)^2.
\end{align*}
Note Lemma \ref{lem:vector.bounds} gives $\frac{1}{2}\leq e^{-\mathbf{g}^i_{-}}\leq e^{\frac{\sqrt{\beta}}{2}}$.  So picking $\delta_{0,\alpha}$ so that if $\underline{\delta}<\delta_{0,\alpha}$ \[\left(\frac{3}{4}e^{\sqrt{\beta}}+e^{2\sqrt{\beta}}\right)\sqrt{\beta}\underline{\delta}+e^{\frac{\sqrt{\beta}}{2}}\underline{\delta}^\alpha< 1/2.\]
We get by Lemma \ref{lem:interior.est} that $\left(\mathcal{R}_{i,0}(k,\underline{\delta})+e^{-\mathbf{g}^i_{-}}\underline{\delta}^\alpha\right)\in(-\frac{1}{2},\frac{1}{2})$ for all $i$ and $k$. Then, by Taylor's remainder theorem we can get
\begin{align*}
    (\log(s_{\mathcal{P}}(x)+\underline{\delta}^\alpha))''&=\frac{\frac{\boldsymbol{\ell}_{i+1}-\boldsymbol{\ell}_i}{\delta_{i}}-\frac{\boldsymbol{\ell}_{i}-\boldsymbol{\ell}_{i-1}}{\delta_{i-1}}}{\frac{\delta_{i}+\delta_{i-1}}{2}}\left(1-e^{-\mathbf{g}^i_{-}}\underline{\delta}^\alpha\right)\\
    & \ \ \ \ + \left(\frac{\boldsymbol{\ell}_i-\boldsymbol{\ell}_{i-1}}{\delta_{i-1}}\right)^2e^{-\mathbf{g}^i_{-}}\underline{\delta}^\alpha+\mathcal{R}_i(k,\underline{\delta}).
\end{align*}
for $\mathcal{R}_i(k,{\underline{\delta}})$ satisfying a bound of the form
\[|\mathcal{R}_i(k,{\underline{\delta}})|\leq K\underline{\delta}^{2\alpha\wedge 1}.\]
for all $i$ and $k$ where $K>0$ depends only on $\beta$ and $M$. Furthermore, by our constraints the first term satisfies
\[0\geq\frac{\frac{\boldsymbol{\ell}_{i+1}-\boldsymbol{\ell}_i}{\delta_{i}}-\frac{\boldsymbol{\ell}_{i}-\boldsymbol{\ell}_{i-1}}{\delta_{i-1}}}{\frac{\delta_{i}+\delta_{i-1}}{2}}\geq-\beta.\]
Most importantly, the leading order correction is positive. Take $\delta_{1,\alpha}<\delta_{0,\alpha}$ such that if $\underline{\delta}<\delta_{1,\alpha}$ then \[K\underline{\delta}^{2\alpha\wedge 1}<\frac{\beta}{4}\underline{\delta}^\alpha \ \ \ \mathrm{and} \ \ \ e^{\frac{\sqrt{\beta}}{2}}\underline{\delta}^\alpha<1.\]
Note here that $2\alpha\wedge 1>\alpha$ so such a $\delta_{1,\alpha}$ can indeed be found. For $\underline{\delta}<\delta_{1,\alpha}$
\begin{align*}
    (\log(s_{\mathcal{P}}(x)+\underline{\delta}^\alpha))''&\geq \frac{\frac{\boldsymbol{\ell}_{i+1}-\boldsymbol{\ell}_i}{\delta_{i}}-\frac{\boldsymbol{\ell}_{i}-\boldsymbol{\ell}_{i-1}}{\delta_{i-1}}}{\frac{\delta_{i}+\delta_{i-1}}{2}}\left(1-e^{-\mathbf{g}^i_{-}}\underline{\delta}^\alpha\right)+\mathcal{R}_i(k,\underline{\delta})\\
    &\geq -\beta\left(1-e^{-\mathbf{g}^i_{-}}\underline{\delta}^\alpha\right)+\mathcal{R}_i(k,\underline{\delta})\\
    &\geq -\beta\left(1-\frac{1}{2}\underline{\delta}^\alpha\right)-\frac{\beta}{4}\underline{\delta}^\alpha\\
    &=-\beta\left(1-\frac{1}{4}\underline{\delta}^\alpha\right)\geq -\beta.
\end{align*}
We have used here that $(1-e^{-\mathbf{g}^i_{-}}\underline{\delta}^\alpha)\in\left[1-e^{\frac{\sqrt{\beta}}{2}}\underline{\delta}^\alpha,1-\frac{1}{2}\underline{\delta}^\alpha\right]$ which implies it is less than 1 and strictly positive for $\underline{\delta}<\delta_{1,\alpha}$. Since the point $x$ and interval $\mathring{I}_i$ were arbitrary we have that the above holds in general across all interior intervals $\mathring{I}_i$ when $\underline{\delta}<\delta_{1,\alpha}$. We close by remarking that the proof of the estimate for $J_i$ and the end intervals uses the constraint \eqref{eqn:constraint.endpoint} in place of the discrete $\beta$-smooth constraint above, but otherwise proceeds almost identically.
\end{proof}

Finally, we are ready to tie these results together and prove the main lemma.

\begin{proof}[Proof of Proposition \ref{prop:construction}]
It is clear that $\ell_{\alpha,\mathcal{P}}$ is exponentially concave and $C^1$ since its exponential is a positive multiple of the concave and $C^1$ function $s_{\mathcal{P}}(x)+{\underline{\delta}}^\alpha$. Moreover by the choice of $a_{\mathcal{P}}$, $\ell_{\alpha,\mathcal{P}}(1/2)=0$ as required. Lemma \ref{lem:approx.of.l.by.salpha} gives us the approximation error 
\begin{align*}
    \left|\ell_{\alpha,\mathcal{P}}'(x)-\boldsymbol{\hat{\ell}}_{\mathcal{P}}'(x)\right|&=\left|(\log(s_{\mathcal{P}}(x)+\underline{\delta}^\alpha))'-\boldsymbol{\hat{\ell}}_{\mathcal{P}}'(x)\right|\leq K{\underline{\delta}}^\alpha,
\end{align*}
where $K>0$ depends only on $\beta$ and $M$.
By Lemma \ref{lem:interior.estimate}, if we take $\delta_{0,\alpha}=\min\{\delta_{1,\alpha},\delta_{2,\alpha}\}$, then for $\underline{\delta}<\delta_{0,\alpha}$ we have $0\geq \ell_{\alpha,\mathcal{P}}''\geq -\beta$ for all points where $\ell_{\alpha,\mathcal{P}}$ is twice differentiable. These points are exactly the interior of the intervals $I_i$, $J_i$ and the interior of the end intervals. By noting that the left (right) limits of $\ell_{\alpha,\mathcal{P}}''(x)$ exist everywhere and so coincide with the left (right) second derivatives at the boundary points of these intervals (see Lemmas \ref{lem:extension.pt1} and \ref{lem:extension.pt2}), we can extend the above bound to the one sided second derivatives on $[0,1]$. Then, by appealing to Lemma \ref{lem:beta.smooth.twicediff.a.e} we conclude that $\ell_{\alpha,\mathcal{P}}$ is indeed $\beta$-smooth. Taking all of the above together we have that for $\underline{\delta}<\delta_{0,\alpha}$, $\ell_{\alpha,\mathcal{P}}\in\mathcal{E}_\beta$.
\end{proof}

\subsection{Lemma \ref{lem:approx.error.J}}
\begin{proof}[Proof of Lemma \ref{lem:approx.error.J}]
Let $\boldsymbol{\pi}_i(\mathbf{p})$ and $\boldsymbol{\pi}^\delta_i(\mathbf{p})$ be the portfolio mappings induced by $\ell$ and $\ell_\delta$, respectively. We have that
\begin{align*}
    \left|\frac{\boldsymbol{\pi}_i(\mathbf{p})}{p_i}-\frac{\boldsymbol{\pi}^\delta_i(\mathbf{p})}{p_i}\right|&\leq\frac{1}{n}\left|\ell'(p_i)-\ell_\delta'(p_i)\right|+\frac{1}{n}\sum_{j=1}^np_j\left|\ell'(p_j)-\ell_\delta'(p_j)\right|\\
    &\leq\frac{2}{n}\sup_{x\in[0,1]}\left|\ell'(x)-\ell_\delta'(x)\right|\leq\frac{2K_0}{n}\delta^\alpha.
\end{align*}
Next using Lemma \ref{lem:beta.interpretation} and the above approximation error
\[e^{-\frac{2\sqrt{\beta}}{n}}-\frac{2K_0}{n}\delta^\alpha\leq\exp\{\mathcal{L}(\mathbf{u},\mathbf{r};\ell_\delta)\}=\sum_{i=1}^n\frac{\boldsymbol{\pi}^\delta_i(\mathbf{u})}{u_i}\frac{u_ir_i}{\mathbf{u}\cdot \mathbf{r}}\leq 1+\frac{\beta}{n^2}+\frac{2K_0}{n}\delta^\alpha.\]
In light of this we can find a $\delta_{0,\alpha}>0$ such that if $\delta<\delta_{0,\alpha}$ we have $\frac{2K_0}{n}e^{\frac{2\sqrt{\beta}}{n}}\delta^\alpha<\frac{1}{2}$. For such values of $\delta$ the quantities $\exp\{\mathcal{L}(\mathbf{u},\mathbf{r};\ell_\delta)\}$ and $\exp\{\mathcal{L}(\mathbf{u},\mathbf{r};\ell)\}$ always lie in the interval $\left[ \frac{1}{2}  e^{-\frac{2\sqrt{\beta}}{n}},1+\frac{\beta}{n^2}+ \frac{1}{2} e^{-\frac{2\sqrt{\beta}}{n}}\right]$. Since the map $x \mapsto \log (x)$ is Lipschitz on this interval with constant $2e^{\frac{2\sqrt{\beta}}{n}}$, we proceed as in the proof of Lemma \ref{lem:log.value.estimate} to get
\begin{equation*}
|\mathcal{L}(\mathbf{u}, \mathbf{r} ; \ell) - \mathcal{L}(\mathbf{u}, \mathbf{r} ; \ell_\delta)|\leq4K_0 e^{\frac{2\sqrt{\beta}}{n}}\delta^{\alpha}.
\end{equation*} 
On the other hand, since $\mathbf{D}_\cdot[\mathbf{u}\oplus\mathbf{r}:\mathbf{r}]$ is linear we have 
\begin{equation*}
|\mathbf{D}_\ell[\mathbf{u}\oplus\mathbf{r}:\mathbf{r}] - \mathbf{D}_{\ell_\delta}[\mathbf{u}\oplus\mathbf{r}:\mathbf{r}]|\leq 2K_0\delta^{\alpha}.
\end{equation*} 
The claim then follows by integration.
\end{proof}

\section*{Acknowledgment}
This work is partially supported by NSERC Grant RGPIN-2019-04419 and an NSERC Alexander Graham Bell Canada Graduate Scholarship (Application No. CGSD3-535625-2019). We thank Johannes Ruf and Kangjianan Xie for sharing the dataset of \cite{RX19} which was used to test ideas in an earlier stage of this project. We also thank Christa Cuchiero, Janka M\"oller and Soumik Pal for their helpful comments on an earlier draft of this paper. Finally, we thank the associate editor and anonymous reviewers whose comments greatly improved the paper, especially the empirical section.

\bibliographystyle{abbrv}
\bibliography{geometry.ref}

\end{document}